\documentclass[superscriptaddress,aps,pra,nofootinbib,notitlepage,10pt,longbibliography]{revtex4-1}
\usepackage{graphicx}
\usepackage{amsmath}
\usepackage{amssymb}
\usepackage{amsthm}
\usepackage{comment}
\usepackage{placeins}
\usepackage[caption=false]{subfig}
\usepackage[colorlinks]{hyperref}
\usepackage[all]{hypcap}
\usepackage{tikz}
\usepackage{verbatim}
\usepackage{subfig}
\usetikzlibrary{arrows}
\usepackage{units}
\usepackage{soul}

\newtheorem{theorem}{Theorem}
\newtheorem{lemma}[theorem]{Lemma}

\newtheorem{corollary}[theorem]{Corollary}

\newcommand{\eq}[1]{Eq.~\hyperref[eq:#1]{(\ref*{eq:#1})}}
\renewcommand{\sec}[1]{\hyperref[sec:#1]{Section~\ref*{sec:#1}}}
\DeclareRobustCommand{\app}[1]{\hyperref[app:#1]{Appendix~\ref*{app:#1}}}
\newcommand{\tab}[1]{\hyperref[tab:#1]{Table~\ref*{tab:#1}}}
\newcommand{\fig}[1]{\hyperref[fig:#1]{Figure~\ref*{fig:#1}}}
\newcommand{\figa}[2]{\hyperref[fig:#1]{Figure~\ref*{fig:#1}#2}}
\newcommand{\figx}[2]{\hyperref[fig:#1]{Figure~\ref*{fig:#1}(#2)}}
\newcommand{\thm}[1]{\hyperref[thm:#1]{Theorem~\ref*{thm:#1}}}
\newcommand{\lem}[1]{\hyperref[lem:#1]{Lemma~\ref*{lem:#1}}}
\newcommand{\cor}[1]{\hyperref[cor:#1]{Corollary~\ref*{cor:#1}}}
\newcommand{\defn}[1]{\hyperref[def:#1]{Definition~\ref*{def:#1}}}
\newcommand{\alg}[1]{\hyperref[alg:#1]{Algorithm~\ref*{alg:#1}}}

\def\avg#1{\mathinner{\langle{#1}\rangle}}
\def\bra#1{\mathinner{\langle{#1}|}}
\def\ket#1{\mathinner{|{#1}\rangle}}
\newcommand{\braket}[2]{\langle #1|#2\rangle}
\newcommand{\proj}[1]{\ket{#1}\!\!\bra{#1}}

\newcommand*{\bunderbrace}[1]{%
  \mathop{%
    \mathchoice
    {\underbrace{\displaystyle#1}}%
    {\underbrace{\textstyle#1}}%
    {\underbrace{\scriptstyle#1}}%
    {\underbrace{\scriptscriptstyle#1}}%
  }\limits
}

\input{Qcircuit}

\begin{document}

\title{Low Depth Quantum Simulation of Electronic Structure}

\date{\today}
\author{Ryan Babbush}
\email[Corresponding author: ]{ryanbabbush@gmail.com}
\affiliation{Google Inc., Venice, CA 90291}
\author{Nathan Wiebe}
\affiliation{Microsoft Research, Redmond, WA 98052}
\author{Jarrod McClean}
\affiliation{Computational Research Division, Lawrence Berkeley National Laboratory, Berkeley, CA 94720}
\author{James McClain}
\affiliation{Division of Chemistry and Chemical Engineering, California Institute of Technology, Pasadena, CA 91125}
\author{Hartmut Neven}
\affiliation{Google Inc., Venice, CA 90291}
\author{Garnet Kin-Lic Chan}
\email[Corresponding author: ]{gkc1000@gmail.com}
\affiliation{Division of Chemistry and Chemical Engineering, California Institute of Technology, Pasadena, CA 91125}

\begin{abstract}
Quantum simulation of the electronic structure problem is one of the most researched applications of quantum computing. The majority of quantum algorithms for this problem encode the wavefunction using $N$ Gaussian orbitals, leading to Hamiltonians with ${\cal O}(N^4)$ second-quantized terms. We avoid this overhead and extend methods to the condensed phase by utilizing a dual form of the plane wave basis which diagonalizes the potential operator, leading to a Hamiltonian representation with ${\cal O}(N^2)$ second-quantized terms. Using this representation we can implement single Trotter steps of the Hamiltonians with linear gate depth on a planar lattice. Properties of the basis allow us to deploy Trotter and Taylor series based simulations with respective circuit depths of ${\cal O}(N^{7/2})$ and $\widetilde{\cal O}(N^{8/3})$ for fixed charge densities -- both are large asymptotic improvements over all prior results. Variational algorithms also require significantly fewer measurements to find the mean energy in this basis, ameliorating a primary challenge of that approach. We conclude with a proposal to simulate the uniform electron gas (jellium) using a low depth variational ansatz realizable on near-term quantum devices. From these results we identify simulations of low density jellium as a promising first setting to explore quantum supremacy in electronic structure.
\end{abstract}
\maketitle

\section*{Introduction}

The problem of electronic structure is to simulate the stationary properties of electrons interacting via Coulomb forces in an external potential.  The solution of this problem has wide implications for all areas of chemistry, condensed matter physics, and materials science, and is of industrial relevance in the design and engineering of new pharmaceuticals, catalysts, and materials. Recently, quantum computers have emerged as promising tools for tackling this challenge, offering the potential to access difficult electronic structure with reduced computational complexity. However as the age of ``quantum supremacy'' dawns, so has the realization that many ``efficient'' quantum algorithms still require more resources than will be available in the near-term.

Originally proposed by Feynman~\cite{Feynman1982}, the efficient simulation of quantum systems by other, more controllable quantum systems formed the basis for modern constructions of quantum computation.  This early insight has since been refined to encompass more universal and versatile constructions of simulation~\cite{Lloyd1996,Abrams1997}. By combining quantum phase estimation~\cite{Kitaev1995} with these techniques, Aspuru-Guzik et al. showed the first efficient quantum algorithm for solving quantum chemistry problems~\cite{Aspuru-Guzik2005}. This initial algorithm was based on adiabatic state preparation combined with Trotter-Suzuki decomposition of the unitary time-evolution operator \cite{Trotter1959,Suzuki1993} in second quantization.

Many algorithmic and theoretical advances have followed since the initial work in this area. The quantum simulation of electronic structure has been proposed via an adiabatic algorithm \cite{BabbushAQChem}, via Taylor series time-evolution \cite{BabbushSparse1}, in second quantization, in real space \cite{Kassal2008,Kivlichan2016}, in the configuration interaction representation \cite{Toloui2013,BabbushSparse2}, and using a quantum variational algorithm \cite{Peruzzo2013,McClean2015}. Starting with \cite{Whitfield2010}, researchers have sought to map these algorithms to practical circuits and reduce the overhead required for implementation by both algorithmic enhancements~\cite{Wecker2014,Poulin2014,Hastings2015,Romero2017} as well as physical considerations~\cite{BabbushTrotter,McClean2014}. As a second quantized formulation is generally regarded as most practical for near-term devices, many works have also tried to find more efficient ways of mapping fermionic operators to qubits \cite{Seeley2012,Tranter2015,Whitfield2016,Bravyi2017,Havlicek2017}.

With recent developments in quantum computing hardware~\cite{Corcoles2015,Riste2015,Kelly2015,Barends2016,Roushan2017}, there is an additional drive to identify early practical problems on which these devices might demonstrate an advantage \cite{Mohseni2017,Boixo2016}. The challenge of using such devices in the near-term is that limited coherence requirements necessitates algorithms with extremely low circuit depth. Toy demonstrations of quantum chemistry algorithms have been performed on architectures ranging from quantum photonics and ion traps to superconducting qubits~\cite{Lanyon2010,Li2011,Wang2014,Peruzzo2013,Shen2015,OMalley2016,Kandala2017}. In particular, the variational quantum algorithm \cite{Peruzzo2013,McClean2015} has been shown experimentally to be inherently robust to certain errors~\cite{OMalley2016}, and is considered to be a promising candidate for performing practical quantum computations in the near-term~\cite{Wecker2015a,Mueck2015}.

A major challenge in developing low depth quantum algorithms for quantum chemistry is that electronic structure Hamiltonians often have as many as ${\cal O}(N^4)$ terms, where $N$ is the number of basis functions. This is problematic as many algorithms for time-evolution and energy estimation have costs which scale explicitly with the number of terms. In this paper, we introduce new basis functions which reduce the number of Hamiltonian terms to $\Theta(N^2)$. We exploit this and other properties of the basis to demonstrate algorithms for state preparation and time-evolution which are simultaneously practical at small sizes and asymptotically more efficient than any in the literature. We conclude with a proposal to simulate the uniform electron gas on a near-term device using planar circuits of only linear depth.

\begin{table*}[t]
\label{tab:scalings}
\begin{tabular}{c|c|c|c|c|c||c}
Year & Reference & Representation & Algorithm & Primitive Depth & Repetitions & Total Depth\\
\hline\hline
2005  & Aspuru-Guzik et al. \cite{Aspuru-Guzik2005}  & JW Gaussians & Trotter & ${\cal O}(\textrm{poly}(N))$ & ${\cal O}(\textrm{poly}(N))$ & ${\cal O}(\textrm{poly}(N))$\\
2008  & Kassal et al. \cite{Kassal2008}  & Real Space & Trotter & ${\cal O}(\textrm{poly}(N))$ & ${\cal O}(\textrm{poly}(N))$ & ${\cal O}(\textrm{poly}(N))$\\
2010 & Whitfield et al. \cite{Whitfield2010} & JW Gaussians & Trotter & $\Theta(N^5)$ & ${\cal O}(\textrm{poly}(N))$ & ${\cal O}(\textrm{poly}(N))$ \\
2012 & Seeley et al. \cite{Seeley2012} & BK Gaussians & Trotter & $\widetilde{\Theta}(N^4)$ & ${\cal O}(\textrm{poly}(N))$ & ${\cal O}(\textrm{poly}(N))$ \\
2013 & Perruzzo et al. \cite{Peruzzo2013} & JW Gaussians & UCC & $\Theta(N^5)$ & Variational & $\Omega(N^5)$ \\
2013 & Toloui et al. \cite{Toloui2013} & CI Gaussians & Trotter & $\widetilde{\Theta}(\eta^2 N^2)$ & ${\cal O}(\textrm{poly}(N))$ & ${\cal O}(\textrm{poly}(N))$\\
2013 & Wecker et al. \cite{Wecker2014} & JW Gaussians & Trotter & $\Theta(N^5)$ & ${\cal O}(N^{5})$ & ${\cal O}(N^{10})$ \\
2014 & Hastings et al. \cite{Hastings2015} & JW Gaussians & Trotter & $\Theta(N^4)$ & ${\cal O}(N^4)$ & ${\cal O}(N^8)$\\
2014 & Poulin et al. \cite{Poulin2014} & JW Gaussians & Trotter & $\Theta(N^4)$ & $ {\cal O}(\sim N^2)$ & ${\cal O}(\sim N^6)$\\
2014 & McClean et al. \cite{McClean2014} & BK Gaussians & Trotter & $\widetilde{\cal O}(N^2)$ & ${\cal O}(N^4)$ & $\widetilde{\cal O}(N^6)$\\
2014 & Babbush et al. \cite{BabbushTrotter} & JW Gaussians & Trotter  & $\Theta(N^4)$ & ${\cal O}(\sim N)$ & ${\cal O}(\sim N^5)$ \\
2015 & Babbush et al. \cite{BabbushSparse1} & JW Gaussians & Taylor & $\widetilde{\Theta}(N)$ & $\widetilde{\cal O}(N^4)$ & $\widetilde{\cal O}(N^5)$  \\
2015 & Babbush et al. \cite{BabbushSparse2} & CI Gaussians & Taylor & $\widetilde{\Theta}(N)$ & $\widetilde{\cal O}(\eta^2 N^2)$ & $\widetilde{\cal O}(\eta^2 N^3)$ \\
%
%
2015 & Wecker et al. \cite{Wecker2015a} & JW Gaussians & TASP & $\Theta(N^4)$ & Variational & $\Omega(N^{4})$ \\
2016 & McClean et al. \cite{McClean2015} & BK Gaussians & UCC & $\widetilde{\Theta}(\eta^2 N^2)$ & Variational & $\widetilde{\Omega}(\eta^2 N^2)$ \\
2016  & Kivlichan et al. \cite{Kivlichan2016}  & Real Space & Taylor & ${\cal O}(\textrm{poly}(N))$ & $\widetilde{\cal O}(\eta^2)$ & ${\cal O}(\textrm{poly}(N))$\\
2017 & This paper & JW Plane Waves & Trotter  & $\Theta(N)$ & ${\cal O}(\eta^{1.83} N^{0.67})$ & ${\cal O}(\eta^{1.83} N^{1.67})$ \\
2017 & This paper & JW Plane Waves & Taylor & $\widetilde{\Theta}(1)$ & $\widetilde{\cal O}(N^{2.67})$ & $\widetilde{\cal O}(N^{2.67})$ \\
2017 & This paper & JW Plane Waves & TASP & $\Theta(N)$ & Variational & $\Omega(N)$ \\
\hline
\end{tabular}
\caption{The lowest circuit depth algorithms for the quantum simulation of electronic structure\footnote{Throughout this paper we use the computer science conventions that $f \in \Theta(g)$ for any functions $f$ and $g$ if $f$ is asymptotically upper and lower bounded by a multiple of $g$. $f \in o(g)$ implies that $f / g \rightarrow 0$ in the asymptotic limit. ${\cal O}$ indicates an asymptotic upper bound and $\Omega$ indicates an asymptotic lower bound. A tilde on top of the bound notation, e.g. $\widetilde{\cal O}(N)$, indicates suppression of polylogarithmic factors. In contrast to formally rigorous bounds, a tilde inside of a bound, e.g. ${\cal O}(~N)$ indicates the bound is obtained empirically.}. Reduction in primitive depth is typically the result of improved algorithms whereas reduction in required repetitions is typically the result of tighter bounds. Bounds on the primitive depth indicate the scaling of that particular implementation (which is why $\Theta$ is often used). As variational algorithms are heuristic, the total depth is listed as a lower bound. $N$ is number of orbitals and $\eta < N$ is number of particles. Second-quantized fermionic encodings including JW (Jordan-Wigner) \cite{Jordan1928} and BK (Bravyi-Kitaev) \cite{Bravyi2002} have ${\cal O}(N)$ spatial complexity whereas first-quantized encodings including CI (Configuration Interaction) \cite{Toloui2013} and Real Space \cite{Kassal2008} have ${\cal O}(\eta \log N)$ spatial complexity. Variational quantum algorithms are abbreviated as UCC (Unitary Coupled Cluster) \cite{Peruzzo2013} and TASP (Trotterized Adiabatic State Preparation) \cite{Wecker2015a}. Unlike other approaches, the Trotter and variational algorithms of this paper require no additional overhead when restricting qubit connectivity to a planar lattice. Though asymptotically equivalent to at least second order in perturbation theory, as discussed in \app{basis_errors}, one sometimes requires a constant factor more plane waves than Gaussians orbitals to achieve the same precision for single-molecule calculations.}
\end{table*}

\sec{sec_one} discusses several strategies for quadratically reducing the number of terms in the second quantized electronic structure Hamiltonian. The approach we focus on is to use a plane wave basis and its dual obtained by a unitary rotation, which we call the ``plane wave dual basis''. In \sec{pwd_basis} we show that the dual basis diagonalizes the potential  operators, leading to a Hamiltonian with $\Theta(N^2)$ terms and other desirable properties. The plane wave basis is especially natural when treating periodic systems (e.g. crystalline solids), allowing us to conveniently extend quantum simulation methods to condensed phase systems of interacting electrons. The basis is compact for uniform and near-uniform electron gasses, realized in simple metals as well as electrons in semiconductor wells, and there is well-developed infrastructure (e.g. pseudopotentials) to enable compact representations of other realistic materials. In \sec{fqft} we describe a generalization of the fast Fourier transform to second quantized systems of fermions. We show that this operation can be implemented on a planar lattice of qubits with linear depth and that it maps a quantum state between the plane wave basis (where the kinetic operator is diagonal) and the plane wave dual basis (where the  potential operator is diagonal).

\sec{sec_two} introduces algorithmic improvements for three different quantum approaches to simulating electronic structure. The scaling advantages of the techniques introduced in this paper are compared to prior results in \tab{scalings}. In \sec{trotter_alg} we take advantage of the fermionic fast Fourier transform to show that single Trotter steps of the Hamiltonian can be implemented using circuits of only ${\cal O}(N)$ gate depth on a planar lattice without ancillae. This is a large improvement over the previous best result of ${\cal O}(N^{9/2})$ depth. We bound the gate depth of Trotterization on a planar lattice within this representation at ${\cal O}(N^{7/2} / \epsilon^{1/2})$ where $\epsilon$ is the target precision. This represents more than a quadratic improvement over the best previously proven bounds for Trotterization. In \sec{taylor_alg}, we show that the Taylor series method of time-evolution has gate depth $\widetilde{\cal O}(N^{8/3})$ with logarithmic dependence on $\epsilon$. In \sec{measurement} we discuss how the structure of the plane wave dual basis reduces the measurements required when estimating the energy through Hamiltonian averaging, significant in the context of variational quantum algorithms. We show that ${\cal O}(N^{4} / \epsilon^2)$ circuit repetitions are sufficient to estimate the energy of the Hamiltonian to absolute error $\epsilon$. However, we also argue that when one desires to study the properties of a material in its thermodynamic limit, relative error $\mu$ is a more relevant metric and in that context, only  ${\cal O}(N^{2} / \mu^2)$ repetitions are required. Even for fixed absolute error, our bounds on the required measurements represent a substantial improvement over prior estimates.

\sec{sec_three} proposes an experiment for simulating the uniform electron gas (also known as jellium) on a near-term quantum device based on the techniques of \sec{sec_one} and \sec{sec_two}. Though one of the simplest models of realistic electronic structure, jellium is readily tuned between simple and complex physics through a single parameter, the electron density. Jellium has foundational importance both for practical computational materials science, as well as basic condensed matter physics:
the energy density of jellium is the starting point for all practical density functional approximations used in 
quantum chemistry and materials simulations; the system can be physically realized to good approximation in real materials such as the alkali metals,
and in semiconductor wells; and two dimensional jellium in a magnetic field is the standard setting in which to discuss the fractional quantum Hall effect. In addition, many questions remain about jellium physics in the low  density regime, where
unbiased classical simulations are intractable for system sizes of interest, and biased simulations do not reach the
accuracy to definitively resolve between competing phases.
In \sec{jellium_alg} we describe a quantum variational algorithm for jellium which can be executed on a planar lattice of qubits with ${\cal O}(N)$ circuit depth, based on the results of \sec{sec_one}. We conclude with an outlook on how to extend these simulations to more general quantum chemical problems and the potential for the jellium problem to serve as a setting for early demonstrations of quantum supremacy over a problem of practical interest.

This paper provides a number of supporting technical results in appendices. In \app{finite_diff} we show a finite-difference discretization of the Hamiltonian with ${\cal O}(N^2)$ terms. In \app{plane_waves} we review the well known form of the Hamiltonian in the plane wave basis and in \app{plane_wave_dual} we derive its representation in the plane wave dual basis, including connections
to discrete variable representations. \app{qubit_ham} shows a closed-form representation of the plane wave dual Hamiltonian mapped to qubits under the Jordan-Wigner transformation. In \app{basis_errors} we discuss the discretization errors associated with Gaussian molecular orbitals and plane wave orbitals and argue that both bases have the same asymptotic error scaling. In \app{operators} we provide bounds on components of the plane wave dual Hamiltonian that are relevant to the results of \sec{sec_two}. In \app{TSerror} we bound the Trotter error in the simulations of \sec{trotter_alg}. In \app{qubit_cycle} we show a method for implementing controlled phase operations between all qubits on a planar lattice with gate depth of only ${\cal O}(N)$. In \app{fqft} we prove results about the scaling of the fermionic fast Fourier transform. In \app{alt_trotter} we provide new circuits for evolving under a sum of commuting Pauli strings and use that result to bound the cost of Trotterizing Hamiltonians in the plane wave dual basis without using the fermionic fast Fourier transform. Finally, in \app{onthefly} we show an alternative implementation of the Taylor series algorithm which improves over the simpler scheme explored in \sec{taylor_alg}.

\section{Electronic Structure Hamiltonians with Fewer Terms}
\label{sec:sec_one}

Within the Born-Oppenheimer approximation, the properties of materials, molecules and atoms emerge from the behavior of electrons interacting in the external potential of positively-charged nuclei. In the non-relativistic case, the dynamics of these electrons are governed by the Coulomb Hamiltonian,
\begin{align}
H = \underbrace{-\sum_i \frac{\nabla^2_i}{2}}_{T} \underbrace{- \sum_{i, j}\frac{\zeta_j}{|R_j - r_i|}}_{U} \underbrace{+ \sum_{i < j} \frac{1}{|r_i - r_j|}}_{V} \underbrace{+ \sum_{i < j}\frac{\zeta_i \zeta_j}{|R_i - R_j|}}_{\textrm{constant}}
\end{align}
where we have used atomic units, $r_i$ represent the positions of electrons, $R_i$ represent the positions of nuclei, and $\zeta_i$ are the charges of nuclei. $T$ is referred to as the kinetic term, $U$ the (nuclear) potential term, and $V$ the electron-electron repulsion potential term. The electronic structure problem is to estimate the properties of the eigenfunctions (especially the lowest energy eigenfunction) of the time-independent Schroedinger equation defined by this Hamiltonian.

To convert the differential equation into a practical computational problem, one typically first chooses some form of discretization.  Moreover, the antisymmetry of electrons must be enforced either in the solutions (first quantization) or in the operators (second quantization)\footnote{Several papers have investigated quantum simulation of electronic structure in first quantization \cite{Kassal2008,Toloui2013,BabbushSparse2,Kivlichan2016}. When scaling to the continuum basis limit (as opposed to scaling towards larger systems), these encodings have an asymptotic spatial advantage; first quantization requires ${\cal O}(\eta \log N)$ qubits whereas second quantization requires ${\cal O}(N)$ qubits. In first quantization one must initialize the simulation in an explicitly antisymmetric initial state, which is potentially costly. As discussed in \cite{Kivlichan2016}, bounded-error quantum simulations in a real space basis may (in the worst case) require that one compute the potential to a number of bits that is exponential in $N$ as a consequence of singularities in the Hamiltonian which occur when electrons occupy the same location in space. But the primary reason why these algorithms remain less popular than their second-quantized counterparts is that all proposed implementations \cite{Kassal2008,Toloui2013,BabbushSparse2,Kivlichan2016} require complex on-the-fly logic which preclude near-term implementation and dramatically increase the number of T gates required for error-correction.}. Most quantum computing research focuses on second quantization, in which the Hamiltonian is formulated as
\begin{equation}
\label{eq:n4}
H = \underbrace{\sum_{p, q} h_{pq} \, a^\dagger_p a_q}_{T + U} + \underbrace{\frac{1}{2} \sum_{p, q, r, s} h_{pqrs} \, a^\dagger_p a^\dagger_q a_r a_s}_{V}
\end{equation}
where $a^\dagger_p$ and $a_p$ are fermionic raising and lowering operators satisfying the anticommutation relation $\{a^\dagger_p, a_q\} = \delta_{pq}$, the coefficients $h_{pq}$ and $h_{pqrs}$ are determined by the discretization that has been chosen, and the sums now run over the number of discretization elements for a single particle. Specifically, if electron $j$ is represented in a space of spin-orbitals $\{\phi_p(r_j)\}$ then $a^\dagger_p$ and $a_p$ are related to Slater determinants through the equivalence,
\begin{align}
&  \bra{r_0,\ldots,r_{\eta-1}} a^\dagger_{p_0} \cdots a^\dagger_{p_{\eta-1}} \ket{0} = \sqrt{\frac{1}{\eta!}}
\begin{vmatrix}
\phi_{p_0}\left(r_0\right) & \phi_{p_1}\left( r_0\right) & \cdots & \phi_{p_{\eta-1}} \left( r_0\right) \\
\phi_{p_0}\left(r_1\right) & \phi_{p_1}\left( r_1\right) & \cdots & \phi_{p_{\eta-1}} \left( r_1\right) \\
\vdots & \vdots & \ddots & \vdots\\
\phi_{p_0}\left(r_{\eta-1}\right) & \phi_{p_1}\left(r_{\eta-1}\right) & \cdots & \phi_{p_{\eta-1}} \left(r_{\eta-1}\right) \end{vmatrix}
\end{align}
where $\eta$ is the number of electrons in the system and $\ket{0}$ is the vacuum. From inspection, one sees that the number of terms in \eq{n4} may be as high as ${\cal O}(N^4)$ where $N$ is the size of the discrete representation. This presents a major problem for realizing quantum simulation algorithms on near-term quantum devices as most quantum algorithms have some explicit dependence on the number of terms. For instance, the cost of implementing a Trotter step requires a number of gates that scales at least linearly in the number of terms. Likewise, the number of measurements required for variational quantum algorithms scales at least linearly in the number of terms. 

The most commonly used discretization in classical electronic structure is known as a Galerkin discretization. The Galerkin discretization is derived from the weak formulation of the Schroedinger equation in Hilbert space, given by finding $\ket{\phi}$ (spanned by the basis vectors $\{\ket{\phi_p}\}$) such that $\bra{\phi_p}H\ket{\phi} = E \braket{\phi_p}{\phi}$ for all $p$. This is contrasted with the strong formulation (see \app{finite_diff}) that insists the original differential equation hold at all points in space $r$, as opposed to assessing error on the restricted subspace spanned by $\{\ket{\phi_p}\}$. The Galerkin formulation leads to the following coefficients which define the second quantized Hamiltonian of \eq{n4}:
\begin{align}
\label{eq:one_body_ints}
h_{pq} &= \left\langle \phi_p \middle | \left(-\frac{\nabla^2}{2} + U\right) \middle | \phi_q \right\rangle = \int dr \, \phi_p^* \left(r\right) \left(-\frac{\nabla^2}{2} + U\left(r\right)\right) \phi_q \left(r\right) \\
h_{pqrs} &= \left\langle \phi_p \middle |  \left \langle \phi_q \middle | V \middle | \phi_r\right\rangle \middle | \phi_s \right\rangle = \int dr \, dr' \,  \phi_p^*\left(r\right) \phi_q^* \left(r'\right) V\left(r, r'\right) \phi_r \left(r'\right)  \phi_s \left(r\right)
\label{eq:two_body_ints}
\end{align}
where $U(r)$ is the external potential Coulomb interaction, $V(r,r')$ is the two-electron Coulomb interaction and the $\phi_p(r)=\braket{r}{\phi_p}$ are the single-particle orbitals that define the basis. An important feature of Galerkin discretizations (again, in contrast to e.g. finite-difference discretizations) is that basis set error is variational, meaning that energies from exact diagonalization monotonically approach the continuum basis set limit from above.
 
The basis functions $\phi_p(r)$ are chosen in a number of ways.  Perhaps the most common choice for treating molecular systems is  atom-centered Gaussian basis functions, conventionally termed an atomic orbital basis. These functions resemble the mean-field orbitals of single atoms and provide a computationally convenient formulation for the evaluation of the above integrals. Parameters of the Gaussians are optimized so that modest numbers of such basis functions can compactly represent the low-energy eigenstates of atomic and molecular Hamiltonians with qualitative accuracy. However, a drawback of these functions is that the associated Hamiltonians contain ${\cal O}(N^4)$ terms for modest size systems, despite their relative locality eventually leading to ${\cal O}(N^2)$ terms in an asymptotic limit \cite{McClean2014}. Moreover, to prepare a compact initial state for a molecular simulation, it is common to rotate from the atomic orbital basis to the molecular orbital basis, which minimizes the mean-field molecular energy. This basis is even more delocalized than the atomic orbital basis and contains even more terms at all system sizes.

Gaussian bases were introduced more than half a century ago to reduce the cost of evaluating the integrals in \eq{one_body_ints} and \eq{two_body_ints} for the mean-field quantum chemistry calculations of interest at the time~\cite{boys1950electronic}. However, with advances in classical computing power, the evaluation of such integrals for systems with up to several hundred atoms is no longer a major bottleneck. Further, the requirements of a basis for efficient quantum algorithms are quite different than for classical algorithms. In a quantum algorithm, we primarily desire the computational basis (i) to have a small number of terms, so as to minimize the cost of basic algorithms such as time evolution, or the number of measurements in variational quantum algorithms, and (ii) to allow for a simple preparation of a relevant initial quantum state.
To some extent these are conflicting requirements, as (i) can be obtained by locality of the basis in real space, while (ii) implies locality of states in energy space, or delocalization in real space. For example, the traditional Gaussian basis satisfies (ii) but not (i) in medium sized molecules. We should note that while the number of basis functions is also an important quantum resource (corresponding to the number of logical qubits), in many cost models the circuit size is more important. In a fault-tolerant architecture, the number of physical qubits required is largely a function of the number of non-Clifford gates in the original algorithm and does not strongly depend on the number of logical qubits. While existing quantum hardware is limited to a small number of qubits, the expectation is that manufacturing more qubits will be easier in the near-future than significantly increasing coherence time, suggesting that even in a non-fault-tolerant context, gate depth is a more important resource than number of qubits.

To consider how one might circumvent the ${\cal O}(N^4)$ scaling of terms in the Hamiltonian, consider a set of spatially disjoint functions $\{\phi_p(r)\}$, which are defined such that the intersection of the supports of $\phi_p(r)$ and $\phi_q(r)$ is the empty set for all $p \neq q$.  The consequence of this is that the product $\phi_p(r) \phi_q(r)=0$ for all $r$ and all $p \neq q$.  Taking this definition with \eq{two_body_ints}, it is clear that $h_{pqrs}=0$ unless $p=s$ and $r=q$; thus, there are at most ${\cal O}(N^2)$ elements defining the Hamiltonian. To enable a meaningful kinetic energy operator, one would match derivatives at the boundaries of the functions (e.g., as in finite element methods) or, alternatively, allow for overlapping basis functions. In either case, one achieves the desired scaling of ${\cal O}(N^2)$ terms in the Hamiltonian for all system sizes. Another possibility is to use a non-Galerkin grid-based representation, as embodied in  finite-difference methods. In \app{finite_diff}, we provide explicit forms for the second quantized molecular electronic structure Hamiltonian in such a discretization with ${\cal O}(N^2)$ terms. We focus on a different route to reducing the number of terms in the Hamiltonian, namely to use a pair of basis sets in which the different components in the Hamiltonian (kinetic and potential) are diagonal. This property is offered by the plane wave basis and its dual representation, which we now discuss.

\subsection{The Plane Wave Dual Basis}
\label{sec:pwd_basis}

Like Gaussian orbitals, plane waves have also enjoyed a long history of use in classical approaches to electronic structure. While plane waves have never been studied as a basis for quantum computation of electronic structure, they have many desirable properties as a basis; for instance, their periodicity makes them convenient for crystalline solids. The plane wave basis is defined subject to periodic boundary conditions in a computational cell of volume $\Omega$ and the integrals in \eq{one_body_ints} and \eq{two_body_ints} are defined using the Coulomb potential obtained from solving Poisson's equation subject to periodic boundary conditions (see \app{plane_waves} for review),
\begin{equation}
\label{eq:periodic_coulomb_r}
V\left(r, r'\right) = \frac{4 \pi}{\Omega} \sum_{\nu} \frac{\cos \left[k_\nu \cdot \left(r - r'\right)\right]}{k_\nu^2}
\quad \quad \quad \quad
U\left(r\right) = - \frac{4 \pi}{\Omega} \sum_{j,\nu} \zeta_j \frac{\cos \left[k_\nu \cdot \left(r - R_j\right)\right]}{k_\nu^2}
\end{equation}
where $R_j$ are nuclei coordinates, $\zeta_j$ are nuclei charges and $k_\nu$ is a vector of the plane wave frequencies at the $\nu^{\textrm{th}}$ harmonic of the computational cell in three dimensions, excluding the zero mode. We will assume a cubic cell for simplicity. The zero mode gives a divergent term  but for all charge-neutral systems the divergence from the electron-electron interaction cancels with the divergence from the external potential and contributes only a constant term which depends on the unit cell shape (for a derivation of this term, see Appendix F in \cite{Martin2004}).

Using a plane wave basis enforces a periodic charge distribution, natural for crystalline solids. As discussed in \app{basis_errors}, one can also represent finite systems such as molecules using plane waves by choosing the cell volume $\Omega$ to be sufficiently large so that the periodic images do not interact \cite{Martin2004} or by using a truncated Coulomb operator, which completely
eliminates periodic images\cite{fusti2002accurate}. Because plane waves have no knowledge of the atomic positions, one requires more plane waves than Gaussian orbitals in order to obtain the same level of basis set accuracy for most materials. Pseudopotentials reduce the ratio of the number of plane waves needed to the number of Gaussian orbitals needed for the same energy accuracy to only roughly a factor of ten~\cite{tosoni2007comparison,booth2016plane}. However, within a pseudopotential formulation, the asymptotic rate of convergence in both basis sets is dominated by the resolution of the electron-electron cusp, giving a basis set discretization error that scales as ${\cal O}(1/N)$~\cite{gruneis2013explicitly,hattig2011explicitly}. Thus, the asymptotic scaling of algorithms for simulating electronic structure (including the single-molecule case) can be compared directly whether $N$ represents Gaussian orbitals or plane wave orbitals. We substantiate this notion concretely in \app{basis_errors}. Also note that
in condensed phase systems with delocalized electrons, plane waves are especially competitive with Gaussians, and
in special cases such as jellium (the focus of \sec{sec_three}) are substantially more compact.

Within the plane wave basis, we can see immediately that the two-body Coulomb operator has only ${\cal O}(N^3)$ terms instead of ${\cal O}(N^4)$ terms. This reduction in the number of terms arises due to momentum conservation which constrains the allowable transitions between plane waves as they are eigenstates of the momentum operator. As we review in \app{plane_waves}, the complete Hamiltonian in the plane wave basis takes the well-known form:
\begin{equation}
 H = \underbrace{\frac{1}{2} \sum_{p, \sigma} k_p^2 \, c_{p,\sigma}^\dagger c_{p,\sigma}}_{T} - \underbrace{\frac{4 \pi}{\Omega} \sum_{\substack{p \neq q \\ j,\sigma}} \left(\zeta_j \frac{e^{i \, k_{q-p} \cdot R_j}}{k_{p-q}^2}\right) c^\dagger_{p, \sigma} c_{q, \sigma}}_{U} + \underbrace{\frac{2 \pi}{\Omega} \sum_{\substack{(p, \sigma) \neq (q, \sigma') \\ \nu \neq 0}} \frac{c^\dagger_{p,\sigma} c_{q,\sigma'}^\dagger c_{q + \nu,\sigma'} c_{p - \nu,\sigma}}{k_\nu^2}}_{V}
 \label{eq:periodic_coulomb}
\end{equation}
where $\sigma \in \{\uparrow, \downarrow\}$ is the spin degree of freedom and we have truncated the operators to the support of plane waves with frequencies $k_\nu = 2 \pi  \nu / \Omega^{1/3}$ such that $\nu$ is a three-dimensional vector of integers with elements in $[-N^{1/3}, N^{1/3}]$. In the above summation notation, addition of momenta is carried out modulo the maximum momentum. Aliasing the momenta in this way is equivalent to evaluating the integrals in \eq{one_body_ints} and \eq{two_body_ints} by sampling at $N$ evenly spaced grid points, a common practice in electronic structure codes sometimes called dualling~\cite{remler1990molecular,mcclain2017gaussian}. Dualling causes the plane wave Hamiltonian matrix elements to deviate from a Galerkin discretization, but this discrepancy is similar to basis error and vanishes as the number of plane waves increases. Importantly, the dualling form of the plane wave matrix elements is essential to give the desirable properties of the matrix elements in the dual basis we now discuss.

The Fourier transform of the complete plane wave basis (i.e. in the limit of infinite volume $\Omega$ and infinite momentum cutoff) is a basis of delta functions (a grid). But by applying the discrete Fourier transformation to a basis of $N$ plane waves, one obtains a new set of basis functions resembling a smooth approximation to a grid with lattice sites at the locations $r_p = p\, (\Omega / N)^{1/3}$. We call these functions the ``plane wave dual basis''. In electronic structure, the plane wave dual basis has previously been considered
in the context of reduced scaling density functional calculations~\cite{skylaris2002nonorthogonal,skylaris2005introducing}.
As a basis set where each function is associated with a real-space coordinate value, the plane wave dual basis can 
also be viewed as a discrete variable representation (DVR)~\cite{dvrreview}. In particular, it is
a relative of the sinc DVR basis widely used in quantum dynamics simulations~\cite{lill1982discrete,shizgal1984discrete,colbert1992novel,Jones2016,dvrreview}; although unlike in the standard sinc basis where the kinetic energy operator is approximate when using a finite basis,
here the kinetic energy operator is always treated exactly. 
However, the primary novelty about the plane wave dual basis in this work is its use in quantum computation and the specific properties of the basis that we exploit to enable especially efficient quantum algorithms. We derive the closed-form expressions for the plane wave dual basis functions and associated operators in \app{plane_wave_dual}, and further elucidate connections to DVR there.

While the plane wave dual basis functions are not strictly localized in space, they nevertheless diagonalize the potential operators of \eq{periodic_coulomb} within the dualling approximation, analogous to the conversion between the plane wave
and real space forms of the potential operators
via a continuous Fourier transform in a complete plane wave basis. As we derive in \app{plane_wave_dual}, by applying this Fourier transform, the Hamiltonian in the dual basis becomes
\begin{align}
\label{eq:dualpot}
H & = \underbrace{\frac{1}{2\, N} \!\!\!\sum_{\nu, p, q, \sigma} \!\! k_\nu^2 \cos \left[k_\nu \cdot r_{q - p} \right] a^\dagger_{p, \sigma} a_{q,\sigma}}_{T}
- \underbrace{\frac{4 \pi}{\Omega} \sum_{\substack{p,\sigma \\ j, \nu\neq 0}} \frac{\zeta_j \, \cos\left[ k_{\nu} \cdot \left(R_j - r_{p}\right)\right]}{k_\nu^2} n_{p, \sigma}}_{U} +
\underbrace{\frac{2 \pi}{\Omega } \!\!\!\!\! \sum_{\substack{(p, \sigma) \neq (q, \sigma') \\ \nu \neq 0}}\!\!\!\!\!\!\!\! \frac{\cos \left[k_\nu \cdot r_{p-q}\right]}{k_\nu^2} \, n_{p, \sigma} n_{q, \sigma'}}_{V}
\end{align}
where $n_{p} = a^\dagger_{p} a_{p}$ is the number operator. As one-body operators, $T$ and $U$ never have more than ${\cal O}(N^2)$ terms. We can also see the two-body potential operator $V$ is diagonal with only $\Theta(N^2)$ terms. Due to the unitarity of the discrete Fourier transform, the operators in \eq{periodic_coulomb} and \eq{dualpot} are exactly isospectral; there is no loss of accuracy associated with using one representation instead of the other. Thus, the plane wave dual basis offers all advantages of the plane wave basis with $\Theta(N^2)$ terms. As we show in \app{qubit_ham}, the plane wave dual basis Hamiltonian can be mapped to qubits under the Jordan-Wigner transformation as
\begin{align}
\label{eq:jw_ham}
H & =  \sum_{\substack{p, \sigma \\ \nu \neq 0}}\left(\frac{\pi}{\Omega \, k_\nu^2} - \frac{k_\nu^2}{4 \, N} + \frac{2\pi}{\Omega} \sum_{j}\zeta_j \frac{\cos\left[k_\nu \cdot \left(R_j-r_p\right)\right]}{k_\nu^2}\right) Z_{p,\sigma}
+ \frac{\pi}{2\,\Omega } \sum_{\substack{(p, \sigma) \neq (q, \sigma') \\ \nu \neq 0}} \frac{\cos \left[k_\nu \cdot r_{p-q}\right]}{k_\nu^2} Z_{p,\sigma} Z_{q,\sigma'}\\
& + \frac{1}{4\, N} \sum_{\substack{p \neq q \\ \nu, \sigma}} k_\nu^2 \cos \left[k_\nu \cdot r_{q - p} \right] \left(X_{p,\sigma} Z_{p + 1,\sigma} \cdots Z_{q - 1,\sigma} X_{q,\sigma} + Y_{p,\sigma} Z_{p + 1,\sigma} \cdots Z_{q - 1,\sigma} Y_{q,\sigma} \right)
+ \sum_{\nu \neq 0} \left(\frac{k_\nu^2}{2}- \frac{\pi \, N}{\Omega \, k_\nu^2} \right) I\nonumber.
\end{align}
where $X_p$, $Y_p$ and $Z_p$ are Pauli operators acting on qubit $p$.

\subsection{The Fermionic Fast Fourier Transform}
\label{sec:fqft}

A useful feature of the Hamiltonian representation introduced in \sec{pwd_basis} is that one can rotate the system from the plane wave dual basis (where the potential operator is diagonal) to the plane wave basis (where the kinetic operator is diagonal) using an efficient quantum circuit that is related to the fast Fourier transform. This operation allows one to efficiently prepare the initial state for classes of interesting physical systems whose ground state is well approximated by a mean-field state of delocalized electron orbitals (\sec{sec_three}), as well as to improve the efficiency of quantum measurements (\sec{measurement}). The usual quantum Fourier transform would be appropriate to diagonalize the kinetic energy operator for a binary encoding of the state in real space, as used in \cite{Meyer1996,Meyer1997,Zalka1998,Boghosian1998,Boghosian1998b,Wiesner1996,Lloyd1996,Kassal2008,Kivlichan2016}. However, our second-quantized encoding of the state necessitates a special version of the fast Fourier transform which we refer to as the ``fermionic fast Fourier transform'' (FFFT). Note that the word ``quantum'' does not appear in this name because our implementation of the FFFT does not offer any quantum advantage over its classical analog.

The fast Fourier transformation was first applied to fermionic systems for quantum computing purposes in \cite{Verstraete2009} and improved in the context of tensor network simulations in \cite{Ferris2014}. While \cite{Ferris2014} showed that the FFFT could be realized with ${\cal O}(\log N)$ depth using arbitrary two-qubit gates, in \app{fqft}, we extend the method of \cite{Verstraete2009} to show that the FFFT can be implemented for three spatial dimensions using a planar lattice of qubits with ${\cal O}(N)$ depth. While past work has focused entirely on describing the FFFT under the Jordan-Wigner transformation \cite{Ferris2014,Verstraete2009}, we generalize the approach to arbitrary mappings including Bravyi-Kitaev \cite{Bravyi2002,Seeley2012,Tranter2015} and other modern approaches \cite{Whitfield2016,Bravyi2017,Havlicek2017}.

The essential function of the FFFT is to perform the following single-particle rotation:
\begin{equation}
c^\dagger_\nu = {\rm FFFT}^\dagger a^\dagger_\nu \, {\rm FFFT} =  \sqrt{\frac{1}{N}}  \sum_{p} a^\dagger_{p} e^{- i \,k_\nu \cdot r_p}
\quad \quad \quad \quad
c_\nu ={\rm FFFT}^\dagger a_\nu \, {\rm FFFT} =  \sqrt{\frac{1}{N}}  \sum_{p} a_{p} e^{i \,k_\nu \cdot r_p}.
\end{equation}
As a clarifying example, the two-dimensional {\rm FFFT} that acts on spin orbitals $p$ and $q$ (which are not necessarily adjacent in lexicographical ordering) is 
\begin{equation}
F_0^\dagger = e^{-i(\pi/4)a_q^\dagger a_q} e^{i(\pi/4)a_p^\dagger a_p} e^{i(\pi/4)f_{\rm swap}}e^{-i(\pi/2) a_q^\dagger a_q},\label{eq:f2def}
\end{equation}
where $f_{\rm swap}$ generates the ``fermionic swap operator'', which has been proposed for use in quantum computer simulations
in \cite{Wecker2015a}. We define $f_{\rm swap}$ in a mapping-independent way (i.e. not specific to Jordan-Wigner) as
\begin{equation}
\label{eq:fswap}
f_{\rm swap}= (1+a_p^\dagger a_q +a_q^\dagger a_p -a_p^\dagger a_p -a_q^\dagger a_q ).
\end{equation}
This operator is referred to as a fermionic swap because it has the property that it swaps the spin orbitals $p$ and $q$ (up to a global phase) while maintaining proper anti-symmetrization.  For example, $f_{\rm swap} a_p^\dagger f_{\rm swap} = a_q^\dagger$ and vice versa.  Using these definitions it can be shown (see ~\app{fqft}) that
\begin{equation}
F_0^\dagger a_p^\dagger F_0 = \frac{a_q^\dagger + a_p^\dagger}{\sqrt{2}}
\quad \quad \quad \quad
F_0^\dagger a_q^\dagger F_0 = \frac{a_q^\dagger - a_p^\dagger}{\sqrt{2}}.
\end{equation}
This reveals that $F_0$ acts as a Hadamard transform, or two-dimensional Fourier transform, on the creation operators that define the two dimensional subspace.  Then, by combining these operations together with phase shifts one can follow the same reasoning used in the Cooley-Tukey fast Fourier transform algorithm \cite{Cooley1965} to construct the FFFT out of these operations and phase shifts.  We present this argument formally in~\app{fqft}. We also show that the entire FFFT in three dimensions can be implemented on a planar lattice of qubits with gate depth of ${\cal O}(N)$.

Just as the quantum Fourier transform diagonalizes the kinetic operator in real space simulations, it is shown in \app{plane_waves} and \app{plane_wave_dual} that
\begin{equation}
\label{eq:kindiag0}
T =  \frac{1}{2\, N} \sum_{\nu, p, q, \sigma} k_\nu^2 \cos \left[k_\nu \cdot r_{q - p} \right] a^\dagger_{p, \sigma} a_{q,\sigma} = {\rm FFFT}^\dagger \left(\frac{1}{2} \sum_{\nu, \sigma} k_\nu^2 \, a_{\nu,\sigma}^\dagger a_{\nu,\sigma}\right) {\rm FFFT}.
\end{equation}
Thus, an alternative expression for the molecular electronic structure Hamiltonian in the plane wave dual basis is
\begin{equation}
\label{eq:pwdb_ham_diag}
H = {\rm FFFT}^\dagger \left(\sum_{\nu, \sigma} \frac{k_\nu^2}{2} \, a_{\nu,\sigma}^\dagger a_{\nu,\sigma}\right) {\rm FFFT}
- \frac{4 \pi}{\Omega} \sum_{\substack{j,p,\sigma\\ \nu\neq 0}} \zeta_j\frac{\cos\left[k_{\nu} \cdot \left(R_j - r_{p}\right)\right]}{k_\nu^2} n_{p, \sigma}
+\frac{2 \pi}{\Omega} \!\!\!\!\!\!\!  \sum_{\substack{(p,\sigma) \neq (q,\sigma') \\ \nu \neq 0}} \!\!\!\!\!\!\! \frac{\cos \left[k_\nu \cdot \left(r_p - r_q\right)\right]} {k_\nu^2} n_{p,\sigma} n_{q,\sigma'}.
\end{equation}
We choose to write the kinetic operator using the FFFT relation to emphasize that the Hamiltonian has the special property that all components of it are diagonal in either the plane wave or plane wave dual representations. In addition to the advantages of having only $\Theta(N^2)$ terms, in the subsequent sections we will make frequent use of this diagonal property. In some circumstances, we will also desire to perform simulation in the plane wave dual basis after preparing an initial state that is a product state in the plane wave basis; this can be accomplished by applying the FFFT to a product state. Finally, we note that while prior work has leveraged the diagonality of momentum and potential operators in real space, our use of second quantization allows us to use dramatically fewer qubits and also avoids the challenge of anti-symmetrizing the initial state which complicates first-quantized methods~\cite{Kivlichan2016,Kassal2010}.

\section{Improved Algorithms for Quantum Computation}
\label{sec:sec_two}

We analyze the cost of applying several types of quantum simulation algorithms to the Hamiltonians introduced in \sec{sec_one} in this section. In \sec{trotter_alg} and \sec{taylor_alg}, we focus on Trotter-Suzuki and Taylor series algorithms for time-evolution, which can be used to prepare electronic structure ground states when used in conjunction with the phase estimation algorithm \cite{Aspuru-Guzik2005,Kitaev1995}. Specifically, the phase estimation algorithm will project a simulation register into eigenstate $\ket{j}$ with probability $|\braket{j}{\psi_0} |^2$ where $\ket{\psi_0}$ is the ``reference state'' from which the phase estimation procedure begins. The phase estimation procedure involves taking a short time dynamical simulation $e^{-iH\delta}$ such that the error in the eigenvalues of the effective Hamiltonian for the simulated unitary is at most ${\cal O}(\epsilon)$.  If the cost of this short dynamical simulation is $F(\epsilon)$ then the cost of phase estimation is then in ${\cal O}(F(\epsilon)/\epsilon)$. Thus minimizing the costs of dynamical simulation is vitally important for phase estimation. 

Typically, one is interested in projecting to the ground state and $\ket{\psi_0}$ is chosen to be the Hartree-Fock state, defined as the lowest energy single Slater determinant approximation to the ground state \cite{Helgaker2002}. The Hartree-Fock algorithm is a classical self-consistent mean-field procedure for finding this state in terms of a series of single-particle rotations. To prepare the Hartree-Fock state from any product state one can evolve under the anti-Hermitian operator $\sum_{pq} \theta_{pq} \, a^\dagger_p a_q$ for some amplitudes $\theta_{pq} = -\theta_{qp}^*$ determined by the Hartree-Fock procedure. The cost of performing this evolution depends on the values of $\theta_{pq}$. However, the Hartree-Fock state does not need to be prepared exactly in order to provide a good reference and a variational outer-loop can be employed to take fewer Trotter steps \cite{Wecker2014}. Accordingly, we assume that the gate complexity of state preparation is less than the cost of the algorithms described in \sec{trotter_alg} and \sec{taylor_alg}. For systems of delocalized electrons, such as jellium (the focus of \sec{sec_three}), state preparation can be efficiently 
accomplished with $\log(N)$ gate depth for arbitrary two-qubit gates, or with linear gate depth using a planar architecture, using the FFFT of \sec{fqft}.

Before beginning our analysis we will make a few comments about how the results of this paper should be compared to prior work. Most prior quantum algorithms for electronic structure have focused on the simulation of finite systems consisting of a small number of atoms \cite{Aspuru-Guzik2005}. As discussed in \sec{sec_one} and \app{basis_errors}, one can also use the plane wave dual basis for such simulations by choosing the unit cell volume $\Omega$ to be large, or by truncating the Coulomb operator, which exactly eliminates
the periodic images. However, we expect that the plane wave dual basis will be most useful for simulating systems with
periodicity in at least one of the spatial dimensions, such as crystalline wires, surfaces, and solids, similar to the
main current uses of plane wave bases in classical electronic structure. In either case, one is interested in the cost of simulation when the number of basis functions $N$ grows towards the \emph{continuum limit}. We do not expect the total energy to be extensive in $N$.

One should also bound the cost of simulation as the number of particles $\eta$ grows. While it is reasonable to wonder how the cost of simulation grows with molecule size, molecules do not necessarily grow in a systematic fashion. For instance, molecules can have larger $\eta$ by replacing lighter atoms with heavier onces or by adding atoms to a molecule. When using plane waves to treat materials one is usually interested in the properties of an infinite material that is periodic over some computational cell (a collection of unit cells) of volume $\Omega$. As with molecules, one is sometimes interested in how the complexity of a method scales as one fixes the computational cell size and increases the number of particles (e.g. by replacing lighter atoms with heavier ones). But unlike when simulating molecules, the notion of scaling towards the \emph{thermodynamic limit} is well-defined in the context of periodic solids. The thermodynamic limit is approached as one grows $\eta$ by increasing the number of unit cells in the computational cell while keeping a fixed averaged density $\rho = \eta / \Omega$. Accordingly, for both molecules and materials we report the asymptotic scaling of algorithms in terms $\eta$, $N$ and $\rho$ but are most interested in the fixed density scalings corresponding to molecules growing by addition of atoms and materials growing towards the thermodynamic limit. In \sec{measurement} we report the number of measurements required in terms of both a fixed absolute error $\epsilon$  and a fixed relative error $\mu = \epsilon\, \eta$ as one is interested in fixed relative error while scaling towards the thermodynamic limit but interested in fixed absolute error otherwise. This is because physical total energies are extensive in $\eta$.

\subsection{The Cost of Time-Evolution using Trotter-Suzuki Methods}
\label{sec:trotter_alg}

Trotterization is perhaps the simplest method for simulating electron dynamics in the plane wave dual basis. Trotterization solves the problem of compiling $e^{-iHt}$ into fundamental gates by noting that if $H=\sum_\ell^L H_\ell$, where each $e^{-iH_\ell t}$ can be easily compiled into fundamental gates, then $e^{-iHt}$ can be simulated by a time-dependent Hamiltonian that rapidly switches each term on and then off.  If the frequency of these switches is sufficiently high then, from the perspective of the quantum system, the entire Hamiltonian is active throughout the evolution; e.g. for large $r$,
\begin{equation}
e^{-iHt}= \left[\left(\prod_{\ell=1}^L e^{-iH_\ell t/2r}\right)\left(\prod_{\ell=L}^1 e^{-iH_\ell t/2r}\right)\right]^r +\mathcal{O}(t^3/r^2).\label{eq:trotdef}
\end{equation}
Here we have employed the second-order Trotter formula, which is often more practical for chemistry simulations than higher-order decompositions \cite{Poulin2014}.

The value of $r$ that is needed for this expansion depends subtly on the terms in the Hamiltonian. If the Hamiltonian terms commute then the error in the simulation is zero.  Thus, the error does not depend on the norms of the Hamiltonian terms, but rather it depends on their commutators.  Specifically it was shown in \cite{Poulin2014} that
\begin{equation}
\left|\max_\psi \bra{\psi} \left[\left(\prod_{\ell=1}^L e^{-iH_\ell t/2r}\right)\left(\prod_{\ell=L}^1 e^{-iH_\ell t/2r}\right)\right]^r  \!\!\! - e^{-iHt}\ket{\psi}\right| \in \mathcal{O}\left( \max_{\psi}\sum_{\substack{\beta,\alpha\le \beta,\\ \gamma <\beta}} \left| \bra{\psi}[H_\alpha,[H_{\beta},H_{\gamma}]]\ket{\psi}\right|\frac{t^3}{r^2}\right),\label{eq:commbound}
\end{equation}
where $\ket{\psi}$ is a state restricted to the $\eta$-electron manifold. Thus, once a particular ordering of the terms is chosen then an upper bound on the scaling of $r$ can be found based on the commutator norms of the terms.

The conventional approach to second-quantized simulation would be to Trotterize the Hamiltonian of \eq{jw_ham}. This approach is outlined in detail along with improved methods for simulating evolution under the kinetic terms in the Hamiltonian in \app{alt_trotter}. However, the approach we analyze here is to simulate evolution by switching between the plane wave dual basis and the plane wave basis to diagonalize the potential and kinetic operators. Using $H= T + U + V$, we write
\begin{equation}
e^{-iHt}=e^{-i(U+V)t/2} \, {\rm FFFT}^\dagger \, e^{-i\frac{t}{2} \sum_{\nu, \sigma} k_\nu^2 \, a_{\nu,\sigma}^\dagger a_{\nu,\sigma}}\,{\rm FFFT}\, e^{-i(U+V)t/2} +\mathcal{O}(t^3).\label{eq:trotterdecomp}
\end{equation}
Representing the kinetic terms as diagonal operators has two effects.  Firstly it reduces the number of commutators in~\eq{commbound} which leads to better bounds on the error.  The second advantage is that the kinetic operator only contains $\mathcal{O}(N)$ local terms that all commute.  This allows us to simulate the kinetic operator in depth $\mathcal{O}(1)$ after performing this basis transformation on a quantum computer that has arbitrary single qubit rotations as a fundamental gate. 

To understand the cost of this approach note that each of the $r$ steps in the Trotter algorithm comprises of, from~\eq{trotterdecomp}, two simulations of the potential energy Hamiltonian, the FFFT and its inverse, and a simulation of the kinetic operator in the plane wave basis.  The operator $U+V$ is the sum of $\Theta(N^2)$ number operators. Each number operator can be simulated using $\mathcal{O}(1)$ CNOT gates and single qubit rotations~\cite{Whitfield2010}.  It follows that $e^{-i(U+V)t}$ can be simulated using $\mathcal{O}(N^2)$ gates and in depth $\mathcal{O}(N)$ if the quantum computer has all to all connectivity between qubits, without the use of ancillae.  We show in~\app{qubit_cycle} that it can also be simulated in a planar nearest neighbor architecture in depth $\mathcal{O}(N)$ without ancillae.  Similarly, $e^{-i\frac{1}{2} \sum_{\nu, \sigma} k_\nu^2 \, a_{\nu,\sigma}^\dagger a_{\nu,\sigma}t}$ requires $\mathcal{O}(N)$ gates to implement.  The number of gates required to perform the FFFT scales as $\mathcal{O}(N\log N)$~\cite{Ferris2014,Verstraete2009} with ${\cal O}(N)$ depth for the three-dimensional transform implemented for qubits connected on a planar lattice, which is proven rigorously in \app{fqft}. Consequently, costs are asymptotically dominated by simulation of the potential. Thus, the Trotter simulation can be performed using a circuit of depth $\mathcal{O}(N r)$ on a planar lattice.

In the \app{TSerror} we show that to simulate for time $t$ and achieve error $\epsilon$ it suffices to choose $r$ such that
\begin{equation}
r\in{\Theta}\left(\frac{\eta^2 N^{5/6}t^{3/2}}{\Omega^{5/6}\sqrt{\epsilon}}\sqrt{1+\frac{\eta\,\Omega^{1/3}}{N^{1/3}}} \right),
\end{equation}
implying that the gate depth of our approach to Trotter-Suzuki based simulations is
\begin{equation}
\mathcal{O}(Nr) \subseteq \mathcal{O}\left(\frac{\eta^2 N^{11/6}t^{3/2}}{\Omega^{5/6}\sqrt{\epsilon}}\sqrt{1+\frac{\eta\,\Omega^{1/3}}{N^{1/3}}}\right).
\label{eq:trotscale}
\end{equation}
There are a number of ways that we can understand this scaling depending on how problem size grows. If we assume we grow our simulation size without changing the system density (i.e. $\rho = \eta / \Omega \in {\cal O}(1)$) then the gate depth is in ${\cal O}(N^{5/3} \eta^{11/6} t^{3/2} / \epsilon^{1/2})$. If we only are interested in the scaling with $N$ then we can take $\eta \in \mathcal{O}(N)$ to find that the gate depth is in $\mathcal{O}(N^{7/2})$. This is more than quadratically better than the best known rigorous bounds on the circuit depth for Trotter-based chemistry simulations, $\mathcal{O}(N^8)$~\cite{Wecker2015a}.  However, just as the bounds for $r$ in~\cite{Wecker2015a} proved to be polynomially loose~\cite{Poulin2014,BabbushTrotter}, we expect the empirical performance of our approach to be better than~\eq{trotscale} suggests.

Despite the obvious differences in scaling, a full comparison between prior Trotter-Suzuki work in different bases and this result remains challenging. This is because comparing costs for any specific system will depend on the precise $N$ needed in the given basis, which will be problem specific. Nonetheless, the quadratic difference between the two complexities strongly suggests that for fault-tolerant applications, our simulation method will be competitive because most of the physical qubits required for the simulation arise from executing single qubit rotations fault tolerantly~\cite{Reiher2017,Jones2012}. As a final note, while we use an exact evaluation of the potential here it is possible to leverage the local nature of the plane wave dual basis in order to approximate the potential on-the-fly using the Barnes-Hut algorithm or other fast multipole methods, which require $\widetilde{\mathcal{O}}(N)$ gates. Thus, it is possible, in principle, to achieve a gate complexity that matches the cited circuit depths to within logarithmic equivalence.  However, a naive application of the fast multipole method would require that the quantum computer coherently apply the algorithm for each configuration in the superposition and is likely to be impractical in near-term quantum computers.

\subsection{Bounding Cost of Time-Evolution using Taylor Series Methods}
\label{sec:taylor_alg}

With the exception of \cite{BabbushSparse1,BabbushSparse2,Kivlichan2016,Low2016,BabbushSymmetry}, all prior papers which analyze the time-evolution of electronic structure Hamiltonians use the Trotter-Suzuki decomposition. Even the most elaborate Trotter schemes scale sub-polynomially but not poly-logarithmically with respect to the reciprocal of the simulation error, $1/\epsilon$,~\cite{Berry2007,Wiebe2008}. In \cite{Berry2013,Berry2015}, Berry \emph{et al.}~ combined the results of \cite{Berry2007,Cleve2009,Berry2014} to show a technique for performing time-evolution of arbitrary Hamiltonians with sub-logarithmic dependence on the inverse precision. Since then, several papers have introduced other ``post-Trotter'' methods with improved dependence on $\epsilon$ \cite{Low2017,Novo2016}. The ``Taylor series'' techniques of \cite{Berry2013} were first applied to chemistry in \cite{BabbushSparse1,BabbushSparse2}. The result of \cite{BabbushSparse1} is an algorithm with gate complexity of $\widetilde{\cal O}(N^5)$ and the result of \cite{BabbushSparse2} is a more complicated algorithm which exploits the sparseness of the configuration interaction representation of the Hamiltonian in order to perform simulation with gate complexity $\widetilde{\cal O}(\eta^2 N^3)$, where $\eta$ is the number of electrons.

Using the Taylor series method in the plane wave dual representation we will be able to outperform both of these bounds. In this section, we will show that one can perform time-evolution of the Hamiltonian with gate complexity of $\widetilde{\cal O}(N^4)$ using an approach that is much simpler than the aforementioned Taylor series based algorithms. This scheme is similar to the ``database algorithm'' protocol of \cite{BabbushSparse1}, which scaled at least as $\widetilde{\cal O}(N^6)$ in that work. In \app{onthefly}, we build on the results of this section to show a more complicated algorithm, inspired by the ``on-the-fly'' algorithms from \cite{BabbushSparse1} and \cite{BabbushSparse2}, which in the plane wave dual basis has gate complexity of $\widetilde{\cal O}(N^{11/3})$ and gate depth of $\widetilde{\cal O}(N^{8/3})$, making this the most efficient algorithm for time-evolution of an electronic structure system in the literature.

We will not go into details about how the Taylor series method works and instead refer readers to \cite{Berry2015}. We will describe what is required to implement the techniques and bound the cost of our approach. The Taylor series method begins with the observation that any local Hamiltonian, e.g. \eq{jw_ham}, can be expressed as
\begin{equation}
\label{eq:taylor_start}
H = \sum_{\ell = 0}^{L - 1} W_\ell H_\ell
\quad \quad \quad
\textrm{s.t.} \quad W_\ell \in \mathbb{R} \quad \quad 
H_{\ell}^{2} = I
\end{equation}
where $W_\ell$ are real scalars and $H_\ell$ are self-inverse operators which act on qubits; e.g., the $H_\ell$ are the strings of Pauli operators in \eq{jw_ham}. The Taylor series simulation technique is described in \cite{Berry2015} in terms of queries to two oracle circuits. The first oracle circuit acts on an empty ancilla register of ${\cal O}(\log L)$ qubits and prepares a particular superposition state related to \eq{taylor_start},
\begin{equation}
\textsc{prepare}(W) \ket{0}^{\otimes \log L} \mapsto \sqrt{\frac{1}{\Lambda}} \sum_{\gamma = 0}^{L - 1} \sqrt{W_\ell} \ket{\ell}
\quad \quad \quad \quad \Lambda = \sum_{\ell = 0}^{L - 1} \left | W_\ell \right |
\end{equation}
where $\Lambda$ is a normalization parameter that turns out to have significant ramifications for the overall algorithm complexity. The second oracle circuit we require acts on the ancilla register $\ket{\ell}$ as well as the system register $\ket{\psi}$ and directly applies one of the $H_\ell$ to the system, controlled on the ancilla register. For this reason, we refer to the ancilla register $\ket{\ell}$ as the ``selection register'' and name the oracle accordingly,
\begin{equation}
\label{eq:selecth}
\textsc{select}(H) \ket{\ell} \ket{\psi} = \ket{\ell} H_\ell \ket{\psi}.
\end{equation}
Note that the self-inverse nature of the $H_\ell$ operators implies that they are both Hermitian and unitary, which means they can be applied directly to a quantum state.

Suppose that the circuit $\textsc{prepare}(W)$ can be applied at gate complexity $P$ and the circuit $\textsc{select}(H)$ can be applied at gate complexity $S$. Then, the main result of \cite{Berry2015} is that one can straightforwardly perform a quantum simulation under $H$ for time $t$ to unitary operator precision $\epsilon$ at gate complexity
\begin{equation}
\label{eq:result}
\widetilde{\cal O}\left(\left(S + P\right) \Lambda \, t\right)
\end{equation}
with spatial overheads and precision costs polylogarithmically bounded in $\epsilon$. Since the bound on the Hamiltonian norm from \app{operators} is obtained using the triangle-inequality, it also asymtotically bounds $\Lambda$ at ${\cal O}(N^{7/3} / \Omega^{1/3} + N^{5/3} / \Omega^{2/3})$. We now describe how these two oracles can be implemented so that $(S + P) \in \widetilde{\cal O}(N^2)$.

First, we discuss implementation of $\textsc{select}(H)$. From \eq{jw_ham} it is clear that there are $L = \Theta(N^2)$ terms which can be indexed by only two indices, $p$ and $q$. For the purposes of this section we will further suppose that $p$ indexes both spin and position so that even values of $p$ correspond to spin-up orbitals and odd values of $p$ correspond to spin-down orbitals. We will ignore the identity term and index the local $Z_p$ terms whenever $p = q$. For $p \neq q$ there are three terms in \eq{jw_ham} which we will refer to as the $ZZ$ term, the $XZX$ term and the $YZY$ term. An ancilla qubit, $b$ will be introduced and if $b = 0$ then the pair $(p, q)$ will refer to the $ZZ$ term whereas if $b = 1$, the pair $(p, q)$ will refer to the $XZX$ and $YZY$ terms. If $p > q$ we will refer to the $XZX$ terms whereas if $q > p$, we will refer to the $YZY$ term. Accordingly, our $\textsc{select}(H)$ circuit should have the following actions,
\begin{equation}
\label{eq:select}
\textsc{select}(H) \ket{p} \ket{q} \ket{b} \ket{\psi} \mapsto \begin{cases}
\ket{p} \ket{q} \ket{b} Z_p \ket{\psi} & p = q \\
\ket{p} \ket{q} \ket{b} Z_p Z_q \ket{\psi} & (b = 0) \wedge (p \neq q) \\
\ket{p} \ket{q} \ket{b} (X_q Z_{q + 1} \cdots Z_{p - 1} X_p) \ket{\psi} & (b = 1) \wedge (p > q) \wedge (p + q \, \!\!\mod\!\! \, 2 = 0 )\\
\ket{p} \ket{q} \ket{b} (Y_p Z_{p + 1} \cdots Z_{q - 1} Y_q) \ket{\psi} & (b = 1) \wedge (q > p) \wedge (p + q \, \!\!\mod\!\! \, 2 = 0)\\
\ket{p} \ket{q} \ket{b} \ket{\psi} & (b = 1) \wedge (p + q \, \!\!\mod\!\! \, 2 = 1).
\end{cases}
\end{equation}
Note that the condition involving $(p + q) \, \!\!\mod\!\! \, 2$ is necessary when the model contains a spin degree of freedom in order to conserve spin. This efficient encoding requires only $2 \log N$ ancilla for the selection register. The logic to select a term, shown in \eq{select}, involves only the operations $>$, $\wedge$ and $=$, which can all execute with $\widetilde{\cal O}(1)$ gates. Since the actual $H_\ell$ contain up to $N$ Pauli operators, we see that $\textsc{select}(H)$ can be circuitized with gate complexity $S \in \widetilde{\cal O}(N)$. For a specific implementation of how even more complex Pauli strings can be implemented from a selection oracle with this same gate complexity, see Section III of \cite{BabbushSparse1}.

Using the notation established in \eq{select} the preparation oracle should have the following actions,
\begin{equation}
\label{eq:prepw}
\textsc{prepare}(W) \ket{p}\ket{q}\ket{b} \mapsto \sqrt{\frac{1}{\Lambda}} \sum_{p,q,b} \sqrt{W_{p,q,b}} \ket{p}\ket{q}\ket{b}
\quad \quad \quad \quad \Lambda = \sum_{p,q,b} \left | W_{p,q,b} \right |
\end{equation}
where
\begin{equation}
\label{eq:amps}
W_{p,q,b} = \begin{cases}
\sum_{\nu \neq 0}\left(\frac{\pi}{2 \, \Omega \, k_\nu^2} - \frac{k_\nu^2}{8 \, N} +  \frac{\pi}{\Omega} \sum_{j}\zeta_j\frac{\cos\left[k_\nu \cdot \left(R_j-r_p\right)\right]}{k_\nu^2}  \right) & p = q \\
 \frac{\pi}{4\,\Omega } \sum_{\nu \neq 0} \frac{\cos \left[k_\nu \cdot r_{p-q}\right]}{k_\nu^2} & b = 0 \wedge (p \neq q) \\
\frac{1}{4\, N} \sum_{\nu} k_\nu^2 \cos \left[k_\nu \cdot r_{q - p} \right]  & b = 1 \wedge (p + q) \, \!\!\mod\!\! \, 2 = 0\\
1 & b = 1 \wedge (p + q) \, \!\!\mod\!\! \, 2 = 1.\\
\end{cases}
\end{equation}
We have added a factor of $1 / 2$ to the $Z$ and $ZZ$ coefficients of \eq{jw_ham} as those terms will execute twice; when $p = q$ this happens due to the $b$ degree of freedom and when $b = 0$ this happens because both $p > q$ and $p < q$ will occur. To implement $\textsc{prepare}(W)$ one can use an approach which mirrors the ``database algorithm'' introduced in \cite{BabbushSparse1}. The idea is based on results from \cite{Shende2006} which show that an arbitrary quantum state on $m$ qubits can be prepared using a circuit with no more than ${\cal O}(2^m)$ CNOT gates. Since $\textsc{prepare}(W)$ initializes a state on ${\cal O}(\log L)$ qubits where $L = \Theta(N^2)$, the techniques of \cite{Shende2006} would allow one to implement  $\textsc{prepare}(W)$ at gate complexity $P \in 2^{{\cal O}(\log L)} \in {\cal O}(N^2)$.

We have thus shown a constructive approach to Taylor series simulation of \eq{jw_ham} with total gate complexity
\begin{equation}
\widetilde{\cal O}\left(\left(S + P\right) \Lambda \, t\right) \in \widetilde{\cal O}\left(N^2 \Lambda \, t\right)
\in \widetilde{\cal O}\left(\frac{N^{11/3} t}{\Omega^{2/3}} + \frac{N^{13/3} t}{\Omega^{1/3}}\right).
\end{equation}
If we fix the system phase at $\rho = \eta / \Omega \in {\cal O}(1)$ and assume that $\eta \in \Theta(N)$ then we see that the algorithm scales asymptotically as $\widetilde{\cal O}(N^4 t)$. Though less efficient than the method of \sec{trotter_alg} by a factor of $\sqrt{N}$, this algorithm has logarithmic dependence on $\epsilon$, which is a superpolynomial advantage in $\epsilon$ over all Trotter schemes and an exponential advantage in $\epsilon$ over the method of \sec{trotter_alg}. In \app{onthefly} we extend these ideas to show a more involved implementation of $\textsc{prepare}(W)$ which results in overall gate complexity $\widetilde{\cal O}(N^{11/3})$ and gate depth of $\widetilde{\cal O}(N^{8/3})$. The concept of that approach is to compute the coefficients ``on-the-fly'' similar to the ``on-the-fly'' algorithm in \cite{BabbushSparse1}.

There are several ways in which these results could be improved. First, our bound on $\Lambda$ for the database algorithm is likely loose and should be studied numerically in order to estimate practical scaling. Second, following the construction detailed in \cite{BabbushSparse2}, one could simulate the plane wave dual Hamiltonian in the configuration interaction representation using the Taylor series approach and in doing so, reduce the spatial requirement of this algorithm from $\widetilde{\cal O}(N)$ to $\widetilde{\cal O}(\eta)$. That improvement would be especially meaningful in the dual basis due to the spatial overhead associated with using plane waves instead of Gaussian orbitals.

\subsection{Fewer Measurements for Variational Quantum Algorithms}
\label{sec:measurement}

An alternative to quantum phase estimation and other methods requiring time evolution for the study of electronic systems are quantum variational algorithms such as the variational quantum eigensolver~\cite{Peruzzo2013,McClean2015,McClean2016}.  These methods have garnered significant recent attention due to their simple experimental implementation and robustness to control errors \cite{OMalley2016}. Variational quantum algorithms involve a parameterized procedure (usually a parameterized quantum circuit) for preparing quantum states (the variational ansatz). The variational ansatz is iteratively improved by measuring an objective function and then using a classical optimization routine to suggest new parameters. The bottleneck we focus on here is the measurement step. While variational algorithms do not require long coherent evolutions, they usually require a large number of circuit repetitions for measurement purposes; the abstract of \cite{Wecker2015a} claims the primary challenge of these methods is that ``the required number of measurements is astronomically large for quantum chemistry applications''. Here, we show that use of the plane wave dual basis enables new bounds and strategies that drastically reduce the number of circuit repetitions required.

Usually (but not always, e.g. see \cite{Santagati2016}) the measurement objective is the expectation value of the energy on the current quantum state. The expense of this step typically depends on the norm and form of the Hamiltonian, and the exact method that is used to evaluate it. The simplest and most practical method of expectation value estimation relies on a form of quantum operator averaging that leverages the structure of these Hamiltonians as sums of tensor products of Pauli operators. The expectation value of the energy may be estimated by measuring the individual tensor products of Pauli operators $H_\ell$ on repeated, independent state preparations and summing the resulting estimates $\avg{H_\ell}$ together, weighted by their coefficient $W_\ell$, to get an estimate of the expectation value,
\begin{equation}
\label{eq:var_ham}
H = \sum_{\ell} W_\ell H_\ell \quad \quad \quad \quad \avg{H} = \sum_{\ell} W_\ell \avg{H_\ell}.
\end{equation}
This method has the advantage that negligible coherence time is required beyond state preparation in order to perform the required measurements, making it particularly amenable to implementation on quantum devices without error-correction. If one assumes no additional prior information and allows a variable number of measurements per term, the number of repeated preparations and measurements $M$ to estimate the value of $\langle H \rangle$ to a precision $\epsilon$ is known~\cite{Wecker2015a,Rubin2018} to be bounded by
\begin{align}
M \in {\cal O} \left(\left(\frac{1}{\epsilon}\sum_{\ell} \left |W_\ell\right|\right)^2\right)
\in \mathcal{O}\left( \left(\frac{N^{7/3}}{\epsilon \, \Omega^{1/3}} + \frac{N^{5/3}}{\epsilon \, \Omega^{2/3}}\right)^2\right)
\in {\cal O}\left(\frac{N^{14/3}}{\epsilon^2 \, \Omega^{2/3}}\left( 1 + \frac{1}{N^{4/3}\Omega^{2/3}}\right)\right)
\end{align}
where we have used the triangle inequality upper-bounds to the norm of the plane wave dual Hamiltonian derived in \app{operators}. Already, this bound is significantly lower than the best proven bound on the number of measurements required when using a Gaussian basis, ${\cal O}(N^8 / \epsilon^2)$; however, both bounds are loose.

Hamiltonians in the plane wave dual basis have a few special properties that allow us to make even fewer measurements. In particular, the Coulomb operators $U$ and $V$, are diagonal. A consequence of this is the local $H_\ell$ terms from each of these operators all commute with each other, allowing the use of a separate, unbiased estimator for the mean of $U + V$, without the use of an ancilla qubit~\cite{Santagati2016}. While the kinetic operator is not diagonal in the plane wave dual basis representation, one can perform the FFFT prior to measurement. This would change to the plane wave basis and diagonalize the kinetic operator. Thus, instead of independent wavefunction preparations for each $H_\ell$ within the sum for $T$, $U$ and $V$, either the entire operator $U + V$ or the entire operator $T$ can be measured completely on each circuit repetition. As variances add linearly for independent measurements, if we were to measure $T$, $U$ and $V$ individually by sampling bit strings in their eigenbasis, we would require a number of circuit repetitions scaling as
\begin{align}
\label{eq:diagonal_measure}
M & \in {\cal O}\left(\frac{\text{Var}_{\ket{\Psi}}\left[T\right] + \text{Var}_{\ket{\Psi}}\left[U + V\right]}{\epsilon^2}\right) 
\in {\cal O}\left(\frac{ \avg{T^2} - \avg{T}^2 + \avg{\left(U + V\right)^2} - \avg{U + V}^2}{\epsilon^2}\right)\\
& \in {\cal O}\left(\frac{\avg{V^2} + \avg{T^2}}{\epsilon^2}\right) \subseteq {\cal O}\left(\frac{\eta^4 N^{2/3}}{\epsilon^2 \, \Omega^{2/3}} + \frac{\eta^2 N^{4/3}}{\epsilon^2 \, \Omega^{4/3}}\right)
\in {\cal O}\left(\frac{\eta^{10/3} N^{2/3}}{\epsilon^2} + \frac{\eta^{2/3} N^{4/3}}{\epsilon^2}\right)\nonumber,
\end{align}
where we have used bounds from \app{operators}. In the final bound we have provided the scaling at fixed density, consistent with scalings in other sections of this paper. We see that for either finite molecules or bulk materials, since it must be the case that $\eta \in {\cal O}(N)$, this scaling is no worse than ${\cal O}(N^4/\epsilon^2)$.

As discussed at the beginning of \sec{sec_two}, one is often interested in studying the cost of converging periodic electronic structure calculations to the thermodynamic limit.  However, simulation of the ground state energy within fixed absolute error is unreasonable in the thermodynamic limit as the energy scale of the system grows scales asymptotically as ${\cal O}(\eta)$. Accordingly, when growing towards the thermodynamic limit one would be interested in achieving a fixed relative error $\mu = \epsilon \, \eta$. In terms of $\mu$, we see that scaling towards the thermodynamic limit is
\begin{equation}
{\cal O}\left(\frac{\eta^{4/3} N^{2/3}}{\mu^2} + \frac{N^{4/3}}{\mu^2\, \eta^{4/3}}\right)
\in {\cal O}\left(\frac{N^{2}}{\mu^2}\right),
\end{equation}
where in the final bound we have made the reasonable assumption that $\eta^{4/3}N^{2/3}$ grows faster than $N^{4/3} / \eta^{4/3}$.

A practical difficulty for simple operator averaging on near-term devices with Pauli operators built from Jordan-Wigner strings is the sensitivity to measurement error on each of the individual qubits in a long Pauli string~\cite{Kandala2017}.  In the plane wave dual basis, one has the advantage that the diagonal operators are always two-local in the Jordan-Wigner representation, thus mitigating this problem. The kinetic operator may be treated in this way by applying the FFFT. However, one might seek to avoid the coherent overhead of applying the FFFT in order to diagonalize the kinetic operator. This would be especially advisable when using a near-term device prone to errors during the FFFT execution. If one were to measure $U + V$ at once but measure $T$ by sampling the $H_\ell$ then the total number of measurements would scale as
\begin{align}
\label{eq:practical_measure}
M & \in {\cal O}\left(\frac{\left \| T \right\|^2 + \text{Var}_{\ket{\Psi}}\left[U + V\right]}{\epsilon^2}\right) 
\in {\cal O}\left(\frac{\left \| T \right\|^2 + \avg{V^2}}{\epsilon^2}\right) \in {\cal O}\left(\frac{N^{10/3}}{\epsilon^2 \, \Omega^{4/3}} + \frac{\eta^4 N^{2/3}}{\epsilon^2 \, \Omega^{2/3}}\right) \in {\cal O}\left(\frac{N^4}{\epsilon^2}\right)\nonumber
\end{align}
where $\| T \|$ is the triangle-inequality upper-bound on the norm of $T$ from \app{operators}. At fixed density $\rho \propto \eta / \Omega \in \Theta(1)$ this quantity also scales as ${\cal O}(N^2 / \mu^2)$ for fixed relative error $\mu$ and $\eta \in \Theta(N)$.

Alternative methods for evaluating the objective function using techniques from phase estimation have been studied in some detail~\cite{Santagati2016}. These methods require a number of initial state preparations that scales quadratically better in $\epsilon$ and measures fewer qubits in the process, which mitigates the impact of measurement error. This quadratic scaling improvement comes at the cost of requiring larger circuit depth, which can make these approaches impractical for existing experimental platforms which are often limited by coherence time. Specifically, if one prepares the state of interest with a unitary $\mathcal{U}$, then it is possible to estimate the expectation value of the energy to a precision $\epsilon$ using $\mathcal{O}(\sum_\ell |W_\ell| / \epsilon)$ applications of $\mathcal{U}$ and $\mathcal{U}^\dagger$.  This implies an asymptotic bound on the cost of energy estimation of
\begin{align}
M_{\mathcal{U}} \in \mathcal{O} \left(\frac{N^{7/3}}{\epsilon \, \Omega^{1/3}} + \frac{N^{5/3}}{\epsilon \, \Omega^{2/3}} \right).
\end{align}
For fixed density, $\eta \in \Theta(N)$, and relative error $\mu$, we can bound this scaling at ${\cal O}(N/\mu)$. As was shown in the original work, this scaling is comparable up to logarithmic factors to the application of iterative phase estimation using the best known Hamiltonian simulation algorithms.  This makes it feasible to use variational approaches to improve state preparation for quantum phase estimation applications on fault-tolerant quantum devices.

\section{Proposal to Simulate Jellium on a near-term device}
\label{sec:sec_three}

In this section, we discuss an experimental proposal for near-term devices based on the advances of \sec{sec_one}. In particular, we focus on the quantum simulation of the homogeneous or uniform electron gas, also known as jellium. We believe that jellium is an attractive system to target with early quantum computers due to its simplicity yet foundational
importance for many areas of physics and materials science. Further, it is naturally compatible with the plane wave and dual basis simulation formalism we have described so far. The widespread use of jellium as a benchmark on which to test new classical simulation methods, as well as continuing unresolved physical questions in the system, positions it as an intriguing arena in which to
contrast quantum and classical simulations.

Jellium is defined as a system of interacting electrons with a uniform electron density $\rho$ and a homogeneous compensating positive background charge, such that the overall system is charge neutral~\cite{giuliani2005quantum}. As a finite realization, we consider a system of $\eta$ electrons in a box of volume $\Omega$ with periodic boundary conditions, where the jellium Hamiltonian becomes exactly \eq{jw_ham} with a constant external potential, i.e. all $\zeta_j = 0$. 
Jellium is of interest in different physical dimensions; both 
 two- and three-dimensional jellium are realized to a good approximation in real materials. For example, two-dimensional jellium is approximated well by electrons confined in  semiconductor wells~\cite{spivak2010colloquium}, while three-dimensional jellium is a model for the valence electron density of alkali metals such as sodium~\cite{brack1993physics}. Historically, the physics of jellium has helped elucidate some of the most basic concepts in condensed matter physics. For example, Wigner's observation that electrons in jellium must crystallize as the electron density is decreased~\cite{wigner1934interaction} was the first example of an interaction driven metal-insulator transition. Later, the ground state physics of jellium in two dimensions in a strong magnetic field
became the canonical setting to understand the quantum Hall effect~\cite{stone1992quantum}. Simulations of jellium also play a central role in computational applications. This is because the energy density of jellium is the starting approximation in density functional calculations, the mostly widely performed calculations in quantum chemistry and materials science. In particular, the local density approximation gives the (exchange-correlation) energy $E_{xc}$ of a material with a generic, non-uniform, electronic density $\rho(r)$, as
\begin{equation}
  E_{xc}[\rho] = \int \rho(r) \epsilon_{xc}^\textrm{UEG}(\rho(r)) \ dr
\end{equation}
where $\epsilon_{xc}^\textrm{UEG}(\rho(r))$ is the (exchange-correlation) energy density of jellium at density $\rho(r)$.
For this reason, the history of density functionals has been tied to improvements in approximate simulations of the jellium energy density~\cite{ceperley1980ground,vosko1980accurate,perdew1992accurate}.

For the above reasons, simulating the properties of jellium with classical methods is a standard classical benchmark. This also argues for using it as a benchmark for quantum simulations, and in this context we briefly outline the current limitations of classical techniques, and the setting in which quantum simulations may be most useful. The phase diagram of jellium is usually discussed in terms of the Wigner-Seitz radius $r_s$, which is related to the density by $4\pi r_s^3 /3 = \Omega / \eta = \rho^{-1}$ in three dimensions. While the ground state of jellium at high densities (metallic, $r_s \sim 1$ Bohr radii per particle) and at very low densities (insulating, $r_s \sim 100$ Bohr radii per particle) is well established,
the precise phase diagram in the low to intermediate density regime is uncertain due to competing electronic and spin phases~\cite{ceperley1980ground,tanatar1989ground,zong2002spin,attaccalite2002correlation,spink2013quantum,drummond2009phase}. 
In the high-density regime, the system is dominated by kinetic energy, and expansion techniques based on perturbation theory perform well~\cite{gell1957correlation,freeman1977coupled}. Outside this density regime, the main simulation tool has been quantum Monte Carlo in the continuum formulation~\cite{ceperley1980ground,tanatar1989ground,zong2002spin,attaccalite2002correlation,spink2013quantum,drummond2009phase}, and more recently, in basis set formulations such as full configuration interaction quantum Monte Carlo (FCIQMC)~\cite{shepherd2012full,shepherd2012investigation} and auxiliary field quantum Monte Carlo (AFQMC)~\cite{wilson1995constrained,motta2015imaginary}.
The latter basis set calculations use plane waves and can be directly compared to quantum simulations in the plane wave dual formulation.
Due to the fermion sign problem, it is difficult to obtain data with acceptable stochastic error with exact quantum Monte Carlo methods (e.g. with released nodes~\cite{ceperley1980ground}, FCIQMC without initiators~\cite{shepherd2012full}, or AFQMC without constrained phase bias~\cite{motta2015imaginary}) for systems with $\eta > 50$. Instead, simulations use a bias to control the sign problem, such as the fixed node approximation. Although much useful information can be extracted in the presence of this bias, the systematic error is hard to estimate, and is thought to be as large as half a percent in the energy~\cite{tanatar1989ground,shepherd2012full}. Unfortunately, this error is on a similar scale to the energy difference between competing phases in the intermediate density regime. We expect quantum simulations, even for modest $\eta \approx 50$ and modest $N \approx 100$, to offer bias free results that cannot currently be obtained by classical techniques; beyond their role in understanding the approximations used in classical methods and in demonstrating ``quantum supremacy', such simulations will provide a new way to resolve the complicated jellium phase diagram in the low  density regime.

In the next part of \sec{sec_three}, we consider how to use the advances introduced in \sec{sec_one} within the specific context of a practical quantum algorithm for jellium simulation on near-term devices. While the Trotter and Taylor algorithms described in \sec{trotter_alg} and \sec{taylor_alg} can be used for ground state simulation, either by simulating adiabatic state preparation \cite{BabbushAQChem} or by projecting to a ground state using quantum phase estimation \cite{Kitaev1995,Aspuru-Guzik2005}, such approaches are likely to require error-correction for their implementation. However, in the case of jellium, a good initial state preparation is extremely simple. This makes variational quantum algorithms for jellium particularly interesting, given their additional suitability for near-term devices.~\cite{Peruzzo2013,McClean2015}.

\subsection{Linear Depth Quantum Variational Algorithm for Planar Architectures}
\label{sec:jellium_alg}

As with all variational algorithms, one prepares an ansatz $\ket{\psi(\vec \theta)}$ for the ground state which is described in terms of parameters $\vec \theta$ which are selected in order to minimize the expectation value of the Hamiltonian, $\bra{\psi(\vec \theta)} H \ket{\psi(\vec \theta)}$. Usually, one prepares $\ket{\psi(\vec \theta)}$ by applying a parameterized quantum circuit to a suitable reference state $\ket{\psi_0}$ so that $\ket{\psi(\vec \theta)} = U(\vec \theta) \ket{\psi_0}$. Thus, the power of a variational algorithm depends on the quality of the reference state $\ket{\psi_0}$ and the structure of the parameterized circuit $U(\vec \theta)$. The reference state is often chosen to be the mean-field solution to the problem. Mean-field solutions to jellium are diagonal in the plane wave basis,
and provide useful starting points for quantum Monte Carlo simulations even at quite low densities~\cite{drummond2009phase}. One can begin quantum simulation in a product state associated with the plane wave basis and then apply the FFFT to obtain the mean-field state of jellium in the dual basis. As shown in \app{fqft}, the FFFT can be implemented with ${\cal O}(N)$ gate depth on a planar lattice.

A variational strategy that is particularly practical for the near-term is based on a low-order Trotter approximation of adiabatic state preparation. This ansatz is related to the quantum approximate optimization algorithm \cite{Farhi2014} and has been shown to perform well in the context of electronic structure \cite{Wecker2015a}. Following the scheme of \cite{Wecker2015a}, the idea is to Trotterize the adiabatic algorithm defined by evolution under
\begin{equation}
\label{eq:schedule}
H\left( \tau \right) = T + U + \tau \, V.
\end{equation}
Thus, the schedule is to start in the ground state of the one-body Hamiltonian and slowly turn on the two-body terms. Note that $H(0) = T$ for jellium, which is the Hamiltonian of a free particle. This choice of schedule further justifies use of $\ket{\psi_0} = \textrm{FFFT} \ket{0}$ as the reference since this makes $\ket{\psi_0}$ an eigenstate of $T$ in the plane wave dual basis. One should choose $\ket{0}$ to have the correct particle number and spin-symmetry to describe the target state as an error-free simulation would conserve these quantum numbers. We use the fact that we can write \eq{schedule} for any molecular Hamiltonian in the Jordan-Wigner transformed plane wave dual basis as
\begin{equation}
H\left(\tau \right) = \textrm{FFFT}^\dagger \left(\sum_{p} \theta_{p} \left(\tau\right) Z_p\right) \textrm{FFFT} + \sum_{p} \theta_{pp} \left(\tau\right) Z_p + \sum_{p\neq q} \theta_{pq}\left(\tau\right) Z_p Z_q
\end{equation}
for scalar values of $\vec \theta$ which should be apparent from \eq{jw_ham}. We can Trotterize the adiabatic evolution as
\begin{equation}
U(\vec \theta) = \prod_{m=1}^M \bunderbrace{\textrm{FFFT}^\dagger \left(\prod_{p} \exp\left[i \, \theta_{p}^m Z_p\right]\right) \textrm{FFFT}}_{U_T(\vec \theta^m)} \bunderbrace{\left(\prod_p \exp\left[i \, \theta_{pp}^m Z_p\right]\right) \!\! \left(\prod_{p\neq q} \exp\left[i \, \theta_{pq}^m Z_p Z_q\right]\right)}_{U_V(\vec \theta^m)}
\quad \quad \!
\vec \theta^m = \frac{\vec \theta \left(\frac{m - 1/2}{M}\right)}{M}
\end{equation}
where $M$ is the total number of repetitions of the Trotter step. As discussed in \sec{trotter_alg}, each of these Trotter steps can be implemented with gate depth ${\cal O}(N)$ on a planar lattice of qubits with no ancilla. Thus, the total gate depth of this ansatz would be ${\cal O}(NM)$. Rather than try to variationally determine all parameters to minimize the final Hamiltonian $H(1)$, the suggestion of \cite{Wecker2015a} is to train the ansatz ``in layers''; i.e., to train the first Trotter step to minimize $H(1/M)$, the second to minimize $H(2/M)$ and so on. The results of \cite{Wecker2015a} suggest that this ansatz may perform well for values of $M$ as low as ten or less. Note that while initial states other than a product state of plane waves may be needed in systems other than jellium, the variational ansatz can be used for any molecule.

Variational algorithms were experimentally demonstrated in \cite{Kandala2017} and \cite{OMalley2016} using superconducting qubit platforms from industrial quantum computing groups which are expected to reach the quantum supremacy threshold in the near-future \cite{Boixo2016}. Such platforms would have qubits connected on a planar lattice and could only implement shallow circuits due to limited coherence. For such an early demonstration, we can make further simplifications to the $M = 1$ variational ansatz. To explain this strategy, we notice that  the expectation value of the Hamiltonian after applying the $M = 1$ variational ansatz can be expressed as
\begin{equation}
\bra{\psi_0} U_V(-\vec \theta) \, \widetilde H (\vec \theta) \, U_V (\vec \theta) \ket{\psi_0}
\quad \quad \quad \quad
\widetilde H(\vec \theta) = U_T(-\vec \theta) \, H \, U_T(\vec \theta)
\end{equation}
where we can see that $\widetilde H(\vec \theta)$ amounts to a local basis transformation on the Hamiltonian $H$. Since this transformation can be applied efficiently with classical post-processing, we see that the ansatz preparation can be simplified to
\begin{equation}
\ket{\psi(\vec \theta)} =  U_V (\vec \theta) \ket{\psi_0} = \left(\prod_{p} \exp\left[i\, \theta_{pp} Z_p \right]\right) \! \! \left( \prod_{p \neq q} \exp\left[i \, \theta_{pq} Z_p Z_q\right] \right) \! \textrm{FFFT} \ket{0}.
\label{eq:minimal_vqe}
\end{equation}
In practice, one would probably also take the rotation angles in the FFFT as variational parameters. Thus, our ``minimal resource variational ansatz'' consists of the FFFT, a high entanglement operation known to produce a good reference, followed by entangling gates between all pairs of qubits and then a single layer of phase gates on each qubit. As a final note, the outer-loop of this variational quantum algorithm will only need to optimize over ${\cal O}(N)$ distinct parameters, as opposed to ${\cal O}(N^2)$ distinct parameters, due to the translational invariance of the jellium system.

In order to resolve distinct phases in low-density jellium, a reasonable target is to obtain energies accurate to a fixed relative error of half of one percent. The minimal variational ansatz of \eq{minimal_vqe} may be sufficient to prepare accurate ground states of jellium in certain parts of the phase diagram; in the high density regime even the mean-field state $\ket{\psi_0}$ is a good initial description. But we also expect that this single Trotter step ansatz will fail to resolve the ground state in more complex regimes. Thus, this proposal immediately raises two unresolved questions: ``how many Trotter steps will we be able to implement on a near-term device?'' and ``how many Trotter steps would be required to surpass all classical methods in the low density regime?''. By compiling all aspects of this procedure to a natively realizable gate set, we should be able to estimate how many Trotter steps would be possible within the limitations of expected coherence times and gate fidelities. This analysis will be the subject of a future paper. However, the second question is more difficult to answer without a quantum device, especially because the radix-2 decimation implementation of the FFFT requires that problem sizes are a power of two. Whether or not quantum supremacy is immediately achievable using this approach to jellium simulation, experimentally studying this ansatz will provide important insights into the effectiveness of Trotter-based variational quantum algorithms for problems of correlated electrons.

\section*{Conclusion}

In this work, we have introduced  efficient techniques that use the plane wave basis and its dual for quantum simulations of electronic structure. The kinetic and potential operators are respectively diagonal in these bases,
providing a Hamiltonian representation with only a quadratic number of terms in basis size. We also described an efficient second quantized fermionic fast Fourier transform to map between the two bases which can be implemented with linear gate depth on a planar lattice of qubits. Using the diagonality of the Hamiltonian components in these dual basis sets, we showed that Trotter steps can be implemented with linear gate depth on a planar lattice. We use these properties to implement time-evolution using Trotter and Taylor series methods with lower overhead than all prior approaches and also reduce the number of measurements required for quantum variational algorithms. Finally, we identified jellium as a concrete electronic structure problem to target on near-term quantum devices.
Jellium is attractive due to its fundamental significance in conceptual and numerical electronic structure theory
and because of its tunability into regimes where classical simulations are currently inadequate.
Exploiting its natural expression in the plane wave basis, 
we proposed a simple quantum variational algorithm which can be executed with low circuit depth on near-term quantum hardware.
Understanding the performance of this algorithm for jellium will provide important
insights into the near-term feasibility of quantum supremacy in realistic problems of electronic structure.

Beyond the confines of this work, we expect that the advances we have described will have
ramifications across many different approaches to quantum simulation. For example, the quadratic reduction in the number of Hamiltonian terms, as well as the lower scaling bounds
on the Hamiltonian norm, will translate generally to decreased complexity in the overhead for perturbative gadgets, or in quantum simulations within the configuration interaction representation. The techniques  may further be used in conjunction with error-corrected simulations. Moving beyond jellium as a physical system,  quantum simulations in the plane wave basis may practically be extended to real materials by incorporating a single-particle pseudopotential, without essential modifications of  the results in this proposal. Ultimately, we believe that our work illustrates the potential of exploring fundamental reformulations of the electronic structure problem in order to reduce the complexity of quantum simulations.

\section*{Acknowledgements}

The authors thank Eddie Farhi, Sergio Boixo, John Martinis, Ian Kivlichan, Craig Gidney, Dominic Berry, Murphy Yuezhen Niu, Pierre-Luc Dallaire-Demers, Peter Love and Matthias Troyer for helpful comments about an early draft. We thank Alireza Shabani for helping to initiate the collaboration between Google and Caltech. The authors thank Wei Sun for contributing code to the open source quantum simulation library OpenFermion (\url{www.openfermion.org}) \cite{openfermion} which was used to verify some equations of this paper.\\

\bibliographystyle{apsrev4-1}
\bibliography{Mendeley}

\appendix

\section{Finite Difference Discretization with $N^2$ Terms}
\label{app:finite_diff}

An alternative to the Galerkin discretization derived from the weak form of the Schroedinger equation is a finite-difference formulation, which is associated with the strong formulation of the differential equation. In the past, many works have explored the use of finite-difference discretizations (either implicity or explicitly) \cite{Boghosian1998,Boghosian1998b,Meyer1996,Meyer1997,Wiesner1996} although never before in a second quantized simulation of an electronic structure system. Still, discretizing these systems in this way is straightforward and follows from this past work. Assuming a uniform partitioning of space, the value of position operators are assigned to a set of grid points with values determined by the position of the grid point. Generalizations to non-uniform grid spacings are also possible.

One might consider this approach analogous to choosing basis functions of the form $\phi_i(r) = \delta(r - r_i)$ in the Galerkin formulation, where $\delta$ is the Dirac delta function and $r_i$ is the location of a grid point, but with several important differences.  In this case, the derivative operators are discretized in an entirely different way, using a finite-difference stencil, rather than integration over such basis functions.  This follows from the discussion of functions with disjoint support in the main text. Moreover, while an inner product in the Galerkin formulation between two functions $\ket{\psi} = \sum_i b_i \ket{\psi_i}$ and $\ket{\phi} = \sum_i c_i \ket{\phi_i}$ has a natural definition induced by the definition of the inner product on the space of $\{\ket{\phi_i}\}$ given by $\braket{\psi}{\phi} = \sum_{i, j} b_i^* c_j \braket{\psi_i}{\phi_j}$, the same is not true in the finite difference scheme.  In this case, one must choose a definition that is consistent with some sensible measure on the space.  

To see how these differences are formulated in practice, we will consider an example.  Assume a uniform volume partition for the system that consists of $N = M^3 $ orbitals which are each indexed by four indices, $x \in \mathbb{Z} \in \left[0, M \right)$, $y \in \mathbb{Z} \in \left[0, M \right)$, $z \in \mathbb{Z} \in \left[0, M \right)$. In this case, the kinetic energy operator may be expressed using a finite-difference 7-point stencil for the Laplacian,
\begin{align}
- \frac{\nabla^2}{2} \phi(x,y,z) = \frac{1}{2\,h^2} \sum_{x,y,z} & \left[ 6\, \phi{(x,y,z)} - \phi(x-1,y,z) - \phi(x + 1,y,z) - \phi(x,y - 1,z) \right. \\
& \left. - \phi(x,y + 1,z)
- \phi(x,y,z-1) - \phi(x,y,z+1) \right]\nonumber
\end{align}
where $h$ is the spacing between grid points.  Central difference stencils of this type, utilizing three points along each axis, have errors that scale as ${\cal O}(h^2)$ in their representation of the derivative operator. In this case, we can see that the kinetic energy operator has exactly $ 7\, N $ terms, and note that other size stencils may be used to reduce the discretization error. The most accurate stencil, which extends across the entire length of the simulated system, would still only have ${\cal O}(N^2)$ terms. An important difference to note between this choice and the Galerkin discretization is that error in expressing the finite-difference formulation of the kinetic energy operator can lead to sub-variational energies in principle. However, this is easily managed in practice with reasonably sized stencils and spatial partitions.

With a uniform grid of points positioned as above and spaced by the same distance $h$ along each axis, we may use the rectangular rule to define an inner product on single particle functions.  In this scheme a single particle function $\ket{\phi}$ is defined by its values at the grid points $\phi(x, y, z, \sigma)$. Note that we will now also consider the spin degree of freedom $\sigma = \{\uparrow, \downarrow\}$. We can define the inner product between two single particle functions $\ket{\psi}$ and $\ket{\phi}$ explicitly as
\begin{align}
\braket{\psi}{\phi} = h^3 \sum_{x, y, z, \sigma} \psi(x, y, z, \sigma)^* \phi(x, y, z, \sigma).
\end{align}
and label individual points $\ket{\phi_{x,y,z,\sigma}}$ such that $\braket{\phi_{x,y,z,\sigma}}{\phi_{x',y',z',\sigma'}} = \delta_{x x'} \delta_{y y'} \delta_{z z'} \delta_{\sigma \sigma'}$ and $\braket{x,y,z,\sigma}{\phi_{x',y',z',\sigma'}} = \phi(x, y, z, \sigma)$. With these definitions of the kinetic energy and inner product, we can express the second quantized coefficients for one-body operators in the following way. If we define compound indices $p=(x_p,y_p,z_p,\sigma_p)$ with corresponding Kronecker delta functions $\delta_{pq} = \delta_{x_p x_q} \delta_{y_p y_q} \delta_{z_p z_q} \delta_{\sigma_p \sigma_q}$
\begin{align}
h_{pq} = \frac{h}{2}\left(6 \, \delta_{pq} - \sum_{\alpha \in \{x, y, z\}} \left( \delta_{pq+_\alpha} + \delta_{pq-_\alpha} \right) \right) + h^3 \, U(p) \, \delta_{pq}
\end{align}
where we have used the shorthand notation $q+_\alpha$ to indicate shifting the $\alpha$ axis by 1 lattice point.
We define the standard number operator as $n_{x,y,z,\sigma}= a_{x,y,z,\sigma}^\dagger a_{x,y,z,\sigma}$. Similarly, the coefficients of the two-body potential become
\begin{align}
h_{pqrs} = \delta_{ps} \delta_{qr} \left[\frac{h^3}{|p_{x,y,z} - q_{x,y,z}|}(1 - \delta_{pq+_\sigma} - \delta_{pq-_\sigma}) + \lambda (\delta_{pq+_\sigma} + \delta_{pq-_\sigma}) \right],
\end{align}
where we have separated the same-point repulsion into a second term characterized by $\lambda$.  It follows that the two-body part of the operator may also be written as
\begin{equation}
V = \lambda \sum_{x,y,z} n_{(x,y,z, \uparrow)} n_{(x,y,z, \downarrow)}
+ \frac{h^3}{2}\sum_{\substack{(x, y, z) \neq (x', y', z') \\ \sigma, \sigma'}} \frac{n_{(x,y,z,\sigma)} n_{(x',y',z',\sigma')}}{\sqrt{(x - x')^2 + (y - y')^2 + (z - z')^2 }}
\end{equation}
where $\lambda$ scales the repulsive interaction between electrons of opposite spin when they occupy the same spatial orbital. We can see that there are $N / 2$ terms on the left and $N (N - 1) / 2$ unique terms on the right, for a total of $N^2 / 2$ terms in the two-body potential. While the exact value of $\lambda$ does not matter in the continuum limit, the chosen value determines the convergence of basis set discretization error. The approximation we advocate here is to treat $\lambda$ as the mean repulsion between a uniform charge distribution in the cell, i.e.
\begin{align}
\lambda & =  \frac{1}{2} \int \frac{dx_1 \, dx_2 \, dy_1 \, dy_2 \, dz_1 \, dz_2 }{\sqrt{\left(x_1 - x_2 \right)^2 + \left(y_1 - y_2\right)^2 + \left(z_1 - z_2\right)^2}} \\
& = \frac{1}{h} \left(\frac{1 + \sqrt{2} - 2 \sqrt{3}}{5} - \frac{\pi}{3} + \log \left[\left(1 + \sqrt{2}\right)\left(2 + \sqrt{3}\right)\right]\right)
\approx \frac{0.941156}{h}.\nonumber
\end{align}

Note that the analytical evaluation of this integral is provided as the main result of \cite{Ciftja2011}. Note further that one could also choose to evaluate the long-range Coulomb interaction between orbitals $p$ and $q$ using integrals which assume uniform charge density within the cell. This choice may lead to slightly different convergence behavior but the results will certainly agree in the continuum limit. Putting these results together, we arrive at the second quantized position space Hamiltonian in a finite-difference representation,
\begin{align}
H =  \, &  \, \frac{h}{2} \sum_{x,y,z,\sigma} \left[ 6 \, n_{(x,y,z,\sigma)}
- a_{(x-1 ,y,z,\sigma)}^\dagger a_{(x,y,z,\sigma)} - a_{(x+1,y,z,\sigma)}^\dagger a_{(x,y,z,\sigma)} \right. \\
& \left.- \, a_{(x,y-1,z,\sigma)}^\dagger a_{(x,y,z,\sigma)} - a_{(x,y+1,z,\sigma)}^\dagger a_{(x,y,z,\sigma)}
- a_{(x,y,z-1,\sigma)}^\dagger a_{(x,y,z,\sigma)} - a_{(x,y,z+1 ,\sigma)}^\dagger a_{(x,y,z,\sigma)}\right] \nonumber\\
& + \frac{h^3}{2} \sum_{\substack{(x, y, z) \neq (x', y', z') \\ \sigma, \sigma'}} \frac{n_{(x,y,z,\sigma)} n_{(x',y',z',\sigma')}}{\sqrt{(x - x')^2 + (y - y')^2 + (z - z')^2 }} \notag \\
&+ h^3 \sum_{x,y,z,\sigma} n_{(x,y,z,\sigma)} U(x,y,z,\sigma)
+ \frac{0.941156}{h} \sum_{x,y,z} n_{(x,y,z, \uparrow)} n_{(x,y,z, \downarrow)} .\nonumber
 \end{align}
which implicitly defines both the one-body and two-body coefficients, $h_{pq}$ and $h_{pqrs}$ for the second quantized Hamiltonian, noting that some normal ordering may be required to bring it to its final form.  This Hamiltonian contains strictly ${\cal O}(N^2)$ terms, as desired. While we do not use this result for any of the algorithms of this paper, understanding the finite-difference formulation on a grid is helpful to appreciate differences with the plane wave dual basis. Furthermore, it is possible that this form of the Hamiltonian has advantages that could make it easier to simulate in the context of future quantum algorithms, perhaps based on 1-sparse decompositions of the finite-difference stencil.

\section{Electronic Structure Hamiltonian in Plane Wave Basis}
\label{app:plane_waves}

In this section we will review analytical forms for the electronic structure Hamiltonian in a basis of plane waves of the following form in three dimensions,
\begin{equation}
\varphi_\nu \left(r\right) = \sqrt{\frac{1}{\Omega}} e^{i \, k_\nu \cdot r}
\quad \quad \quad
k_\nu = \frac{2 \pi \nu}{\Omega^{1/3}}
\quad \quad \quad 
\nu \in \left[-N^{1/3}, N^{1/3}\right]^3 \subset \mathbb{Z}^3.
\end{equation}
The length scale of our basis is parameterized by the cell volume $\Omega$.

The kinetic energy operator is a one-body operator. The coefficients of the kinetic energy operator $T$ are
\begin{equation}
\int_\Omega dr^3 \, \varphi_p^*\left(r\right) \left(\frac{-\nabla^2}{2}\right) \varphi_q\left(r\right) =  \frac{p^2}{2}\delta\left(p,q\right).
\end{equation}
Thus,
\begin{equation}
\label{eq:pw_t}
T = \frac{1}{2} \sum_{\nu, \sigma} k_\nu^2 \, c_{\nu,\sigma}^\dagger c_{\nu,\sigma}
\end{equation}
where $c^\dagger_\nu$ and $c_\nu$ are canonical fermionic raising and lowering operators and $\sigma \in \{\uparrow, \downarrow\}$ represents spin. Clearly, this operator is diagonal since plane waves are eigenstates of the momentum operator.

When working with plane waves it is convenient to define the Fourier transform of the Coulomb potential,
\begin{equation}
V_\nu =  \frac{1}{\Omega} \int_\Omega dr^3 \, V(r) \, e^{-i \, k_\nu \cdot r} = \frac{4\pi}{k_\nu^2 \, \Omega}
\end{equation}
and the inverse of this Fourier transform, a solution to Poisson's equation with periodic boundary conditions:
\begin{equation}
V\left(r\right) = \sum_\nu V_\nu \, e^{i \, k_\nu \cdot r}.
\label{eq:period_col}
\end{equation}
Note that there would appear to be a singularity in this periodized representation of the Coulomb operator when $k_\nu = 0$; however, whenever treating a charge-neutral system the singularities from interactions with the positive and negative charges cancel to contribute only a finite constant which depends on the cell shape. This factor can be computed using an Ewald sum, shown explicitly in Appendix F of \cite{Martin2004}.

The external potential arising from interactions with nuclei can be expressed as
\begin{equation}
U \left(r\right) = -\sum_{j} \zeta_j V\left(r - R_j\right) = -\sum_{j,\nu} \zeta_j \, V_\nu \, e^{i \, k_\nu \cdot \left(r - R_j\right)} 
\end{equation} 
where nuclei $j$ has position $R_j$ and atomic number $\zeta_j$. With this we compute the external potential coefficients as
\begin{align}
& \int_\Omega dr^3 \, \varphi_p^*\left(r\right) U\left(r\right) \varphi_q\left(r\right) =
\int_\Omega dr^3 \, \varphi_p^*\left(r\right) \left(- \sum_{j,\nu} \zeta_j \, V_\nu \, e^{i \, k_\nu \cdot \left(r - R_j\right)} \right) \varphi_q\left(r\right)\\
& =- \sum_{j,\nu} \zeta_j \, V_\nu \, e^{-i \, k_\nu \cdot R_j} \int_{\Omega} dr^3 \, \varphi_p^*\left(r\right) e^{i \, k_\nu \cdot r} \varphi_q\left(r\right)
 = - \sum_{j} \zeta_j \, V_{p-q} \, e^{-i \, k_{p-q} \cdot R_j}
=  -\frac{4 \pi}{\Omega} \sum_{j} \zeta_j \frac{e^{i \, k_{q-p} \cdot R_j}}{k_{p-q}^2}.\nonumber
\end{align}
Accordingly, we can write the external potential operator as
\begin{equation}
\label{eq:pw_u}
U = -\frac{4 \pi}{\Omega} \sum_{\substack{p \neq q \\ j,\sigma}} \left(\zeta_j \frac{e^{i \, k_{q-p} \cdot R_j}}{k_{p-q}^2}\right) c^\dagger_{p, \sigma} c_{q, \sigma}
\end{equation}
where the condition $p \neq q$ is equivalent to dropping the zero momenta mode of the external potential which, as explained earlier, cancels with the zero mode of the electron-electron interaction. As explained in the main text, we choose to alias the momenta modes so that, in this case, $k_{p-q}$ is always contained within the original set of plane waves. This introduces a slight deviation from the Galerkin formulation and corresponds to evaluating matrix elements by  $N$ evenly spaced samples on a real space grid.
Doubling the quadrature spacing
would yield an exact evaluation but without the aliasing (dualling) approximation we would not obtain the convenient exactly diagonal form of the potential matrix elements in the the dual basis that we rely upon.

The two-electron interaction coefficients are obtained from the integrals,
\begin{align}
& \int_\Omega dr_1^3 \, dr_2^3 \, \varphi_p^*\left(r_1\right) \varphi_q^*\left(r_2\right) V\left(r_1 - r_2\right) \varphi_r\left(r_2\right) \varphi_s\left(r_1\right) \\
& = \int_\Omega dr_1^3 \, dr_2^3 \, \varphi_p^*\left(r_1\right) \varphi_q^*\left(r_2\right) \left(\sum_\nu V_\nu \, e^{i \, k_\nu \cdot \left(r_1 - r_2\right)} \right) \varphi_r\left(r_2\right) \varphi_s\left(r_1\right)\nonumber\\
& = \sum_\nu V_\nu \left(\int_\Omega dr_1^3 \,  \varphi_p^*\left(r_1\right) e^{i \, k_\nu \cdot r_1} \varphi_s\left(r_1\right) \right)\left(\int_\Omega  dr_2^3 \,  \varphi_q^*\left(r_2\right) e^{-i \, k_\nu \cdot r_2} \varphi_r\left(r_2\right)\right) \nonumber\\
& = \sum_\nu V_\nu \,\delta\left(\nu, p - s\right) \delta\left(\nu, r - q\right)  = \frac{4 \pi}{\Omega} \sum_\nu \frac{ \delta\left(p - s, r - q\right)}{k_\nu^2}\nonumber.
\end{align}
The condition that $\nu = p - s = r - q$ arises from conservation of momentum. From this we arrive at $r = q + \nu$ and $s = p - \nu$, which implies the final form of the two-electron term in momentum space is
\begin{equation}
V = \frac{2 \pi}{\Omega} \sum_{\substack{(p, \sigma) \neq (q, \sigma') \\ \nu \neq 0}} \frac{c^\dagger_{p,\sigma} c_{q,\sigma'}^\dagger c_{q + \nu,\sigma'} c_{p - \nu,\sigma}}{k_\nu^2}
 \end{equation}
 where we can see that this summation satisfies momentum conservation since $\nu = p - (p - \nu) = (q + \nu) -  q$. Thus, the total expression for $H = T + U + V$ (up to a constant shift that depends on the unit cell shape) is given by
 \begin{equation}
 H = \frac{1}{2} \sum_{p, \sigma} k_p^2 \, c_{p,\sigma}^\dagger c_{p,\sigma} - \frac{4 \pi}{\Omega} \sum_{\substack{p \neq q \\ j,\sigma}} \left(\zeta_j \frac{e^{i \, k_{q-p} \cdot R_j}}{k_{p-q}^2}\right) c^\dagger_{p, \sigma} c_{q, \sigma} + \frac{2 \pi}{\Omega} \sum_{\substack{(p, \sigma) \neq (q, \sigma') \\ \nu \neq 0}} \frac{c^\dagger_{p,\sigma} c_{q,\sigma'}^\dagger c_{q + \nu,\sigma'} c_{p - \nu,\sigma}}{k_\nu^2}.
 \label{eq:pw_ham}
 \end{equation}

\section{Electronic Structure Hamiltonian in Plane Wave Dual Basis}
\label{app:plane_wave_dual}

In the prior section we derived a closed form for the molecular electronic structure Hamiltonian in the plane wave basis. We now translate that Hamiltonian into the plane wave dual basis via unitary discrete Fourier transform. The unitary discrete Fourier transform of the plane wave basis is computed in each dimension separately as
\begin{equation}
\phi_{p_x}\left(x\right) = \sqrt{\frac{1}{N^{1/3}}} \sum_{\nu_x} \left(e^{-2\pi i \, p_x / N^{1/3}}\right)^{\nu_x} \varphi_{\nu_x} \left(x\right)
= \frac{1}{(\Omega \, N)^{1/6}} \sum_{\nu_x} \exp\left[2 \pi i \left(\frac{x}{\Omega^{1/3}}-\frac{p_x}{N^{1/3}} \right) \right]^{\nu_x}
\label{eq:udft}
\end{equation}
where $\phi_{p_x}(x)$ is the $x$-component of the plane wave dual basis function $\phi_p(r) = \phi_{p_x}(x)\phi_{p_y}(y)\phi_{p_z}(z)$, $\varphi_{\nu_x}(x)$ is the $x$-component of the plane wave basis function $\varphi_\nu(r) = \varphi_{\nu_x}(x)\varphi_{\nu_y}(y)\varphi_{\nu_z}(z)$, $\nu = (\nu_x, \nu_y, \nu_z)$ and $r = (x, y, z)$. As the expression for $\phi_{p_x} (x)$ in \eq{udft} takes the form of a geometric series, we can find the following closed form,
\begin{equation}
\phi_{p}\left(r\right) = \sqrt{\frac{1}{\Omega \, N}} \left(\frac{\sin \left[\pi \, p_x - \frac{\pi N^{1/3} x}{\Omega^{1/3}} \right]}{\sin\left[\frac{\pi \, p_x}{N^{1/3}} - \frac{\pi \, x}{\Omega^{1/3}}\right]}\right)
\left(\frac{\sin \left[\pi \, p_y - \frac{\pi N^{1/3} y}{\Omega^{1/3}} \right]}{\sin\left[\frac{\pi \, p_y}{N^{1/3}} - \frac{\pi \, y}{\Omega^{1/3}}\right]}\right)
\left(\frac{\sin \left[\pi \,p_z - \frac{\pi N^{1/3} z}{\Omega^{1/3}} \right]}{\sin\left[\frac{\pi \, p_z}{N^{1/3}} - \frac{\pi \, z}{\Omega^{1/3}}\right]}\right)
\label{eq:pwd_functions}
\end{equation}
which is a smooth approximation to a grid with lattice sites at the locations $r_p = p\, (\Omega / N)^{1/3}$.

Basis functions of the above form (which can be conveniently labeled by the real-space coordinates of their centers) are
commonly used in quantum dynamics simulations under the name of discrete variable representations (DVR) ~\cite{lill1982discrete,shizgal1984discrete,dvrreview,colbert1992novel,Jones2016}. The sinc DVR, introduced in \cite{colbert1992novel} is closely related to the plane wave dual basis. As seen from \eq{udft}, the plane wave dual basis is obtained as a sum over unit weighted plane waves with reciprocal lattice momenta up to a maximum cutoff momentum. The sinc DVR is obtained as a {\it continuous} integral over unit weight plane waves up to the maximum cutoff momentum. One of the primary weaknesses of the sinc DVR basis is the need to approximate the kinetic energy operator when using a finite number of sinc functions. This is removed in the plane wave dual basis, as the kinetic energy operator is represented exactly.

Rather than compute the integrals over these basis functions by quadrature, it is more straightforward to Fourier transform \eq{pw_ham} in order to obtain a representation of the electronic structure Hamiltonian in the plane wave dual basis. Accordingly, we define raising and lowering operators in the plane wave basis as the Fourier transform of their plane wave dual counterparts,
\begin{equation}
\label{eq:ladder_def}
c^\dagger_\nu =  \sqrt{\frac{1}{N}}  \sum_{p} a^\dagger_{p} e^{- i \,k_\nu \cdot r_p}
\quad \quad \quad \quad
c_\nu =  \sqrt{\frac{1}{N}}  \sum_{p} a_{p} e^{i \, k_\nu \cdot r_p}.
\end{equation}
Using these relations we can write the kinetic energy operator of the previous section in the dual space as
\begin{align}
\label{eq:pwd_t}
T = \, &  \,  \frac{1}{2} \sum_{\nu, \sigma} k_\nu^2 \,  c_{\nu, \sigma}^\dagger c_{\nu, \sigma}
=  \frac{1}{2} \sum_{\nu, \sigma} k_\nu^2 \,  \left(\sqrt{\frac{1}{N}} \sum_{p} a_{p, \sigma}^\dagger e^{- i \, k_\nu \cdot r_p}\right) \left(\sqrt{\frac{1}{N}} \sum_{q} a_{q,\sigma} e^{i \, k_\nu \cdot r_q}\right)\\
=  \, & \,  \frac{1}{2 \, N} \sum_{p, q}  \left(\sum_{\nu, \sigma} k_\nu^2 \, e^{i \, k_\nu \cdot \left(r_q -r_p \right)} \right)a^\dagger_{p, \sigma} a_{q,\sigma} 
= \frac{1}{2\, N} \sum_{\nu, p, q, \sigma} k_\nu^2 \, \cos \left[k_\nu \cdot r_{q - p} \right] a^\dagger_{p, \sigma} a_{q,\sigma}. \nonumber
\end{align}
We can transform the external potential in a similar fashion
\begin{align}
U & = -\sum_{\substack{p \neq q \\ j,\sigma}}\zeta_j \, V_{p-q} \exp\left[i \, k_{q-p} \cdot R_j\right] c^\dagger_{p, \sigma} c_{q, \sigma}\\
& =  - \sum_{\substack{p \neq q \\ j,\sigma}}\zeta_j \, V_{p-q}  \exp\left[i \, k_{q-p} \cdot R_j\right]
\left(\sqrt{\frac{1}{N}} \sum_{p'} a_{p', \sigma}^\dagger \exp\left[- i \, k_p \cdot r_{p'}\right]\right)
\left(\sqrt{\frac{1}{N}} \sum_{q'} a_{q',\sigma} \exp\left[i \, k_q \cdot r_{q'}\right]\right)\nonumber\\
& =  -\frac{1}{N} \sum_{\substack{p \neq q \\ j,\sigma}} \zeta_j \, V_{p-q}  \exp\left[i \, k_{q-p} \cdot R_j\right]
 \sum_{p', q'} a_{p', \sigma}^\dagger a_{q',\sigma} \exp\left[i \, k_q \cdot r_{q'- p'} \right] \exp\left[- i \,k_{p-q} \cdot r_{p'}\right]\nonumber\\
& =  -\frac{1}{N} \sum_{\substack{p', q'}} \sum_{\substack{p \neq q \\ j, \sigma}} \zeta_j \, V_{p-q} \exp\left[i \, k_{q-p} \cdot \left(R_j - r_{p'}\right)\right]
\left(a_{p', \sigma}^\dagger a_{q',\sigma} \exp\left[i \, k_q \cdot r_{q'-p'}\right] \right) \nonumber.
\end{align}
Recognizing that $p - q$ spans the full set of momentum vectors in our system due to aliasing (dualling), we can replace the sum over $p \neq q$ and the indices $p-q$ and $q$ with a sum over $\nu \neq 0$ and $q \neq 0$. This leads to a DVR-like representation
with diagonal  potential operators.
We find
\begin{align}
\label{eq:pwd_u}
U & =  -\frac{1}{N} \sum_{\substack{p', q'}}  \left(\sum_{\substack{\nu \neq 0 \\ j}} \zeta_j \, V_{\nu} \exp\left[i \, k_{\nu} \cdot \left(R_j - r_{p'}\right)\right]\right)
\left( \sum_{q\neq 0, \sigma} a_{p', \sigma}^\dagger a_{q',\sigma} \exp\left[i \, k_q \cdot r_{q'-p'}\right] \right)\\
& =  - \sum_{p,\sigma} \left(\sum_{\substack{\nu \neq 0 \\ j}} \zeta_j \, V_{\nu} \exp\left[i \, k_{\nu} \cdot \left(R_j - r_{p'}\right)\right]\right) n_{p, \sigma}
= - \frac{4 \pi}{\Omega} \sum_{\substack{p,\sigma \\ j, \nu\neq 0}} \frac{\zeta_j \, \cos\left[k_{\nu} \cdot \left(R_j - r_{p}\right)\right]}{k_\nu^2} n_{p, \sigma} \nonumber
\end{align}
where we have used the fact that the summation grouped on the right side of the first equation is equal to zero unless $p' = q'$. This is because the negative modes of $k_q$ will have exactly the opposite phase as the positive modes of $k_q$. This leads to the
diagonal form of the final expression.

Finally, we turn our attention towards transforming the two-electron operator. The following relations are helpful,
\begin{equation}
\sum_{p} c^\dagger_p c_p = \sum_p a^\dagger_{p} a_p
\quad \quad \quad \quad
\sum_{p} c^\dagger_p c_{p \pm q} = \sum_p a^\dagger_{p} a_{p} \, e^{\mp i \, k_q \cdot r_p}
\end{equation}
where the first relation comes from conservation of particle number and the second relation is the Fourier convolution theorem. We can compute the interaction term in the plane wave dual basis as
\begin{align}
\label{eq:pwd_v}
V &=  \frac{2 \pi}{\Omega} \sum_{\substack{(p, \sigma) \neq (q, \sigma') \\ \nu \neq 0}} \frac{c^\dagger_{p,\sigma} c_{q,\sigma'}^\dagger c_{q + \nu,\sigma'} c_{p - \nu,\sigma}}{k_\nu^2}
 = \frac{2 \pi}{\Omega} \sum_{\nu \neq 0} \frac{1}{k_\nu^2}\left( \sum_{\substack{p,q \\ \sigma, \sigma'}} c^\dagger_{p,\sigma} c_{p - \nu,\sigma} c^\dagger_{q,\sigma'} c_{q + \nu,\sigma'}  - \sum_{p,\sigma} c_{p,\sigma}^\dagger c_{p,\sigma}\right)\\
 & = \frac{2 \pi}{\Omega} \sum_{\nu \neq 0} \frac{1}{k_\nu^2}\left[ \left( \sum_{p,\sigma} c^\dagger_{p,\sigma} c_{p - \nu,\sigma}\right)\left( \sum_{q, \sigma'} c^\dagger_{q,\sigma'} c_{q + \nu,\sigma'} \right)  - \sum_{p,\sigma} c^\dagger_{p,\sigma} c_{p,\sigma} \right] \nonumber\\
 & = \frac{2 \pi}{\Omega} \sum_{\nu \neq 0} \frac{1}{k_\nu^2}\left[ \left( \sum_{p,\sigma} a^\dagger_{p,\sigma} a_{p,\sigma} \, e^{i \, k_\nu \cdot r_p}\right)\left( \sum_{q, \sigma'} a^\dagger_{q,\sigma'} a_{q,\sigma'} \, e^{-i \, k_\nu \cdot r_q} \right)  - \sum_{p,\sigma} a^\dagger_{p,\sigma} a_{p,\sigma} \right] \nonumber\\
 & =  \frac{2 \pi}{\Omega} \sum_{\nu \neq 0} \frac{1}{k_\nu^2} \left( \sum_{\substack{p,q \\ \sigma, \sigma'}} e^{i \, k_\nu \cdot r_{p -q}}  a_{p, \sigma}^\dagger a_{p, \sigma} a_{q, \sigma'}^\dagger a_{q,\sigma'} - \sum_{p, \sigma} a_{p, \sigma}^\dagger a_{p,\sigma} \right) 
 =  \frac{2 \pi}{\Omega } \sum_{\substack{(p, \sigma) \neq (q, \sigma') \\ \nu \neq 0}} \frac{\cos \left[k_\nu \cdot r_{p-q}\right]}{k_\nu^2} \, n_{p, \sigma} n_{q, \sigma'}.\nonumber
 \end{align}
Putting these results together, we find the final expression for the total Hamiltonian in the plane wave dual basis,
\begin{align}
\label{eq:pwd_ham}
H & = \frac{1}{2\, N} \!\!\!\sum_{\nu, p, q, \sigma} \!\! k_\nu^2 \cos \left[k_\nu \cdot r_{q - p} \right] a^\dagger_{p, \sigma} a_{q,\sigma}
- \frac{4 \pi}{\Omega} \sum_{\substack{p,\sigma \\ j, \nu\neq 0}} \frac{\zeta_j \, \cos\left[ k_{\nu} \cdot \left(R_j - r_{p}\right)\right]}{k_\nu^2} n_{p, \sigma} +
\frac{2 \pi}{\Omega } \!\!\!\!\! \sum_{\substack{(p, \sigma) \neq (q, \sigma') \\ \nu \neq 0}}\!\!\!\!\!\!\!\! \frac{\cos \left[k_\nu \cdot r_{p-q}\right]}{k_\nu^2} \, n_{p, \sigma} n_{q, \sigma'}.
\end{align}
As we can see, there are only ${\cal O}(N^2)$ terms.

\section{Plane Wave Dual Basis Hamiltonian Mapped to Qubits}
\label{app:qubit_ham}

Whereas fermions are indistinguishable, anti-symmetric particles, qubits are distinguishable and have no special symmetries. Thus, in order to encode a fermionic system on a quantum computer in second quantization one must map the operator algebra of fermions to the operator algebra of qubits. The algebra of fermions is defined by the canonical fermionic anti-commutation relations,
\begin{equation}
\left\{ a^\dagger_p, a^\dagger_q \right\} = \left\{ a_p, a_q \right\} = 0 \quad \quad \quad \quad \left\{ a^\dagger_p, a_q \right\} = \delta_{pq}.
\label{eq:commutation}
\end{equation}
The oldest (and simplest) method which accomplishes this is the Jordan-Wigner transformation \cite{Jordan1928}. A significantly more complicated method is known as the Bravyi-Kitaev transformation \cite{Bravyi2002,Seeley2012,Tranter2015}. The Bravyi-Kitaev transform yields operators that are $\log N$ local as opposed to the Jordan-Wigner transformation, which is $N$ local, in general. More recently, there has been work on generalizing these transformations \cite{Whitfield2016,Havlicek2017,Bravyi2017}. Understanding the structure of these transformations is important for compiling circuits efficiently. However, for our purposes, the locality overhead is not necessarily detrimental in terms of gate depth (although it does effect gate count on a fully connected architecture) and so we analyze the Jordan-Wigner transformation for the sake of simplicity. The Jordan-Wigner transformation consists of the following mapping,
\begin{equation}
a^\dagger_p \mapsto \frac{1}{2} \left(X_p - i \, Y_p\right) \bigotimes_{\ell = 0}^{p - 1} Z_{p - \ell}
\quad \quad \quad \quad
a_p \mapsto \frac{1}{2} \left(X_p + i \, Y_p\right) \bigotimes_{\ell = 0}^{p - 1} Z_{p - \ell}
\label{eq:jw}
\end{equation}
where $X_p$, $Y_p$ and $Z_p$ represent Pauli operators acting on tensor factor $p$. By inspection, one can confirm that the mapping of \eq{jw} reproduces the algebra of \eq{commutation}.

To actually apply the Jordan-Wigner transformation, one must map the fermionic orbitals specified in \eq{pwd_ham} by the indices $(p, \sigma)$ to a single qubit index; e.g.,
\begin{equation}
\left(p, \sigma\right) \mapsto \frac{1 - \sigma}{2} + 2 \left(p_x + p_y\, N^{1/3} + p_z \, N^{2/3}\right)
\quad \quad \quad
\sigma \in \left\{-1, 1\right\}.
\end{equation}
The Jordan-Wigner transformation is particularly simple for the plane wave dual basis molecular Hamiltonian. Applying \eq{jw} to operators that appear in \eq{pwd_ham}, we find that
\begin{align}
\label{eq:jw_ops}
n_p & \mapsto \frac{1}{2} \left(I - Z_p\right)\\
n_{p} n_{q} & \mapsto  \frac{1}{4} \left(I + Z_{p} Z_{q} - 
Z_{p} - Z_{q} \right)\nonumber\\
a^\dagger_{p} a_{q} + a^\dagger_{q} a_{p} & \mapsto \frac{1}{2} \left(X_{p} Z_{p + 1} \cdots Z_{q - 1} X_{q} + Y_{p} Z_{p + 1} \cdots Z_{q - 1} Y_{q}\right) \nonumber.
\end{align}
We note that all of the qubit terms that come out of $n_p n_q$ are diagonal (and thus commute). From \eq{jw_ops} we can write the position space second quantized Jordan-Wigner transformed qubit Hamiltonian as
\begin{align}
H & = \frac{1}{4\, N} \sum_{\nu, p, q, \sigma} k_\nu^2 \cos \left[k_\nu \cdot r_{q - p} \right] \left(X_{p,\sigma} Z_{p + 1,\sigma} \cdots Z_{q - 1,\sigma} X_{q,\sigma} + Y_{p,\sigma} Z_{p + 1,\sigma} \cdots Z_{q - 1,\sigma} Y_{q,\sigma} \right)\\
& - \frac{2 \pi}{\Omega} \sum_{\substack{p,\sigma \\ j, \nu\neq 0}} \frac{\zeta_j \, \cos\left[ k_{\nu} \cdot \left(R_j - r_{p}\right)\right] }{k_\nu^2}\left(I - Z_{p,\sigma}\right) + \frac{\pi}{2\,\Omega } \sum_{\substack{(p, \sigma) \neq (q, \sigma') \\ \nu \neq 0}} \frac{\cos \left[k_\nu \cdot r_{p-q}\right]}{k_\nu^2} \left(I + Z_{p,\sigma} Z_{q,\sigma'} - 
Z_{p,\sigma} - Z_{q,\sigma'} \right).\nonumber
\end{align}
Expanding these terms and recollecting the qubit operator coefficients we find
\begin{align}
\label{eq:qubit_ham}
H & =  \sum_{\substack{p, \sigma \\ \nu \neq 0}}\left(\frac{\pi}{\Omega \, k_\nu^2}  - \frac{k_\nu^2}{4 \, N} + \frac{2\pi}{\Omega} \sum_{j}\zeta_j \frac{\cos\left[ k_{\nu} \cdot \left(R_j - r_{p}\right)\right] }{k_\nu^2}\right) Z_{p,\sigma}
+ \frac{\pi}{2\,\Omega } \sum_{\substack{(p, \sigma) \neq (q, \sigma') \\ \nu \neq 0}} \frac{\cos \left[k_\nu \cdot r_{p-q}\right]}{k_\nu^2} Z_{p,\sigma} Z_{q,\sigma'}\\
& + \frac{1}{4\, N} \sum_{\substack{p \neq q \\ \nu, \sigma}} k_\nu^2 \cos \left[k_\nu \cdot r_{q - p} \right] \left(X_{p,\sigma} Z_{p + 1,\sigma} \cdots Z_{q - 1,\sigma} X_{q,\sigma} + Y_{p,\sigma} Z_{p + 1,\sigma} \cdots Z_{q - 1,\sigma} Y_{q,\sigma} \right)
+ \sum_{\nu \neq 0} \left(\frac{k_\nu^2}{2}- \frac{\pi \, N}{\Omega \, k_\nu^2} \right) I\nonumber.
\end{align}

\section{Comparing Discretization Error in Fourier and Gaussian Bases}
\label{app:basis_errors}

In this section we discuss convergence of basis set discretization errors in both plane wave and Gaussian bases. The basis set discretization error is defined with respect to the ground state energy in the continuum basis $(N = \infty)$ as
\begin{equation}
\Delta E = \left | \min_{\psi} \bra{\psi_\infty} H \ket{\psi_\infty} - \min_{\psi} \bra{\psi_N} H \ket{\psi_N}\right|
\end{equation}
where $\ket{\psi_N}$ is a wavefunction limited to the support of Slater determinants with up to $N$ single-particle basis functions (in our context those functions are either plane waves or Gaussian orbitals). Throughout this work, but especially in \sec{sec_two} and \tab{scalings}, we directly compare the asymptotic scaling of algorithms using a plane wave basis and algorithms using a Gaussian orbital basis. We compare these scalings in terms of the same parameter, ``$N$'', which represents the number of plane waves for some algorithms and the number of Gaussian orbitals for others. In order for such comparisons to be valid, we need to establish that the number of plane waves required for a particular calculation is asymptotically equivalent to the number of Gaussian orbitals required for the same calculation.

In \app{intrinsic_error} we review results from the literature which establish that $\Delta E \in {\cal O} (1/N)$ regardless of the detailed form of the single-particle basis functions. This has been established by many numerical studies over the years and also proved up to second-order in perturbation theory for Gaussians in \cite{kutzelnigg1992rates} and for plane waves in \cite{shepherd2012convergence}. Although most of the results we describe are standard, we gather them here for completeness and also provide an intuitive explanation for this phenomenon based on simple arguments from approximation theory.

In \app{image_error}, we describe how a plane wave basis calculation is done in practice for systems with reduced periodicity, e.g.~for molecules or surfaces. Using the methodology of \cite{fusti2002accurate}, we show that one can exponentially suppress errors arising from the fictitious periodic image charges that occur when using plane waves to describe non-periodic systems. Taken together, these results allows us to directly compare the asymptotic scalings of algorithms using a plane wave basis with the asymptotic scalings of algorithms using a Gaussian orbital basis, even for non-periodic systems such as single-molecules. As the dual basis is a unitary rotation of the plane wave basis, all results presented here also hold equally for the dual basis.

\subsection{Scaling of Intrinsic Discretization Error}
\label{app:intrinsic_error}

We first present an intuitive argument for the basic result and then discuss several earlier works which establish the result more rigorously. As is well known from approximation theory and Fourier analysis, the rate of convergence of a basis expansion of a function is governed by its smoothness. For example, for an infinitely differentiable function (in any dimension), the asymptotic Fourier amplitudes from a Fourier transform decay exponentially in magnitude with respect to the number of Fourier modes, and thus approximating the function with a cutoff in the Fourier series (e.g.~a finite basis) incurs an exponentially small error with the size of the basis, i.e. ${\cal O} (e^{-\kappa N})$ for some finite positive $\kappa$. For non-analytic functions, if the basis functions themselves do not incorporate the non-analytic behavior, then the error of the basis expansion only converges algebraically like ${\cal O} (N^{-\alpha})$, where $\alpha$ depends on the particular expectation value we are interested in as well as the nature of the non-analyticity.

Kato proved that the electronic wavefunction we are interested is continuous but has a discontinuous (yet finite) first derivative at the nuclei (the electron-nuclear cusp) and at the electron-electron coincidences (the electron-electron cusp) \cite{Kato1957}. The asymptotic rate of convergence of both the plane wave expansion and Gaussian expansion is governed by their ability to capture these cusp-like behaviors. Around a cusp, the wavefunction may be expanded as
\begin{align}
\Psi\left(s\right) = \Psi\left(0\right)\left(1 + a_1 s + a_2 s^2 + \ldots\right)
\end{align}
where $s$ is a radial coordinate around the cusp (e.g. $|r_p -R_j|$ for the electron-nuclear cusp or $|r_p - r_q|$ for the electron-electron cusp) and where we have kept the spherical part of the wavefunction for simplicity. The linear coefficient $a_1$ is determined by the type of cusp (e.g. $a_1 = -\zeta$ for a nuclear cusp and $a_1 = 1/2$ for the electron-electron cusp). An expansion in an analytic function basis (e.g. plane waves or Gaussians) necessarily omits the linear $s$ (or it would have a discontinuous first derivative by assumption) and thus, asymptotically incurs error in some volume $S$ close to the cusp, where $S$ is the smallest spatial feature resolvable by the basis, which is ${\cal O}(1/N)$. While appropriately constructed Gaussian basis sets can resolve local features such as the electron-nuclear cusp at a rate faster the ${\cal O}(1/N)$ (see below for more detail), the same is not true of the electron-electron cusp, which occurs at all points in configuration space where coordinates of two or more electrons coincide. Evaluating the energy error in the ground state as
\begin{equation}
\Delta E \approx 4\pi \int_S s^2\, \Psi\left(s\right) H \, \Psi\left(s\right) \textrm{d}s,
\end{equation}
and using the leading terms in the kinetic energy and potential energy in the Hamiltonian, proportional to  $(1/s)(\textrm{d}/\textrm{d}s)$ and $1/s$ respectively, the linear term in the wavefunction gives an error, to leading order in $s$, as $\Delta E \in {\cal O}(S)$. By this intuitive argument, the error in the energy incurred by the cusp should scale asymptotically as ${\cal O}(1/N)$.

The ${\cal O}(1/N)$ scaling for the contribution of the electron-electron cusp to the energy has long been observed empirically using Gaussian basis sets, see e.g. \cite{helgaker1997basis,klopper1995ab,Halkier1998} and extrapolating the so-called electron correlation energy using this asymptotic form is a common practice in electron structure theory \cite{Helgaker2002}. The complicated form of molecular Gaussian basis sets prevents a more rigorous derivation of this form beyond arguments similar to the ones we presented above. However, for the case of two-electron atoms (the simplest electronic structure system demonstrating an electron-electron cusp), a rigorous partial wave analysis is possible at the level of a perturbative treatment of the electron-electron interaction~\cite{kutzelnigg1992rates}. This finds that at second order perturbation theory, the energy convergence of each partial wave goes like $(\ell+1/2)^{-4}$ where $\ell$ is the angular momentum of the partial wave. Adding up the contributions of each partial wave to a maximum cutoff $\ell=L$, gives a convergence like ${\cal O}(1/L^3)$, and the total number of angular functions up to the cutoff $L$ is also ${\cal O}(L^3)$, thus the convergence in this case is again ${\cal O}(1/N)$ \cite{kutzelnigg1992rates}. In the case of plane waves, the ${\cal O}(1/N)$ scaling for the contribution of the electron-electron cusp has been shown under both the random phase approximation~\cite{Harl2008} and second order perturbation theory~\cite{shepherd2012convergence}. In \cite{shepherd2012convergence}, there is also a comprehensive numerics study which demonstrates the ${\cal O}(1/N)$ plane wave convergence.

In practice, when using a Gaussian basis, one includes basis functions that are centered on the nuclei. Then, although the Gaussians are formally analytic around the nucleus, one can choose series of Gaussians with increasingly large exponents such that they effectively mimic the sharp features of the electron-nuclear cusp. For an optimally chosen set of coefficients, one can thus improve on the algebraic convergence for the electron-nuclear cusp, and
it has been shown that the convergence of the Gaussian basis for the electron-nuclear contribution scales as ${\cal O}(e^{-\kappa \sqrt{N}})$~\cite{klopper1986gaussian,kutzelnigg1994theory}. However, this improvement is not possible using a single-particle basis alone for the electron-electron cusp, as this is a cusp in the inter-electron coordinate. In the case of plane waves, an equivalent acceleration of convergence for the electron-nuclear cusp can be obtained if one uses pseudopotentials, which restores the analyticity of the wavefunction around the nucleus. In this case, as argued above using arguments from Fourier analysis, the smoothness of the wavefunction means that neglecting electron-electron interaction effects (e.g. as is done in density functional theory) the plane wave error scales as ${\cal O}(e^{-\kappa N})$. In real materials, pseudopotentials are a mainstay of plane wave calculations. It is also possible to introduce a second set of functions to augment the plane wave description of the wavefunction around the nuclear region~\cite{pahl2002plane,lippert1997hybrid,sun2017gaussian}, and such augmented basis sets allow for exponential convergence in the electron-nuclear cusp without pseudopotentials. Thus, the convergence of both Gaussian and plane wave calculations is limited by resolution of the electron-electron cusp, which scales as ${\cal O}(1/N)$, as discussed earlier.

Since the asymptotic convergence of the Gaussian basis and plane wave basis is the same, the asymptotic complexity of algorithms designed using either the plane wave basis or the Gaussian basis may be directly compared for real molecules and materials. However, it is also useful to have an idea of the relative prefactors in the convergence. The precise prefactor depends on the system and accuracy required. As a concrete  example, the cubic diamond and cubic silicon density functional energies using the Perdew-Burke-Ernzerhof exchange-correlation functional and the Goedecker-Teter-Hutter pseudopotential can be converged to an accuracy of 10 milli-eV per atom using approximately 150 plane waves per atom and 250 plane waves per atom respectively; the same accuracy in a Gaussian basis with the same pseudopotential requires a quadruple-zeta double-polarization basis or larger, which for these systems has 26 Gaussian basis functions per atom, a factor of 6-10. While this example is for a density functional calculation, it serves to  illustrate the relative spatial resolution of the two bases, which is the main factor in resolving the electron-electron cusp in correlated calculations. In \cite{booth2016plane} an analysis carried out at the correlated wavefunction level finds that the number of Gaussians needed to reproduce a plane wave calculation of fixed dimension (for a surface adsorption problem) to chemical accuracy is approximately less by a factor of 20-30, although this is a significant overestimate since the number of plane waves used is substantially more than is required for chemical accuracy. In summary, a rough estimate for the plane wave basis size versus Gaussians basis size for the same accuracy is approximately ten.

Within the context of performing quantum simulation experiments on the most advanced hardware platforms (specifically industrial transmon platfroms being designed at Google, IBM, Intel, Rigetti and elsewhere) in the next few years, gate count (not qubit count) is the primary concern. While most expect that more qubits can be manufactured in a scalable fashion, there is no clear path to substantially improving the gate fidelities already achieved by the most advanced transmon platforms. And the total fidelity of a circuit decreases exponentially in the number of gates. Less obvious is the fact that gate count (not logical qubit count) also determines the primary overhead in quantum error-correction. This is because a very large number of physical qubits (often hundreds of times more than the number of qubits required for a logical bit) are required to perform state distillation in order to implement non-transversal gates (e.g. T gates in the toric / surface code). Thus, we expect the scaling advantages of our approach to translate into practical gains for a variety of interesting quantum simulations, both in the near-term and in the distant future.

\subsection{Modeling Non-Periodic Systems with a Periodic Basis}
\label{app:image_error}

Plane waves are often used as a basis for systems with reduced periodicity, e.g. surfaces (periodic in two dimension), nanowires (periodic in one dimension), or even single-molecules (periodic in zero dimensions)~\cite{marx2009ab}. The main concern to address with plane waves in such simulations is that that they enforce a periodic charge density and thus produce fictitious image interactions between computational cells.
A simple way to avoid this is to make the computational cell volume $\Omega$ sufficiently large so that periodic images of the cells do not interact. This is typically what is done for surface calculations, where it is necessary only to extend the cell volume in one or two of the spatial directions. However, a more efficient and rigorous procedure, particularly for systems that are periodic in zero dimensions such as single-molecules, is to use a truncated Coulomb operator with a slightly larger cell size~\cite{rozzi2006exact,sundararaman2013regularization,ismail2006truncation}.

To see how this works, we consider  an isolated molecule. The total density
of a molecule decays exponentially quickly away from its center, and thus the molecule may be inscribed in a box of volume $\Omega = D^3$ with only exponentially small parts of the density (and contributions to the energy) outside of the box. By using a Coulomb operator truncated at distance $D$~\cite{fusti2002accurate}, such that
\begin{equation}
V\left(r, r'\right) = \begin{cases} \frac{1}{\left | r - r' \right |} & \left | r - r' \right | \leq D \\
0  & \left | r - r' \right | > D,
\end{cases}
\end{equation}
and by carrying out the simulation in a box of size $8\, \Omega = (2\,D)^3$, we ensure that there is no Coulomb interaction at all between the repeated images of the molecule, up to exponentially small terms in $\Omega$ arising from the density of the molecule outside of the box. While the Fourier amplitudes of the normal Coulomb operator are $4 \pi / k^2$, the Fourier amplitudes of the truncated Coulomb interaction become $4 \pi (1 - \cos[|k|D]) / k^2$. The exact analytical form of this correction gives the following Coulomb operators in the plane wave basis:
\begin{equation}
V = \frac{2 \pi}{\Omega} \!\!\!\!\!\!\!\!  \sum_{\substack{\nu \neq 0\\ (p, \sigma) \neq (q, \sigma')}} \!\!\!\!\!\!\!\! \left(1 - \cos\left[\left | k_\nu \right| D\right] \right)\frac{c^\dagger_{p,\sigma} c_{q,\sigma'}^\dagger c_{q + \nu,\sigma'} c_{p - \nu,\sigma}}{k_\nu^2}%
\qquad
U = \frac{4 \pi}{\Omega} \sum_{\substack{p \neq q \\ j,\sigma}} \left(\cos\left[\left | k_\nu \right| D\right] - 1\right) \left(\zeta_j \frac{e^{i \, k_{q-p} \cdot R_j}}{k_{p-q}^2}\right) c^\dagger_{p, \sigma} c_{q, \sigma}.
\end{equation}
These operators follow straightforwardly from the derivation in \app{plane_waves} if the Fourier transformed potentials of \eq{periodic_coulomb_r} are convolved with the $(1 - \cos[|k_\nu|D])$ correction inside of the sum over $\nu$.

In the dual basis, the truncated Coulomb operator can be implemented even more straightforwardly: one simply drops all $n_p n_q$ terms for which $|r_p - r_q| > D$. As with the plane waves, to maintain resolution, we increase the number of basis functions by exactly a factor of eight. Taken together with the prior arguments in this appendix, this concludes our argument that electronic structure simulations of systems of reduced periodicity can be carried out using plane wave (and dual) orbitals with the same asymptotic scaling as Gaussian orbitals.

\section{Operator Norm Bounds}
\label{app:operators}

In this appendix we bound the norms of the Hamiltonian components $H = T + U + V$ in the plane wave dual basis. These bounds are used extensively in determining the asymptotic scalings discussed in \sec{sec_two}. However, we note that these bounds are likely loose and that one should compute these bounds numerically in order to determine practical scaling. Recall that we restrict the support of all operators to $N$ plane waves with momenta in each dimension not exceeding absolute value proportional to $N^{1/3}/\Omega^{1/3}$.

We begin with  the two-body potential operator $V$, as given in \eq{pwd_v}. For any state $\ket{\psi}$ inside the $\eta$-electron manifold of the Hilbert space we wish to estimate
\begin{align}
\max_{\psi} | \bra{\psi} V\ket{\psi}| & =  \max_\psi \left | \bra{\psi} \frac{2\pi}{\Omega} \sum_{\substack{(p,\sigma) \ne (q,\sigma') \\ \nu \neq 0}} \frac{\cos\left[k_\nu\cdot \left(r_p-r_q\right)\right]}{k_\nu^2} \, n_{p,\sigma}n_{q,\sigma'} \ket{\psi} \right| \\
& \le \frac{2\pi \eta^2}{\Omega} \sum_{\nu \neq 0} \frac{1}{k_\nu^2} = \frac{\eta^2}{2 \pi \, \Omega^{1/3}} \sum_{(\nu_x, \nu_y, \nu_z) \neq (0,0,0)} \frac{1}{\nu_x^2 + \nu_y^2 + \nu_z^2}.\nonumber
\end{align}
As the sum above does not have a closed form in three dimensions, we will upper bound it using integrals. In particular, we use the fact that for monotonically decreasing $f$,
\begin{equation}
\sum_{x=a}^b f(x) \le f(a)+\int_{a}^b f(x) \, dx.
\end{equation}
We will break the sum into three cases corresponding to one, two and three dimensional sums for the potential operator. First let us consider the case of the one-dimensional sum encountered when $\nu_y = \nu_z = 0$,
\begin{equation}
\sum_{\nu_x\ne 0}\frac{1}{\nu_x^2} \le 1+\int_{1}^\infty \frac{dx}{x^2} \in {\cal O}(1).\label{eq:1Dsum}
\end{equation}
Now consider the two-dimensional case encountered when $\nu_z = 0$,
\begin{equation}
 \sum_{(\nu_x, \nu_y) \neq (0,0)} \frac{1}{\nu_x^2 + \nu_y^2} 
= \sum_{\nu_x \neq 0} \frac{2}{\nu_x^2} + \sum_{\substack{\nu_x \neq 0 \\ \nu_y \neq 0}} \frac{1}{\nu_x^2 + \nu_y^2}.
\end{equation}
The one-dimensional sum above occurs when $\nu_x > 0,\nu_y=\nu_z = 0$.  This sum is ${\cal O}(1)$ from~\eq{1Dsum}. The second case term can be bounded using the fact that $1/(\nu_x^2 + \nu_y^2)$ is a monotonically decreasing function of both variables
\begin{align}
\label{eq:2Dsum}
\sum_{\substack{\nu_x \neq 0 \\ \nu_y \neq 0}} \frac{1}{\nu_x^2 + \nu_y^2}
& \le \sum_{\nu_x \neq 0}\frac{1}{\nu_x^2}\sum_{\nu_y\neq 0}\frac{1}{1+\nu_y^2/\nu_x^2}
\le \sum_{\nu_x\neq 0}\frac{1}{\nu_x^2}\left[1+\int_{1}^{N^{1/3}} \frac{dy}{1+y^2/\nu_x^2}\right]\\
& \le \sum_{\nu_x \neq 0} \frac{1}{\nu_x^2} + \int_{1}^{N^{1/3}} \!\!\! \int_{1}^{N^{1/3}}  \frac{dx \, dy}{x^2+y^2}
\le \int_1^{N^{1/3}} \frac{2 \, dx}{x^2}+\int_{1}^{N^{1/3}} \! \! \! \int_{1}^{N^{1/3}}  \frac{dx \, dy}{x^2+y^2} \nonumber\\
& \le \int_1^{N^{1/3}} \frac{2 \, dx}{x^2}+ 2\pi \int_1^{\sqrt{2} N^{1/3}} \frac{dr}{r}\in {\cal O}\left(\log N\right).\nonumber
\end{align} 
Finally consider the three-dimensional case. Using the exact same reasoning spelled out before, but using a spherical integral rather than a polar integral, we find from~\eq{2Dsum} that in three dimensions
\begin{equation}
\sum_{(\nu_x, \nu_y, \nu_z) \neq (0,0,0)} \frac{1}{\nu_x^2 + \nu_y^2 + \nu_z^2} 
 \leq 4\pi\left(\sqrt{3} \, \frac{N^{1/3}}{2} -1\right)+\int_{1}^{N^{1/3}} \frac{3 \, dz}{z^2} + \int_{1}^{N^{1/3}} \!\!\! \int_{1}^{N^{1/3}} \frac{3 \, dx \, dy}{x^2 +y^2} \in {\cal O} \left(N^{1/3}\right).
 \label{eq:3dVbd}
\end{equation}
Thus, from~\eq{1Dsum},~\eq{2Dsum} and~\eq{3dVbd} we find that,
\begin{equation}
\label{eq:v_bound}
\max_{\psi} | \bra{\psi} V\ket{\psi}| \in {\cal O}\left(\frac{\eta^2 \, N^{1/3}}{\Omega^{1/3}}\right)
\quad \quad \quad \left\| V \right\| \in {\cal O}\left(\frac{N^{7/3}}{\Omega^{1/3}}\right).
\end{equation}
Note that the dimensions of the potential are in units of inverse length and the energy scales as $\eta^2$ as expected.

We can now bound the norm of the external potential operator $U$, as given in \eq{pwd_u}. For any state $\ket{\psi}$ inside the $\eta$-electron manifold of the Hilbert space we wish to estimate
\begin{align}
\max_{\psi} | \bra{\psi} U\ket{\psi}| & = \max_\psi \left | \bra{\psi} \frac{4 \pi}{\Omega} \sum_{\substack{p,\sigma \\ j, \nu\neq 0}} \frac{\zeta_j \, \cos\left[ k_{\nu} \cdot \left(R_j - r_{p}\right)\right] }{k_\nu^2} n_{p, \sigma} \ket{\psi} \right| \\
& \le \frac{4\pi \eta}{\Omega} \left(\sum_j \zeta_j\right) \sum_{\nu \neq 0} \frac{1}{k_\nu^2} = \frac{\eta^2}{\pi \, \Omega^{1/3}} \sum_{(\nu_x, \nu_y, \nu_z) \neq (0,0,0)} \frac{1}{\nu_x^2 + \nu_y^2 + \nu_z^2}\nonumber
\end{align}
where we have assumed that $\sum_j \zeta_j = \eta$ as this must be true when treating periodic systems which, in general, must be charge neutral. Thus, we can see that the external potential has the same bound as the two-body potential,
\begin{equation}
\label{eq:u_bound}
\max_{\psi} | \bra{\psi} U\ket{\psi}| \in {\cal O}\left(\frac{\eta^2 \, N^{1/3}}{\Omega^{1/3}}\right)
\quad \quad \quad \left\| U \right\| \in {\cal O}\left(\frac{N^{7/3}}{\Omega^{1/3}}\right).
\end{equation}
We use the equality of these bounds when estimating the variance of measuring $U + V$ in \sec{measurement}.

Finally, we bound the norm of the kinetic energy operator $T$. It turns out that the kinetic energy operator is much easier to tightly bound in momentum space and so we derive the bound from \eq{pw_t} rather than from \eq{pwd_t}. The bound holds for both cases as a consequence of Parseval's theorem and the unitarity of the discrete Fourier transform. The bound is computed as
 \begin{align}
 \label{eq:t_bound}
 \max_{\psi} | \bra{\psi} T\ket{\psi}| & \le \left| \sum_{\nu,\sigma} \frac{k_\nu^2}{2} \, \bra{\psi}n_{\nu, \sigma}\ket{\psi} \right|
 \leq  \frac{2\pi^2 \, \eta}{\Omega^{2/3}}  \nu_{\max}^2 \in {\cal O}\left(\frac{\eta \, N^{2/3}}{\Omega^{2/3}}\right).
 \end{align}
However, in some cases (e.g. for the value of $\Lambda$ in the Taylor series method in \sec{taylor_alg}) we are interested in the triangle inequality upper bound on the operator norm, which is not invariant under a Fourier transform. Thus, it may also be useful to bound $T$ in the plane wave dual basis with a triangle inequality as
\begin{align}
\frac{1}{2\, N} \sum_{p, q, \sigma} \left | \left(\sum_\nu k_\nu^2 \cos \left[k_\nu \cdot r_{q - p} \right]\right) a^\dagger_{p, \sigma} a_{q,\sigma} \right |
=  \frac{1}{2} \sum_{p} \left | \sum_\nu k_\nu^2 \cos \left[k_\nu \cdot r_p\right] \right |
=  \frac{1}{2} \sum_{p} \left |\nabla^2  \sum_\nu \exp\left[\frac{2\pi \, i}{\Omega^{1/3}} \nu \cdot r_p\right] \right |
\end{align}
where $\nabla^2 = \partial^2_x + \partial^2_y + \partial^2_z$ is the Laplacian, which acts on $r$. We do this so that we can expand the inner sum using a geometric series in $\nu$:
\begin{align}
\sum_\nu \exp\left[\frac{2\pi \, i}{\Omega^{1/3}} \nu \cdot r_p\right] & =
\left(\sum_{\nu_x} \exp\left[\frac{2\pi \, i}{\Omega^{1/3}} \nu_x r_{p_x}\right]\right)
\left(\sum_{\nu_y} \exp\left[\frac{2\pi \, i}{\Omega^{1/3}} \nu_y  r_{p_y}\right]\right)
\left(\sum_{\nu_z} \exp\left[\frac{2\pi \, i}{\Omega^{1/3}} \nu_z r_{p_z}\right]\right)\\
& = \left(\frac{1- \exp\left[\frac{2\pi \, i \, N^{1/3}}{\Omega^{1/3}} r_{p_x}\right]}{1 - \exp\left[\frac{2\pi \, i}{\Omega^{1/3}} r_{p_x}\right]}\right)
\left(\frac{1- \exp\left[\frac{2\pi \, i \, N^{1/3}}{\Omega^{1/3}} r_{p_y}\right]}{1 - \exp\left[\frac{2\pi \, i}{\Omega^{1/3}} r_{p_y}\right]}\right)
\left(\frac{1- \exp\left[\frac{2\pi \, i \, N^{1/3}}{\Omega^{1/3}} r_{p_z}\right]}{1 - \exp\left[\frac{2\pi \, i}{\Omega^{1/3}} r_{p_z}\right]}\right)\nonumber.
\end{align}
We can now see that
\begin{align}
\left(\partial^2_x + \partial^2_y + \partial^2_z\right)\left(\frac{1- \exp\left[\frac{2\pi \, i \, N^{1/3}}{\Omega^{1/3}} r_{p_x}\right]}{1 - \exp\left[\frac{2\pi \, i}{\Omega^{1/3}} r_{p_x}\right]}\right)
\left(\frac{1- \exp\left[\frac{2\pi \, i \, N^{1/3}}{\Omega^{1/3}} r_{p_y}\right]}{1 - \exp\left[\frac{2\pi \, i}{\Omega^{1/3}} r_{p_y}\right]}\right)
\left(\frac{1- \exp\left[\frac{2\pi \, i \, N^{1/3}}{\Omega^{1/3}} r_{p_z}\right]}{1 - \exp\left[\frac{2\pi \, i}{\Omega^{1/3}} r_{p_z}\right]}\right)
\in {\cal O}\left(\frac{N^{2/3}}{\Omega^{2/3}}\right)
\end{align}
and this leads us to our final bound by completing the sum over $p$,
\begin{equation}
\left \| T \right \| 
\leq \frac{1}{2} \sum_{p} \left |\nabla^2  \sum_\nu \exp\left[\frac{2\pi \, i}{\Omega^{1/3}} \nu \cdot r_p\right] \right |
\in {\cal O}\left(\frac{N^{5/3}}{\Omega^{2/3}}\right).
\end{equation}
This turns out to be exactly consistent with the bound we obtained from the momentum space operator but it was necessary to show the triangle inequality norm remained the same. Finally, the norm of the Hamiltonian $H = T + U + V $ and the upper bound on its expectation value is thus
\begin{equation}
\label{eq:h_bound}
\max_{\psi} | \bra{\psi} H\ket{\psi}| \in {\cal O}\left(\frac{\eta^2 \, N^{1/3}}{\Omega^{1/3}} + \frac{\eta \, N^{2/3}}{\Omega^{2/3}}\right)
\quad \quad \quad
\left \| H \right \| \in {\cal O}\left(\frac{N^{7/3}}{\Omega^{1/3}} + \frac{N^{5/3}}{\Omega^{2/3}}\right).
\end{equation}

\section{Error Bounds for Trotter Suzuki Formulas}
\label{app:TSerror}
Now that we are equipped with the operator bounds in~\app{operators} we can prove bounds on the Trotter error.  For simplicity we state our result below as a lemma to allow the result to be easily reused in subsequent work.

\begin{lemma}
Let $H$ be the Hamiltonian of~\eq{pwd_ham}, let $\sum_j \zeta_j = \eta$, $\sum_j |\zeta_j| \in {\cal O}(\eta)$ and $\mathcal{K}$ be the set of states such that $\ket{\phi} \in \mathcal{K}$ if and only if $\sum_{\nu,\sigma} n_{\nu,\sigma} \ket{\phi}=\eta\ket{\phi}$.  Under these assumptions we have that $$\left|\max_{\psi\in \mathcal{K}} \bra{\psi} \left(e^{-i(U+V)t/2r}{\rm FFFT}^\dagger e^{-i\frac{1}{2} \sum_{\nu, \sigma} a_{\nu,\sigma}^\dagger a_{\nu,\sigma} t/r}{\rm FFFT}e^{-i(U+V)t/2r}\right)^r  - e^{-iHt}\ket{\psi}\right|\le \epsilon$$ for a value of $r$ that obeys
$$
r\in{\Theta}\left(\frac{\eta^2 N^{5/6}t^{3/2}}{\Omega^{5/6}\sqrt{\epsilon}}\sqrt{1+\frac{\eta\Omega^{1/3}}{N^{1/3}}} \right).
$$
\end{lemma}
\begin{proof}
For methods that simulate the Trotter steps using the FFFT~\eq{commbound} gives us that the error is dominated by two commutators: $[T,[T,U+V]]$ and $[U+V,[T,U+V]]$.  We need to bound both of these terms.
Because both $T$ and $U+V$ conserve particle number we know that total particle number commutes with the Hamiltonian, the total particle number is a constant of motion for the evolution.  As such, let $\mathcal{K}$ be the manifold of states that contains $\eta$ electrons.  Then for $\ket{\psi}\in\mathcal{K}$, $H\ket{\psi} = \sum_{j=1}^N \proj{j}H\ket{\psi} = \sum_{\ket{\phi}\in \mathcal{K}} \proj{\phi}H\ket{\psi}:=P_{\mathcal{K}} H\ket{\psi}$.  This implies for the induced $2$--norm that

\begin{align}
|\bra{\psi} T^2(U+V) \ket{\psi}|&=| \bra{\psi} TP_{\mathcal{K}} TP_{\mathcal{K}}([U+V]P_{\mathcal{K}}) \ket{\psi}|\le \|TP_{\mathcal{K}}\|^2 \|[U+V]P_{\mathcal{K}}\|\nonumber\\
& = \max_{\psi \in \mathcal{K}}|\bra{\psi}T\ket{\psi}|^2\max_{\psi \in \mathcal{K}}|\bra{\psi}[U+V]\ket{\psi}|^2.
\end{align}
By repeating the same argument for each term that appears in the nested commutators and using the triangle inequality we then have that the error in the Trotter--Suzuki decomposition is in
\begin{equation}
\mathcal{O}\left(\max_{\psi \in \mathcal{K}}|\bra{\psi}T\ket{\psi}|^2\max_{\psi \in \mathcal{K}}|\bra{\psi}[U+V]\ket{\psi}|+\max_{\psi \in \mathcal{K}}|\bra{\psi}[U+V]\ket{\psi}|^2\max_{\psi \in \mathcal{K}}|\bra{\psi}T\ket{\psi}| \right).
\end{equation}
Then using the bounds on the kinetic and potential magnitudes we find that the Trotter error scales as
\begin{equation}
\mathcal{O}\left(\left[\frac{\eta^4N^{5/3}}{\Omega^{5/3}}+\frac{\eta^5N^{4/3}}{\Omega^{4/3}}\right]\frac{t^3}{r^2} \right).
\end{equation}
If we wish the error to be at most $\epsilon$ it therefore suffices to take a value of $r$ in
\begin{equation}
{\Theta}\left(\frac{\eta^2 N^{5/6}t^{3/2}}{\Omega^{5/6}\sqrt{\epsilon}}\sqrt{1+\frac{\eta\Omega^{1/3}}{N^{1/3}}} \right),
\end{equation}
as claimed.
\end{proof}

\section{Linear Depth Circuit to Place all Qubits Adjacent on Planar Lattice}
\label{app:qubit_cycle}

In this section we will describe a circuit which swaps qubits on a planar lattice so as to place them all adjacent at least once with circuit depth ${\cal O}(N)$. This circuit is useful in many contexts, including for the implementation of the potential operator which consists of terms having the form $Z_i Z_j$. We describe the process informally below for the case of square lattices before providing a formal proof that the method works for a wide class of rectangular lattices. The motivation for restricting qubit connectivity to planar lattice comes from existing superconducting qubit platforms which have this restriction. For the purpose of explanation, we will illustrate the scheme for a 4 by 4 grid of qubits. Our circuit is implemented in four steps.\\

{\bf Step 1}. Define a closed-loop 1D path through the qubits. This will always be possible on any rectangular  arrangement of qubits on a planar lattice. For instance, for the 4 by 4 grid, one possible closed-loop path is shown in \fig{closed_loop}. We then decompose this path into two different, disconnected graphs which we will call the ``left stagger'' and ``right stagger''. We show an example of this decomposition in \fig{paths}.\\

\begin{figure}[h]
\centering
\subfloat[1D closed-loop path]{\label{fig:closed_loop}
\begin{tikzpicture}
\foreach \x in {0,2}
\foreach \y in {0,2} 
{\pgfmathtruncatemacro{\label}{12 + \x - 4 * \y}
\node [circle,draw,inner sep=0pt, text width=6mm, align=center,text opacity=0] (\x\y) at (\x,\y) {\label};}
\foreach \x in {0,2}
\foreach \y in {1,3} 
{\pgfmathtruncatemacro{\label}{12 + \x - 4 * \y}
\node [circle,draw,inner sep=0pt, text width=6mm, align=center,text opacity=0] (\x\y) at (\x,\y) {\label};}
\foreach \x in {1,3}
\foreach \y in {0,2} 
{\pgfmathtruncatemacro{\label}{12 + \x - 4 * \y}
\node [circle,draw,inner sep=0pt, text width=6mm, align=center,text opacity=0] (\x\y) at (\x,\y) {\label};}
\foreach \x in {1,3}
\foreach \y in {1,3} 
{\pgfmathtruncatemacro{\label}{12 + \x - 4 * \y}
\node [circle,draw,inner sep=0pt, text width=6mm, align=center,text opacity=0] (\x\y) at (\x,\y) {\label};}
\draw (13) -- (23);
\draw (23) -- (33);
\draw (33) -- (32);
\draw (32) -- (22);
\draw (22) -- (12);
\draw (12) -- (11);
\draw (11) -- (21);
\draw (21) -- (31);
\draw (31) -- (30);
\draw (30) -- (20);
\draw (20) -- (10);
\draw (10) -- (00);
\draw (00) -- (01);
\draw (01) -- (02);
\draw (02) -- (03);
\draw (03) -- (13);
\end{tikzpicture}}
\hfill
\subfloat[``Left stagger'' $(U_L)$]{\label{fig:left_stagger}
\begin{tikzpicture}
\foreach \x in {0,2}
\foreach \y in {0,2} 
{\pgfmathtruncatemacro{\label}{12 + \x - 4 * \y}
\node [circle,draw,inner sep=0pt, text width=6mm, align=center,text opacity=0] (\x\y) at (\x,\y) {\label};}
\foreach \x in {0,2}
\foreach \y in {1,3} 
{\pgfmathtruncatemacro{\label}{12 + \x - 4 * \y}
\node [circle,draw,inner sep=0pt, text width=6mm, align=center,text opacity=0] (\x\y) at (\x,\y) {\label};}
\foreach \x in {1,3}
\foreach \y in {0,2} 
{\pgfmathtruncatemacro{\label}{12 + \x - 4 * \y}
\node [circle,draw,inner sep=0pt, text width=6mm, align=center,text opacity=0] (\x\y) at (\x,\y) {\label};}
\foreach \x in {1,3}
\foreach \y in {1,3} 
{\pgfmathtruncatemacro{\label}{12 + \x - 4 * \y}
\node [circle,draw,inner sep=0pt, text width=6mm, align=center,text opacity=0] (\x\y) at (\x,\y) {\label};}
\draw (23) -- (33);
\draw (32) -- (22);
\draw (12) -- (11);
\draw (21) -- (31);
\draw (30) -- (20);
\draw (10) -- (00);
\draw (01) -- (02);
\draw (03) -- (13);
\end{tikzpicture}}
\hfill
\subfloat[``Right stagger'' $(U_R)$]{\label{fig:right_stagger}
\begin{tikzpicture}
\foreach \x in {0,2}
\foreach \y in {0,2} 
{\pgfmathtruncatemacro{\label}{12 + \x - 4 * \y}
\node [circle,draw,inner sep=0pt, text width=6mm, align=center,text opacity=0] (\x\y) at (\x,\y) {\label};}
\foreach \x in {0,2}
\foreach \y in {1,3} 
{\pgfmathtruncatemacro{\label}{12 + \x - 4 * \y}
\node [circle,draw,inner sep=0pt, text width=6mm, align=center,text opacity=0] (\x\y) at (\x,\y) {\label};}
\foreach \x in {1,3}
\foreach \y in {0,2} 
{\pgfmathtruncatemacro{\label}{12 + \x - 4 * \y}
\node [circle,draw,inner sep=0pt, text width=6mm, align=center,text opacity=0] (\x\y) at (\x,\y) {\label};}
\foreach \x in {1,3}
\foreach \y in {1,3} 
{\pgfmathtruncatemacro{\label}{12 + \x - 4 * \y}
\node [circle,draw,inner sep=0pt, text width=6mm, align=center,text opacity=0] (\x\y) at (\x,\y) {\label};}
\draw (13) -- (23);
\draw (33) -- (32);
\draw (22) -- (12);
\draw (11) -- (21);
\draw (31) -- (30);
\draw (20) -- (10);
\draw (00) -- (01);
\draw (02) -- (03);
\end{tikzpicture}}
\caption{In the first step we draw a closed-loop 1D path through the qubits, e.g. \fig{closed_loop}. We then decompose the 1D path into a ``left stagger'' (\fig{left_stagger}) and a ``right stagger'' (\fig{right_stagger}).}
\label{fig:paths}
\end{figure}

{\bf Step 2.} Alternate layers of SWAP gates on the ``left stagger'' and ``right stagger'' conformations of the graph. If $U_L$ is a layer of SWAP gates associated with the ``left stagger'' and $U_R$ is a layer of SWAP gates associated with the ``right stagger'' then one should implement $(U_R U_L)^{N / 2}$ where $N$ is the number of qubits. This circuit has depth of exactly $N$ cycles and returns all of the qubits to their original positions.\\

A key insight is that half of the qubits will circulate along the 1D path in a clockwise fashion and half of the qubits will circulate around the circuit in a counter-clockwise fashion. To see this, it is helpful to imagine the qubits as being colored in a checkerboard fashion. We demonstrate the first four layers of this pattern for the 4 by 4 lattice in \fig{checkers}. If we imagine the qubits colored as in \fig{checkers} then we can clearly see that the blue qubits will circulate clockwise and the red qubits will circulate counterclockwise. Because the qubits will return to their original locations after $(U_R U_L)^{N / 2}$, all of the blue qubits must have ``moved through'' all of the red qubits and thus, all of the blue qubits have been adjacent to all of the red qubits. What remains is to make all of the blue qubits adjacent to all of the blue qubits and all of the red qubits adjacent to all of the red qubits.\\

\begin{figure}[h!]
\centering
\subfloat[Cycle 1]{
\begin{tikzpicture}[>=stealth, >-<]
\foreach \x in {0,2}
\foreach \y in {0,2} 
{\pgfmathtruncatemacro{\label}{12 + \x - 4 * \y}
\node [circle,draw,fill=red!50,inner sep=0pt, text width=6mm, align=center] (\x\y) at (\x,\y) {\label};}
\foreach \x in {0,2}
\foreach \y in {1,3} 
{\pgfmathtruncatemacro{\label}{12 + \x - 4 * \y}
\node [circle,draw,fill=blue!50,inner sep=0pt, text width=6mm, align=center] (\x\y) at (\x,\y) {\label};}
\foreach \x in {1,3}
\foreach \y in {0,2} 
{\pgfmathtruncatemacro{\label}{12 + \x - 4 * \y}
\node [circle,draw,fill=blue!50,inner sep=0pt, text width=6mm, align=center] (\x\y) at (\x,\y) {\label};}
\foreach \x in {1,3}
\foreach \y in {1,3} 
{\pgfmathtruncatemacro{\label}{12 + \x - 4 * \y}
\node [circle,draw,fill=red!50,inner sep=0pt, text width=6mm, align=center] (\x\y) at (\x,\y) {\label};}
\draw [thick, >-<] (23) -- (33);
\draw [thick, >-<] (32) -- (22);
\draw [thick, >-<] (12) -- (11);
\draw [thick, >-<] (21) -- (31);
\draw [thick, >-<] (30) -- (20);
\draw [thick, >-<] (10) -- (00);
\draw [thick, >-<] (01) -- (02);
\draw [thick, >-<] (03) -- (13);
\end{tikzpicture}}
\hfill
\subfloat[Cycle 2]{
\begin{tikzpicture}[>=stealth, >-<]
\node [circle,draw,fill=red!50,inner sep=0pt, text width=6mm, align=center] (03) at (0,3) {1};
\node [circle,draw,fill=blue!50,inner sep=0pt, text width=6mm, align=center] (13) at (1,3) {0};
\node [circle,draw,fill=red!50,inner sep=0pt, text width=6mm, align=center] (23) at (2,3) {3};
\node [circle,draw,fill=blue!50,inner sep=0pt, text width=6mm, align=center] (33) at (3,3) {2};
\node [circle,draw,fill=blue!50,inner sep=0pt, text width=6mm, align=center] (02) at (0,2) {8};
\node [circle,draw,fill=red!50,inner sep=0pt, text width=6mm, align=center] (12) at (1,2) {9};
\node [circle,draw,fill=blue!50,inner sep=0pt, text width=6mm, align=center] (22) at (2,2) {7};
\node [circle,draw,fill=red!50,inner sep=0pt, text width=6mm, align=center] (32) at (3,2) {6};
\node [circle,draw,fill=red!50,inner sep=0pt, text width=6mm, align=center] (01) at (0,1) {4};
\node [circle,draw,fill=blue!50,inner sep=0pt, text width=6mm, align=center] (11) at (1,1) {5};
\node [circle,draw,fill=red!50,inner sep=0pt, text width=6mm, align=center] (21) at (2,1) {11};
\node [circle,draw,fill=blue!50,inner sep=0pt, text width=6mm, align=center] (31) at (3,1) {10};
\node [circle,draw,fill=blue!50,inner sep=0pt, text width=6mm, align=center] (00) at (0,0) {13};
\node [circle,draw,fill=red!50,inner sep=0pt, text width=6mm, align=center] (10) at (1,0) {12};
\node [circle,draw,fill=blue!50,inner sep=0pt, text width=6mm, align=center] (20) at (2,0) {15};
\node [circle,draw,fill=red!50,inner sep=0pt, text width=6mm, align=center] (30) at (3,0) {14};
\draw [thick, >-<] (13) -- (23);
\draw [thick, >-<] (33) -- (32);
\draw [thick, >-<] (22) -- (12);
\draw [thick, >-<] (11) -- (21);
\draw [thick, >-<] (31) -- (30);
\draw [thick, >-<] (20) -- (10);
\draw [thick, >-<] (00) -- (01);
\draw [thick, >-<] (02) -- (03);
\end{tikzpicture}}
\hfill
\subfloat[Cycle 3]{
\begin{tikzpicture}[>=stealth, >-<]
\node [circle,draw,fill=blue!50,inner sep=0pt, text width=6mm, align=center] (03) at (0,3) {8};
\node [circle,draw,fill=red!50,inner sep=0pt, text width=6mm, align=center] (13) at (1,3) {3};
\node [circle,draw,fill=blue!50,inner sep=0pt, text width=6mm, align=center] (23) at (2,3) {0};
\node [circle,draw,fill=red!50,inner sep=0pt, text width=6mm, align=center] (33) at (3,3) {6};
\node [circle,draw,fill=red!50,inner sep=0pt, text width=6mm, align=center] (02) at (0,2) {1};
\node [circle,draw,fill=blue!50,inner sep=0pt, text width=6mm, align=center] (12) at (1,2) {7};
\node [circle,draw,fill=red!50,inner sep=0pt, text width=6mm, align=center] (22) at (2,2) {9};
\node [circle,draw,fill=blue!50,inner sep=0pt, text width=6mm, align=center] (32) at (3,2) {2};
\node [circle,draw,fill=blue!50,inner sep=0pt, text width=6mm, align=center] (01) at (0,1) {13};
\node [circle,draw,fill=red!50,inner sep=0pt, text width=6mm, align=center] (11) at (1,1) {11};
\node [circle,draw,fill=blue!50,inner sep=0pt, text width=6mm, align=center] (21) at (2,1) {5};
\node [circle,draw,fill=red!50,inner sep=0pt, text width=6mm, align=center] (31) at (3,1) {14};
\node [circle,draw,fill=red!50,inner sep=0pt, text width=6mm, align=center] (00) at (0,0) {4};
\node [circle,draw,fill=blue!50,inner sep=0pt, text width=6mm, align=center] (10) at (1,0) {15};
\node [circle,draw,fill=red!50,inner sep=0pt, text width=6mm, align=center] (20) at (2,0) {12};
\node [circle,draw,fill=blue!50,inner sep=0pt, text width=6mm, align=center] (30) at (3,0) {10};
\draw [thick, >-<] (23) -- (33);
\draw [thick, >-<] (32) -- (22);
\draw [thick, >-<] (12) -- (11);
\draw [thick, >-<] (21) -- (31);
\draw [thick, >-<] (30) -- (20);
\draw [thick, >-<] (10) -- (00);
\draw [thick, >-<] (01) -- (02);
\draw [thick, >-<] (03) -- (13);
\end{tikzpicture}}
\hfill
\subfloat[Cycle 4]{
\begin{tikzpicture}[>=stealth, >-<]
\node [circle,draw,fill=red!50,inner sep=0pt, text width=6mm, align=center] (03) at (0,3) {3};
\node [circle,draw,fill=blue!50,inner sep=0pt, text width=6mm, align=center] (13) at (1,3) {8};
\node [circle,draw,fill=red!50,inner sep=0pt, text width=6mm, align=center] (23) at (2,3) {6};
\node [circle,draw,fill=blue!50,inner sep=0pt, text width=6mm, align=center] (33) at (3,3) {0};
\node [circle,draw,fill=blue!50,inner sep=0pt, text width=6mm, align=center] (02) at (0,2) {13};
\node [circle,draw,fill=red!50,inner sep=0pt, text width=6mm, align=center] (12) at (1,2) {11};
\node [circle,draw,fill=blue!50,inner sep=0pt, text width=6mm, align=center] (22) at (2,2) {2};
\node [circle,draw,fill=red!50,inner sep=0pt, text width=6mm, align=center] (32) at (3,2) {9};
\node [circle,draw,fill=red!50,inner sep=0pt, text width=6mm, align=center] (01) at (0,1) {1};
\node [circle,draw,fill=blue!50,inner sep=0pt, text width=6mm, align=center] (11) at (1,1) {7};
\node [circle,draw,fill=red!50,inner sep=0pt, text width=6mm, align=center] (21) at (2,1) {14};
\node [circle,draw,fill=blue!50,inner sep=0pt, text width=6mm, align=center] (31) at (3,1) {5};
\node [circle,draw,fill=blue!50,inner sep=0pt, text width=6mm, align=center] (00) at (0,0) {15};
\node [circle,draw,fill=red!50,inner sep=0pt, text width=6mm, align=center] (10) at (1,0) {4};
\node [circle,draw,fill=blue!50,inner sep=0pt, text width=6mm, align=center] (20) at (2,0) {10};
\node [circle,draw,fill=red!50,inner sep=0pt, text width=6mm, align=center] (30) at (3,0) {12};
\draw [thick, >-<] (13) -- (23);
\draw [thick, >-<] (33) -- (32);
\draw [thick, >-<] (22) -- (12);
\draw [thick, >-<] (11) -- (21);
\draw [thick, >-<] (31) -- (30);
\draw [thick, >-<] (20) -- (10);
\draw [thick, >-<] (00) -- (01);
\draw [thick, >-<] (02) -- (03);
\end{tikzpicture}}
\caption{In the second step we alternate between applying $U_L$ (\fig{left_stagger}) and $U_R$ (\fig{right_stagger}). If we color the qubits in a checkboard fashion then we can see that all of the qubits of one color (in this case, blue) will move along the 1D path in a clockwise fashion whereas all of the qubits of the other color (in this case, red) will move along the 1D path in a counterclockwise fashion. We show the first four, of sixteen layers, required to circulate these qubits all the way through the 1D path.}
\label{fig:checkers}
\end{figure}
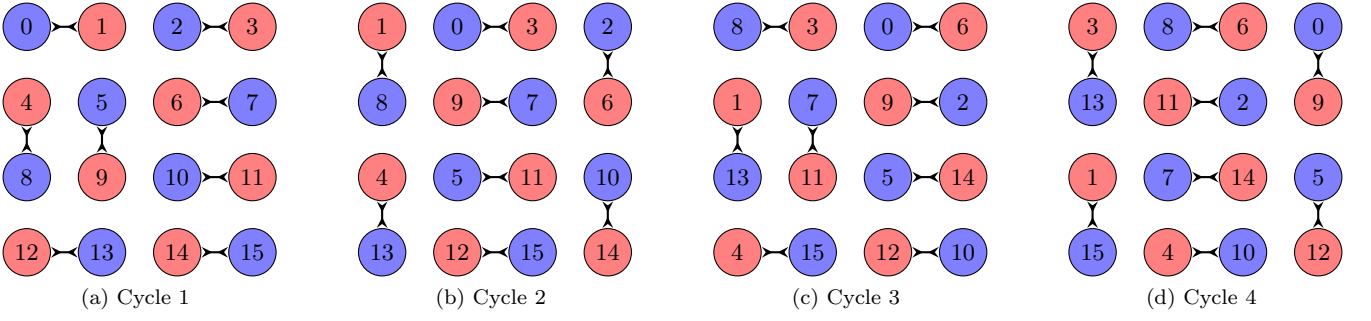

{\bf Step 3.} Alternate between two staggered layers of parallel SWAP gates to move all the ``colors'' of the checkerboard pattern to seperated sides of the qubit array. In the worst case, this will require $\sqrt{N} / 2$ cycles. We demonstrate this in \fig{color_1} and \fig{color_2}.\\

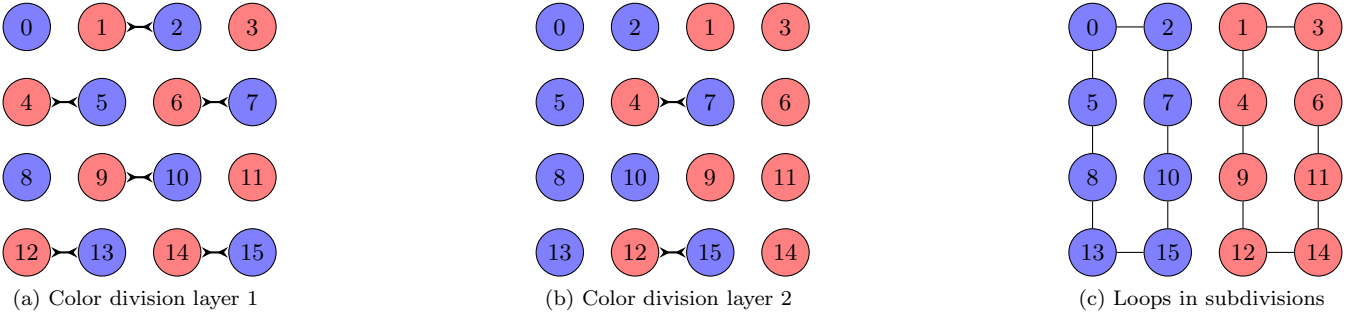
\begin{figure}[h]
\centering
\subfloat[Color division layer 1]{\label{fig:color_1}
\begin{tikzpicture}[>=stealth, >-<]
\foreach \x in {0,2}
\foreach \y in {0,2} 
{\pgfmathtruncatemacro{\label}{12 + \x - 4 * \y}
\node [circle,draw,fill=red!50,inner sep=0pt, text width=6mm, align=center] (\x\y) at (\x,\y) {\label};}
\foreach \x in {0,2}
\foreach \y in {1,3} 
{\pgfmathtruncatemacro{\label}{12 + \x - 4 * \y}
\node [circle,draw,fill=blue!50,inner sep=0pt, text width=6mm, align=center] (\x\y) at (\x,\y) {\label};}
\foreach \x in {1,3}
\foreach \y in {0,2} 
{\pgfmathtruncatemacro{\label}{12 + \x - 4 * \y}
\node [circle,draw,fill=blue!50,inner sep=0pt, text width=6mm, align=center] (\x\y) at (\x,\y) {\label};}
\foreach \x in {1,3}
\foreach \y in {1,3} 
{\pgfmathtruncatemacro{\label}{12 + \x - 4 * \y}
\node [circle,draw,fill=red!50,inner sep=0pt, text width=6mm, align=center] (\x\y) at (\x,\y) {\label};}
\draw [thick, >-<] (02) -- (12);
\draw [thick, >-<] (22) -- (32);
\draw [thick, >-<] (00) -- (10);
\draw [thick, >-<] (20) -- (30);
\draw [thick, >-<] (13) -- (23);
\draw [thick, >-<] (11) -- (21);
\end{tikzpicture}}
\hfill
\subfloat[Color division layer 2]{\label{fig:color_2}
\begin{tikzpicture}[>=stealth, >-<]
\node [circle,draw,fill=blue!50,inner sep=0pt, text width=6mm, align=center] (03) at (0,3) {0};
\node [circle,draw,fill=blue!50,inner sep=0pt, text width=6mm, align=center] (13) at (1,3) {2};
\node [circle,draw,fill=red!50,inner sep=0pt, text width=6mm, align=center] (23) at (2,3) {1};
\node [circle,draw,fill=red!50,inner sep=0pt, text width=6mm, align=center] (33) at (3,3) {3};
\node [circle,draw,fill=blue!50,inner sep=0pt, text width=6mm, align=center] (02) at (0,2) {5};
\node [circle,draw,fill=red!50,inner sep=0pt, text width=6mm, align=center] (12) at (1,2) {4};
\node [circle,draw,fill=blue!50,inner sep=0pt, text width=6mm, align=center] (22) at (2,2) {7};
\node [circle,draw,fill=red!50,inner sep=0pt, text width=6mm, align=center] (32) at (3,2) {6};
\node [circle,draw,fill=blue!50,inner sep=0pt, text width=6mm, align=center] (01) at (0,1) {8};
\node [circle,draw,fill=blue!50,inner sep=0pt, text width=6mm, align=center] (11) at (1,1) {10};
\node [circle,draw,fill=red!50,inner sep=0pt, text width=6mm, align=center] (21) at (2,1) {9};
\node [circle,draw,fill=red!50,inner sep=0pt, text width=6mm, align=center] (31) at (3,1) {11};
\node [circle,draw,fill=blue!50,inner sep=0pt, text width=6mm, align=center] (00) at (0,0) {13};
\node [circle,draw,fill=red!50,inner sep=0pt, text width=6mm, align=center] (10) at (1,0) {12};
\node [circle,draw,fill=blue!50,inner sep=0pt, text width=6mm, align=center] (20) at (2,0) {15};
\node [circle,draw,fill=red!50,inner sep=0pt, text width=6mm, align=center] (30) at (3,0) {14};
\draw [thick, >-<] (12) -- (22);
\draw [thick, >-<] (10) -- (20);
\end{tikzpicture}}
\hfill
\subfloat[Loops in subdivisions]{\label{fig:loops}
\begin{tikzpicture}
\node [circle,draw,fill=blue!50,inner sep=0pt, text width=6mm, align=center] (03) at (0,3) {0};
\node [circle,draw,fill=blue!50,inner sep=0pt, text width=6mm, align=center] (13) at (1,3) {2};
\node [circle,draw,fill=red!50,inner sep=0pt, text width=6mm, align=center] (23) at (2,3) {1};
\node [circle,draw,fill=red!50,inner sep=0pt, text width=6mm, align=center] (33) at (3,3) {3};
\node [circle,draw,fill=blue!50,inner sep=0pt, text width=6mm, align=center] (02) at (0,2) {5};
\node [circle,draw,fill=blue!50,inner sep=0pt, text width=6mm, align=center] (12) at (1,2) {7};
\node [circle,draw,fill=red!50,inner sep=0pt, text width=6mm, align=center] (22) at (2,2) {4};
\node [circle,draw,fill=red!50,inner sep=0pt, text width=6mm, align=center] (32) at (3,2) {6};
\node [circle,draw,fill=blue!50,inner sep=0pt, text width=6mm, align=center] (01) at (0,1) {8};
\node [circle,draw,fill=blue!50,inner sep=0pt, text width=6mm, align=center] (11) at (1,1) {10};
\node [circle,draw,fill=red!50,inner sep=0pt, text width=6mm, align=center] (21) at (2,1) {9};
\node [circle,draw,fill=red!50,inner sep=0pt, text width=6mm, align=center] (31) at (3,1) {11};
\node [circle,draw,fill=blue!50,inner sep=0pt, text width=6mm, align=center] (00) at (0,0) {13};
\node [circle,draw,fill=blue!50,inner sep=0pt, text width=6mm, align=center] (10) at (1,0) {15};
\node [circle,draw,fill=red!50,inner sep=0pt, text width=6mm, align=center] (20) at (2,0) {12};
\node [circle,draw,fill=red!50,inner sep=0pt, text width=6mm, align=center] (30) at (3,0) {14};
\draw (03) -- (13);
\draw (13) -- (12);
\draw (12) -- (11);
\draw (11) -- (10);
\draw (10) -- (00);
\draw (00) -- (01);
\draw (01) -- (02);
\draw (02) -- (03);
\draw (23) -- (33);
\draw (33) -- (32);
\draw (32) -- (31);
\draw (31) -- (30);
\draw (30) -- (20);
\draw (20) -- (21);
\draw (21) -- (22);
\draw (22) -- (23);
\end{tikzpicture}}
\hfill
\caption{In Step 3 we alternate between staggered layers of parallel SWAP gates in order to divide the colors of the checkerboard into two disjoint sectors of the array. In Step 4, we repeat Steps 1-3 within each of these divisions, in parallel.}
\label{fig:division}
\end{figure}

{\bf Step 4.} Repeat Steps 1 through 3, in parallel, for the divided sectors of the array. One should alternate between horizontal and vertical color divisions for Step 3. Once the divided sector size has reached four, a single layer of SWAPs is all that remains to ensure every qubit has neighbored at least once.\\

Steps 1-3 require exactly $N + \sqrt{N}/2$ layers of gates in the worst case. After every repetition of Steps 1-3, the circuit is divided into sectors of half the number of qubits as in the prior iteration. Accordingly, one will need to repeat Steps 1-3 a total of $\log N$ times. Thus, the total gate depth required is as follows,
\begin{equation}
\sum_{k=0}^{\log N} \left(\frac{N}{2^k} + \frac{1}{2}\sqrt{\frac{N}{2^k}}\right) \in \Theta\left(N\right).
\end{equation}

\subsection{Formal Proof}
\begin{lemma}\label{lem:graphcycle}
Let $C_{M}$ be a cycle graph on $M=2mn$ entries for integer $m$ and $n$.  Also  let $P_M$ be a transformation that cyclically permutes the odd vertices in the graph in a counter-clockwise fashion and the even vertices in a clockwise fashion within the cycle graph. Then $(x,y)$ is in the edge set of $C_{M}\bigcup P_{M}(C_{M})\bigcup\cdots \bigcup P_{M}^{{M/2}}(C_{M})$ if and only if $x-y = 1 \mod 2$. 
\end{lemma}
\begin{proof}
First let us formalize what we mean by a cyclic permutation of the vertex labels.  Let 
$$P_{M}: x\mapsto \begin{cases} x-2 \mod M,& \mbox{ if }x \!\!\mod 2 = 0\\ x+2 \mod M,& \mbox{ if } x\!\! \mod 2 =1\end{cases}.$$
For simplicity, let us consider all arithmetic in the following to be modulo $M$.  We have for the cyclic graph that $(x-1,x)$ and $(x,x+1)$ are edges in $C_M$ for every $x\in \mathbb{Z}_{M}$.  Let $x$ be even then $(x+1,x-2)$ and $(x+3,x-2)$ are in $P_M(C_M)$ for all $x\in \mathbb{Z}_{M}$.  Therefore $(x,x+3)$ and $(x,x+5)$ are in $P_M(C_M)$.  By iterating this $q$ times we have that $(x,x+4q-1)$ and $(x,x+4q+1)$ are in $P^q_M(C_M)$.  Also it follows directly from the definition of $P_M$ that $P_M^{M/2}$ is the identity transformation because $x-M = (x+M) \mod M$.

Finally, we need to show that for each odd $y$ that there exists an edge $(x,y)$ in some $P_M^n(C_M)$.  To see this assume that for all $0\le p\le r$ that $(x,y)$ is in $C_M \bigcup \cdots\bigcup P_M^{r}(C_M)$ for all odd $y$ is $[x-1,\ldots,x+4r+1]$.   We can therefore apply $P_M$ to the graph $P_M^r (C_M)$ which includes the edges $(x,x+4(r+1)-1)$ and $(x,x+4(r+1)+1)$.   Thus $(x,x+4(r+1)-1)$ and $(x,x+4(r+1)+1)$ is  in $C_M \bigcup \cdots\bigcup P_M^{r+1}(C_M)$.  Thus $(x,y)$ is in the union for all odd $y$ greater than $x-1$ and less than $(x+4(r+1)-1)$.  Since this trivially holds  for $r=0$ and because the function is periodic with period $M/2$ we then have that our claim holds for all even $x$.

Assume $x$ is odd and there exists even $y$ such that $(x,y)$ is not in $C_M\bigcup P_M(C_M)\bigcup\cdots \bigcup P_k^{M/2}(C_M)$.  Since edges are symmetric this implies that $(y,x)$ is not in the union of graphs as well.  We have shown above that each even vertex has every even vertex as a neighbor and hence this is impossible.  Therefore the claim holds for all $x$.
\end{proof}

\begin{theorem}\label{thm:cycle}
Let $Q_M:C_M \mapsto (C_{M/2},C_{M/2})$ be a function that maps vertices with odd and even labels to two disjoint graphs via an invertible transformation for $M$ a power of $2$.  Further, let $Q_M((C_M,C_M,\ldots,C_M)) = (Q_M(C_M),Q_M(C_M),\ldots,Q_M(C_M))$.  Let $F_M$ be the method of~\lem{graphcycle} defined to act similarly on tuples of graphs.  There exists an algorithm that requires ${\cal O}(\log(M))$ applications of $Q_M$ and $F_M$ such that the union of the edges output by the algorithm is the complete graph on $M$ elements.
\end{theorem}
\begin{proof}
The proof is constructive.  It consists of the following steps.  For $p=M,M/2,\ldots,1$ do the following a)
Apply $F_p$ b) Save all edges that are found in the prior step c) Apply $Q_p$ d) Decode all edges saved in the previous steps to their equivalent edges on the vertex set $\mathbb{Z}_{M}$.   

To understand how this works, let us first consider $F_M({C_M})$.  As argued in~\lem{graphcycle} each vertex in $\mathbb{Z}_{M}$ appears in an edge with every other vertex in~$\mathbb{Z}_{M}$ that has the opposite parity after an appropriate number of applications of the $P_M$ operation. Thus, the union of the resulting edges forms a complete bipartite graph on $Z_{M}$ elements. The results are then saved to ensure that every edge that we have found can be decoded as an edge later.

Next we apply $Q_M$ to the graph.  This mapping is equivalent to splitting both layers of the complete bipartite graph into separate sub-graphs and drawing edges between the vertices to form a cycle graph isomophic to $C_{M/2}$.  If $M\ge 4$ then neither of these graphs consists of elements that have not shared an edge with each other.  Thus we reduce the original problem to two instances of the initial problem.  By recursing we again reduce the sub-graph to a complete bipartite graph, which reduces the number of edges in the complete graph that have not been observed by a factor of $2$.  After recursing this process ${\cal O}(\log(M))$ times it is then clear that every possible combination of edges is observed and saved.  Since the map is by construction invertible, these saved edges can be decoded to edges in the original vertex set which completes our proof.
\end{proof}

\thm{cycle} is notably restricted to cases where $M$ is a power of two.  This is an important restriction for the simple scheme outlined here because if we do not make this assumption then the approach that we take to recursively building the edge sets will not work.  We can also make this work in cases where $M = 2^q X$ for $X\in {\cal O}(1)$ using ${\cal O}(q)$ operations from the above set by recursing until the problem is reduced to building edges between sets of size $X$, which can be handled brute force using bubble sort in ${\cal O}(1)$ steps.  However, in general, if $M=2P$ for prime $P$ then such a construction will not lead to a low depth circuit and idiosyncratic approaches may be needed to make the strategy work. For this reason we focus our attention in the following on graphs with $M=2^k$ vertices. In the following lemma we will use these techniques to show how to simulate the potential term in low-depth on a nearest-neighbor quantum computer that consists of an an integer number of qubits laid out in a rectangular lattice.

\begin{lemma}\label{lem:gates}
Let $S$ be a set of $2^k$ qubits on a nearest neighbor rectangular lattice of dimension $2^{d}\times 2^{k-d}$ such that swap gates and $e^{-iZZ\phi}$ gates can only be performed between neighboring qubits in $S$.  Then $\prod_{(x,y)\in S} e^{-i\phi_{xy} a_x^\dagger a_x a_y^\dagger a_y}$ can be performed on a quantum computer in depth ${\cal O}(2^k)$.
\end{lemma}
\begin{proof}
We prove this result by leveraging~\thm{cycle} but to do so we need to embed the cycle graph described in the theorem within the square lattice.  To see that such an embedding is possible, first note that every cycle has a Hamiltonian path.  Any rectangular grid of size $2^{d}-1 \times 2^{k-d}$ also contains the disjoint union of $2^{k-d-1}$ cycles and edges that connect these cycles to their neighbors.  In particular, if we start a path at $(0,0)$ then by following the Hamiltonian path we can arrive at $(1,0)$.  This qubit is adjacent to vertex $(2,0)$ which is also part of a disjoint cycle and hence there exists a Hamiltonian path for the union of both cycles that links $(0,0)$ to $(0,3)$.  Repeating this argument we see that there is a Hamiltonian path connecting each vertex in the union of these cycles that terminates at $(0,2^{k-d} -1)$.  Now if we introduce another row of vertices beneath this cycle with labels $(-1,0),\ldots (-1,2^{k}-1)$ that have edges between horizontally adjacent qubits as well as edge s between vertex $(-1,0)$ and $(0,0)$ as well as $(-1,2^{k-d}-1)$ and $(0,2^{k-d}-1)$.  Thus there exists a Hamiltonian cycle that can be embedded in every rectangular lattice of dimension $2^{d} \times 2^{k-d}$.  This cycle can be viewed as the cycle graph $C_{2^{k}}$.

Now that we have shown we can implement qubits on a cycle graph in a square lattice we next need to show that we can manipulate the qubits in the manner described in~\thm{cycle}.  To do so we need to first discuss implementing $F_{2^q}$ for $q=1,\ldots,k$.
The operation $F_{2^q}$ can be implemented by swapping qubits every even qubit and its odd neighbor with higher index, and then swapping each even qubit with its odd neighbor with lower index.  This shifts the value of every even qubit two sites in the opposite direction from the data in the corresponding odd qubits.  Ergo it performs the transformation $f$ on the labels ascribed to each qubit site.  Each transformation can be done in depth ${\cal O}(1)$ swaps and in turn the whole series of swaps requires ${\cal O}(2^{q})$ depth.  Furthermore, for each unique edge that is found in this process we can easily apply $e^{-i\phi ZZ}$ to each edge in depth ${\cal O}(1)$.  Thus, we can apply $F_{2^q}$ and perform the necessary phase rotations in depth ${\cal O}(2^{q})$.

The operation $Q_{2^q}$ can be implemented in the following way.  Apply bubble sort using local swap operations to the qubits.  Since there are $2^{q}$ vertices within each set that $Q_{2^q}$ acts on this can be done using a serial bubble sort algorithm using $2^{2q}$ swap operations, however by using parallel bubble sort one can perform ${\cal O}(2^{q})$ comparisons at the same time allowing the algorithm to execute in depth $2^{q}$.  This allows us to sort the qubits such that the vertices $0,\ldots, 2^{q-1}-1$ are assigned even labels and the remaining vertices are assigned odd labels.  Thus an application of $F_{2^q}$, $Q_{2^q}$ requires depth ${\cal O}(2^{q})$.
If the graph has already been partitioned into a disjoint union of Hamiltonian cycles then it is clear that applying $C_{2^q}$ to each of these cycles can be done in depth ${\cal O}(2^{q})$ because these graphs do not interact  gate operations can be applied on them simultaneously. Following the steps outlined in~\thm{cycle} we can produce every edge in the complete graph on $2^{k}$ entries in depth
$$
\sum_{j=0}^{k-1} 2^{k-j} \in {\cal O}(2^{k}),
$$
which completes our proof.
\end{proof}

Now if we define the total number of vertices on the graph to be $N=2^{k}$ then the depth required by the simulation is ${\cal O}(N)$.  This means that in architectures that allow nearest neighbor interactions that act on disjoint qubits to be applied at unit cost requires at most linear time.  This is significant for Trotter based simulations as well as variational algorithms where exponentials of such terms have to be employed in state preparation.

\section{Fermionic Fast Fourier Transform Scaling}
\label{app:fqft}

In this section we will outline a method for applying the ``fermionic fast Fourier transform'' (FFFT) in three dimensions on a planar lattice of quantum bits with ${\cal O}(N)$ depth. In this section we will assume that $M$ frequencies are kept in each direction so that in three dimensions, $N = M^3$. Our derivation of the FFFT will begin by first showing that a one-dimensional FFFT can be performed expediently on a quantum computer and then we will focus on how exactly to perform the three-dimensional analogue on a quantum computer. The first part of the proof of the validity requires us to work out some commutation relations between operators.

One of the basic primitives that we need to construct the FFFT is the fermionic swap operation. The purpose of the fermionic swap operation is to permute the ordering of the spin-orbitals. Under mappings such as the Jordan-Wigner transformation, the ordering of the qubits determines how the operators are anti-symmetrized.  While the ordering of the spin-orbitals is irrelevant to their quantum dynamics, a poor ordering of spin-orbitals can have a major impact on the performance of a quantum simulation algorithms. The fermionic swap operator allows the canonical ordering of these operators to be swapped on the fly. Some important properties of the fermionic swap operator are given below.
\begin{lemma}\label{lem:prop}
Let $a_p^\dagger$ and $a_q^\dagger$ be fermionic creation operators acting on two disjoint spin orbitals and let $f_{\rm swap}$ be the fermionic swap operator between those two spin orbitals given in~\eq{fswap}.  Then the following properties hold,
\begin{enumerate}
\item $[a_p^{\dagger},f_{\rm swap}]=a_p^\dagger-a_q^\dagger$ and $[a_q^{\dagger},f_{\rm swap}]=a_q^\dagger-a_p^\dagger$.
\item $f_{swap}$ is Hermitian and unitary.
\item $f_{\rm swap} a_p^\dagger f_{\rm swap}=a_q^\dagger$ and $f_{\rm swap} a_q^\dagger f_{\rm swap}=a_p^\dagger$.
\item $e^{if_{\rm swap}\theta}a_p^\dagger e^{-if_{\rm swap} \theta}= \frac{1}{2}\left(e^{-2i\theta}[a_p^\dagger -a_q^\dagger] + [a_p^\dagger +a_q^\dagger] \right)$ and $e^{if_{\rm swap}\theta}c_q^\dagger e^{-if_{\rm swap} \theta}= \frac{1}{2}\left(e^{-2i\theta}[a_q^\dagger -a_p^\dagger] + [a_p^\dagger +a_q^\dagger] \right)$ for any $\theta \in \mathbb{R}$.
\end{enumerate}
\end{lemma}
\begin{proof}
For property (1) we have that
\begin{align}
[a_p^\dagger,f_{\rm swap}]&=[a_p^\dagger,a_q^\dagger a_p -a_p^\dagger a_p]=a_p^\dagger a_q^\dagger a_p -a_q^\dagger a_pa_p^\dagger +a_p^\dagger a_p a_p^\dagger =-a_q^\dagger+a_p^\dagger.
\end{align}
The operation $f_{\rm swap}$ is symmetric under exchange of labels of $p$ and $q$; therefore,
\begin{equation}
[a_q^\dagger,f_{\rm swap}]=a_q^\dagger -a_p^\dagger.
\end{equation}
Property (2) can be shown in two steps.  First $f_{\rm swap}$ is manifestly Hermitian.  To show it is unitary we demonstrate that it maps a complete orthonormal basis of unit vectors to another complete orthonormal basis of unit vectors.
First, the fermionic swap operator has a trivial action on the vacuum which is easy to see from its definition,
\begin{equation}
f_{\rm swap}\ket{0}= \ket{0}.
\end{equation}
It is then easy to show from the above relation and commutation relations of $f_{\rm swap}$ that
\begin{align}
f_{\rm swap}a_p^\dagger \ket{0} = a_p^{\dagger}f_{\rm swap}\ket{0}-[a_p^\dagger,f_{\rm swap}]\ket{0} = a_q^\dagger \ket{0}.
\end{align}
It then follows from symmetry arguments that $f_{\rm swap} a_q^\dagger \ket{0}=a_p^\dagger \ket{0}$.  The final case follows from
\begin{align}
f_{\rm swap}a_p^\dagger a_q^\dagger \ket{0} &= a_p^\dagger f_{\rm swap}a_q^\dagger\ket{0}+[a_p^\dagger,f_{\rm swap}]a_q^\dagger\ket{0}=a_p^\dagger a_q^\dagger \ket{0}.
\end{align}
The result then follow from noting that all four of these states are orthonormal and unit vectors. Since the subspace is four-dimensional this demonstrates the claim.

Property (3) follows from properties (1) and (2) and the fact that $a_p^2 =0 =a_q^2$
\begin{align}
f_{\rm swap}a_p^\dagger f_{\rm swap} &= a_p^\dagger +f_{\rm swap}[a_p^\dagger,f_{\rm swap}]=a_p^\dagger +f_{\rm swap}(a_p^\dagger -a_q^\dagger) \\
&=2a_p^\dagger -a_q^\dagger +a_q^\dagger a_p a_p^\dagger -a_p^\dagger a_p a_p^\dagger -a_q^\dagger a_q a_p^\dagger -a_p^\dagger a_q a_q^\dagger + a_p^\dagger a_p a_q^\dagger +a_q^\dagger a_q a_q^\dagger =a_q^\dagger.\nonumber
\end{align}
Again because $f_{\rm swap}$ is invariant under exchange of labels $p$ and $q$, property (3) also holds when $p$ and $q$ are exchanges.

Finally, property (4) follows directly from Hadamard's lemma and the previous properties.  Specifically note that
\begin{equation}
[f_{\rm swap},[f_{\rm swap},a_p^\dagger]]= 2(a_p^\dagger - a_q^\dagger),
\end{equation}
by nesting this $k$ times we see that
\begin{equation}
{\rm ad}_{f_{\rm swap}}^k a_p^\dagger = (-1)^k2^{k-1}(a_p^\dagger -a_q^\dagger),
\end{equation}
where ${\rm ad}_{f_{\rm swap}}$ is adjoint endomorphism (meaning the nested commutator operator).  It then follows that
\begin{align}
e^{if_{\rm swap} \theta} a_p^\dagger e^{-if_{\rm swap} \theta}&= a_p^\dagger +i\theta [f_{\rm swap},a_p^\dagger]-\frac{\theta^2}{2!}[f_{\rm swap},[f_{\rm swap},a_p^\dagger]]\\
&=a_p^\dagger -i\theta (a_p^\dagger -a_q^\dagger) -\frac{\theta^2}{2!}(a_p^\dagger -a_q^\dagger) +\cdots,\nonumber\\
&=\frac{1}{2}\left(e^{-2i\theta}[a_p^\dagger - a_q^\dagger] + [a_p^\dagger +a_q^\dagger] \right).\nonumber
\end{align}
The analogous claim for $a_q^\dagger$ follows again by symmetry.
\end{proof}

Next given this result, we need to examine the two-level fermionic Fourier transform.  This is important because it is the primitive upon which the FFFT is built.  The circuit in \fig{linear_fqft} illustrates how the eight-mode FFFT leverages, $F_0^\dagger$ and the related operators $F_k^\dagger = e^{-i2\pi k/M a_q^\dagger a_q}F_0^{\dagger}$, to perform a fermionic Fourier transform \cite{Verstraete2009}. The following corollary illustrates that $F_0^\dagger$ performs the necessary two-mode transformation and then the subsequent theorem will use this fact to demonstrate the general construction for the FFFT for more than eight modes and for representations other than the Jordan-Wigner transform.
\begin{figure}[h]
\[\Qcircuit @R 1em @C .75em {
&& \qw & \multigate{1}{F_0} & \qw & \multigate{1}{F_0} & \qw & \qw & \qw & \multigate{1}{F_0} & \qw & \qw & \qw & \qw \\
&& \multigate{1}{f_{\rm swap}} & \ghost{F_0} & \multigate{1}{f_{\rm swap}} &  \ghost{F_0} & \multigate{1}{f_{\rm swap}} & \qw & \multigate{1}{f_{\rm swap}} & \ghost{F_0} & \multigate{1}{f_{\rm swap}} & \qw & \qw & \qw \\
&& \ghost{f_{\rm swap}} & \multigate{1}{F_0} & \ghost{f_{\rm swap}} & \multigate{1}{F_2} & \ghost{f_{\rm swap}} & \multigate{1}{f_{\rm swap}} & \ghost{f_{\rm swap}} & \multigate{1}{F_1} & \ghost{f_{\rm swap}} & \multigate{1}{f_{\rm swap}} & \qw & \qw \\
&& \qw & \ghost{F_0} & \qw & \ghost{F_2} & \multigate{1}{f_{\rm swap}} & \ghost{f_{\rm swap}} & \multigate{1}{f_{\rm swap}} & \ghost{F_1} & \multigate{1}{f_{\rm swap}} & \ghost{f_{\rm swap}} & \multigate{1}{f_{\rm swap}} & \qw \\
&& \qw & \multigate{1}{F_0} & \qw & \multigate{1}{F_0} & \ghost{f_{\rm swap}} & \multigate{1}{f_{\rm swap}} & \ghost{f_{\rm swap}} & \multigate{1}{F_2} & \ghost{f_{\rm swap}} & \multigate{1}{f_{\rm swap}} & \ghost{f_{\rm swap}} & \qw \\
&& \multigate{1}{f_{\rm swap}} & \ghost{F_0} & \multigate{1}{f_{\rm swap}} &  \ghost{F_0} & \multigate{1}{f_{\rm swap}} & \ghost{f_{\rm swap}} & \multigate{1}{f_{\rm swap}} & \ghost{F_2} & \multigate{1}{f_{\rm swap}} & \ghost{f_{\rm swap}} & \qw & \qw \\
&& \ghost{f_{\rm swap}} & \multigate{1}{F_0} & \ghost{f_{\rm swap}} & \multigate{1}{F_2} & \ghost{f_{\rm swap}} & \qw & \ghost{f_{\rm swap}} & \multigate{1}{F_3} & \ghost{f_{\rm swap}} & \qw & \qw & \qw \\
&& \qw & \ghost{F_0} & \qw & \ghost{F_2} & \qw & \qw & \qw & \ghost{F_3} & \qw & \qw & \qw & \qw 
}\]
\caption{\label{fig:linear_fqft} Circuit to implement one-dimensional FFFT on $M = 8$ sites, as described in \cite{Verstraete2009}. Circuit is composed of $f_{\textrm{swap}}$ gates and $F_k$ gates, defined in \eq{fermion_gates} and \eq{fswap}, respectively. The circuit size is ${\cal O}(M^2\log(M))$ and its depth is ${\cal O}(M \log M)$.}
\end{figure}
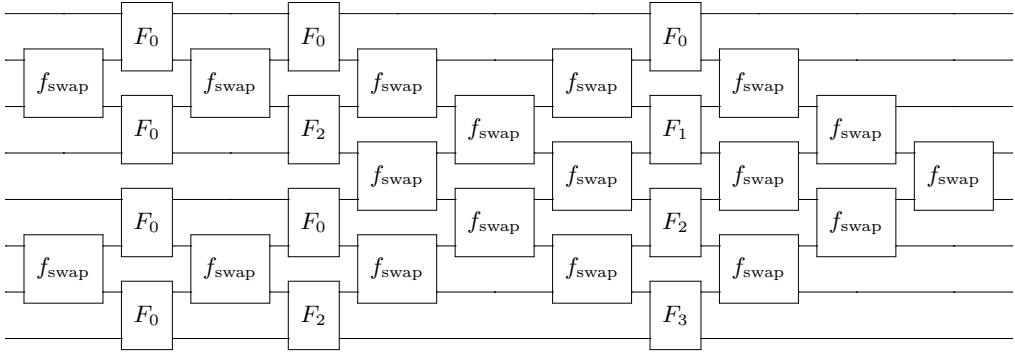

\begin{corollary}
Let $a^\dagger_p$ and $a^\dagger_q$ be creation operators acting on disjoint spin orbitals and let $F_0^\dagger$ be defined as per~\eq{f2def} then $F_0^\dagger a_p^\dagger F_0 = (a_p+a_q)/\sqrt{2}$ and $F_0^\dagger a_q^\dagger F_0 = (a_p-a_q)/\sqrt{2}$.\label{cor:2dfour}
\end{corollary}
\begin{proof}
From~\lem{prop} we have that
\begin{equation}
e^{if_{\rm swap}\pi/4}a_p^\dagger e^{-if_{\rm swap}\pi/4}=\frac{1}{\sqrt{2}}\left(a_p^\dagger e^{-i\pi/4}+a_q^\dagger e^{i\pi/4} \right),
\end{equation}
and
\begin{equation}
e^{if_{\rm swap}\pi/4}a_q^\dagger e^{-if_{\rm swap}\pi/4}=\frac{1}{\sqrt{2}}\left(a_q^\dagger e^{-i\pi/4}+a_p^\dagger e^{i\pi/4} \right),
\end{equation}
Although the magnitudes of the creation operator matches what is needed by the two-dimensional FFFT, the phases are not correct.  The phases for the transformation of $a_p^\dagger$ can be corrected by introducing two phase shift operators:
\begin{equation}
e^{i\pi/4 a_p^\dagger a_p}e^{-i\pi/4 a_q^\dagger a_q} e^{if_{\rm swap}\pi/4}a_p^\dagger e^{-if_{\rm swap}\pi/4} e^{-i\pi/4 a_p^\dagger a_p}e^{i\pi/4 a_q^\dagger a_q} = \frac{1}{\sqrt{2}}\left(a_{p}^\dagger + a_{q}^\dagger \right).
\end{equation}
However, if we apply the same transformation to $a_q^\dagger$ then we find

\begin{equation}
e^{i\pi/4 a_p^\dagger a_p}e^{-i\pi/4 a_q^\dagger a_q} e^{if_{\rm swap}\pi/4}a_q^\dagger e^{-if_{\rm swap}\pi/4} e^{-i\pi/4 a_p^\dagger a_p}e^{i\pi/4 a_q^\dagger a_q} = \frac{i}{\sqrt{2}}\left(a_{p}^\dagger - a_{q}^\dagger \right).
\end{equation}
This unwanted phase of $i$ can be corrected by applying a $e^{-i(\pi/2)a_q^\dagger a_q}$ gate prior to the application of the partial fermionic swap $e^{if_{\rm swap}\pi/4}$ and gives us the claimed unitary gate.
\end{proof}

\begin{theorem}\label{thm:ffft}
The {\rm FFFT} on $M$ spin orbitals, where $M$ is a positive integer power of two can be implemented using $\widetilde{\mathcal{O}}(M^2)$ quantum gates taken from a library that includes $F_0$ gates on nearest neighbor gates, fermionic swap gates and phase gates.  It also requires requires depth $\widetilde{\mathcal{O}}(M)$.
\end{theorem}
\begin{proof}
Our construction for the FFFT consists of two types of gates.  Specifically, we use $F_0$ gates between two adjacent spin orbitals, $f_{\rm swap}$ gates and finally phase shifting gates $e^{-i n_s \phi}$ where $n_s$ is the number operator acting on an arbitrary spin orbital $s$.  For every two level subsystem in the problem we can represent the corresponding creation operators as a vector.  For example, let $c^\dagger_p = [1,0]^\top$ and $c_q^\dagger = [0,1]^\top$.  Thus, applying $F_0$ on this subspace is equivalent to applying the two-dimensional Fourier transform on the vectors that correspond to the elements.  Similarly the phase shifters can be used to set the phases arbitrarily for the creation operators, which allows us to shift the phases of the corresponding vector components arbitrarily.  Thus, these components allow the Hadamard gate and an arbitrary diagonal unitary to be performed on the corresponding set of vectors.

The FFFT of a vector of length $M=2^k$ for positive integer $k$ requires ${\cal O}(M\log(M))$ operations from our gate library. The result is such that, for the $p^{\rm th}$ computational basis vector that this process maps $e_p \mapsto \frac{1}{\sqrt{M}} \sum_j e^{-2\pi i jp/M}e_j$. The algorithm does this by applying a divide and conquer approach to the Fourier transform wherein the discrete Fourier transform on dimension $M$ is broken up into two Fourier transforms on dimension $M/2$.  The elements of these two Fourier transforms are combined by first applying phases to the components of the vector of the form
\begin{equation}
[1,0]^\top\mapsto \frac{[1,e^{-i2\pi k/M}]^\top}{\sqrt{2}}\qquad\qquad [0,1]^\top\mapsto \frac{[1,-e^{-i2\pi k/M}]^\top}{\sqrt{2}},\label{eq:vectrans}
\end{equation}
on two dimensional subspaces corresponding to different mixtures of even and odd Fourier components.

In order to estimate the gate complexity of the algorithm we first need to convert these two-level transformations into operators on the fermionic modes.  Again encoding $c_p^\dagger$ as $[1,0]^\top$ and $c_q^\dagger$ as $[0,1]^\top$ we have that the equivalent fermionic transformation is carried out by a unitary $F_k$ such that
\begin{equation}
F_k^\dagger a_p^\dagger F_k = \frac{a_p^\dagger+e^{-i2\pi k/M}a_q^\dagger}{\sqrt{2}}
\qquad \qquad
F_k^\dagger a_q^\dagger F_k = \frac{a_p^\dagger-e^{-i2\pi k/M}a_q^\dagger }{\sqrt{2}}.
\end{equation}
Since $e^{i\phi a_q^\dagger a_q} a_q^\dagger e^{-i\phi a_q^\dagger a_q}=e^{i\phi} a_q^\dagger$, the gate $F_k$ can be expressed using~\cor{2dfour} as
\begin{equation}
F_k^\dagger = e^{-i2\pi k/M a_q^\dagger a_q}F_0^{\dagger}.\label{eq:fermion_gates}
\end{equation}
$F_k$ is also unitary, as required, since $F_0$ is unitary from~\lem{prop}.
This requires $\mathcal{O}(1)$ gates from our gate set.

By translating the gate operations between the two sets it is clear that if we were not restricted to two-level $F_k$ gates then the process could be executed in $\mathcal{O}(M\log M)$ gates from this gate library. However, owing to this restriction we have to perform fermion swap gates in order to move each $q$ to be adjacent to its corresponding $p$. To do this, $\mathcal{O}(\log(M))$ such fermionic swaps are required. We choose to implement the sort using parallel bubble sort along the lexicographical ordering of the fermion modes, which on $M$ elements requires ${\cal O}(M^2)$ nearest neighbor fermionic swaps to re-arrange the elements.  Since this process needs to be repeated $\mathcal{O}(\log(M))$ times, the number of fermionic swaps required in the overall algorithm is at most $\mathcal{O}(M^2 \log(M))$.  However, the depth is a factor of $M$ lower than this if parallel bubble sort is employed.
\end{proof}

We can now use the previous result to explain how the three-dimensional FFFT can be performed with low depth. The result follows similar reasoning as the previous theorem but with the complication that the FFFT is not easily expressible as a low depth circuit using nearest-neighbor gates when applied to two out of the three dimensions.  The strategy that we employ to avoid this problem is to reorder the spin-orbitals using fermionic swaps.
\begin{corollary}
The three-dimensional {\rm FFFT} on $N=M^3$ spin orbitals, where $M$ is a positive integer power of two can be implemented using ${\mathcal{O}}(N^2)$ quantum gates taken from a library that includes $F_0$ gates on nearest neighbor gates, fermionic swap gates and phase gates.  It also requires requires depth ${\mathcal{O}}(N)$.
\end{corollary}
\begin{proof}
Let us begin by assuming the following canonical ordering: $n(\nu_x,\nu_y,\nu_z) = \nu_x +\nu_y M+\nu_z M^2$.  The three dimensional {\rm FFFT} by definition is composed of independent FFFTs in the $x$-direction, $y$-direction and $z$-direction.  Let each node correspond to a vertex label of a Hamiltonian path embedded in the lattice. Such a path exists because the number of lattice sites is even since $M$ is even. For fixed $\nu_y$ and $\nu_z$, all the fermionic modes which participate in the Fourier transform are contiguous by the definition of a Hamiltonian path. Therefore, each can be simulated using the result of~\thm{ffft}. There are $M^2$ groups of qubits with fixed $\nu_y$ and $\nu_z$ and $\widetilde{\mathcal{O}}(M^{2})$ gates are required to apply the $x$-Fourier transform with each group. Thus, the entire process requires $\widetilde{\mathcal{O}}(M^4)\subset \mathcal{O}(N^2)$ gates from~\thm{ffft}. Each of the $M^2$ FFFTs are independent and can be parallelized. Therefore, we can perform the $x$ component of the fermionic Fourier transform with depth $\widetilde{\mathcal{O}}(M)\subset {\mathcal{O}}(N)$.

Next, let us consider the $y$-Fourier transform.  We apply this Fourier transform by using fermionic swap operations to transform the basis to one where the effective ordering is now changed to $n(\nu_y,\nu_x,\nu_z)$. We achieve this by again performing a bubble sort along the lexicographical ordering of the fermion modes, using fermionic swap operations for the exchange. Bubble sort on $N$ elements requires, in the worst case scenario $\mathcal{O}(N^2)$ swap operations (the evaluation of $n$ is performed in classical preprocessing and thus does not require any quantum operations). Thus, we can sort the qubits into the ordering $n(\nu_y,\nu_x,\nu_z)$ using $O(N^2)$ fermionic swap gates. These swaps are carried out between adjacent vertices on the Hamiltonian path inscribed in the two-dimensional lattice and thus commute and can be directly simulated using nearest neighbor interactions. By parallelizing swaps in bubble sort we see that depth $O(N)$ can be attained. Once sorted, we can again apply the result of~\thm{ffft} to the resulting $M^2$ $y$-Fourier transforms within groups of qubits for which $\nu_x = \nu_z$. Thus, the $y$-component of the {\rm FFFT} can be performed in $\mathcal{O}(N^2 + M^4\log(M)) = \mathcal{O}(N^2)$ gates and depth $\mathcal{O}(N)$.

The $z$-component of the FFFT can be performed using the exact same protocol as the $y$-component, this time sorting the bits so that the ordering is $n(\nu_z, \nu_y, \nu_x)$ and then (if necessary) using fermionic swaps to sort back to the original ordering of spin-orbitals. Thus, by summing the complexities of the Fourier transforms along each of the three components we obtain the claimed complexities for a nearest neighbor architecture on a planar lattice where $M$ is a positive integer power of two. Although the fermionic swap gate between two lexicographically adjacent fermionic modes is not necessarily a two-local qubit gate, this is the case under the Jordan-Wigner transformation. Thus, we have demonstrated that ${\cal O}(N)$ layers of gates suffice to implement the FFFT on a planar qubit architecture.
\end{proof}

Note that the fermionic swap operation has many other potential uses in quantum simulation. As an example, one application would be in the implementation of operator nesting \cite{Poulin2014}. While this procedure typically requires ancilla to evaluate Jordan-Wigner strings when parallelizing commuting operations, one could perform nesting in-place by using fermionic swap operations to move qubits acted upon by Hamiltonian terms that act on disjoint sets of qubits next to each other in lexicographical ordering.

\section{Alternative Trotter-Suzuki Algorithm}
\label{app:alt_trotter}
While we have examined simulation using the fermionic fast Fourier transform within the plane wave dual basis, it is important to note that this approach is not necessary.  For purposes of comparison, we outline here the method by which one would simulate chemical dynamics within the basis using the Jordan-Wigner representation of the spin orbitals. 
The Hamiltonian is well suited for such simulations because it can be conveniently expressed as a sum of Pauli operators as shown in \eq{qubit_ham}.
The simplest term that appears in such a Trotter decomposition is of the form $e^{-i 2\phi n_p}$.  Such terms are easy to implement.  It is easy to see from \eq{jw_ops} that this is equal to $e^{i\phi Z}$, up to an irrelevant global phase. This is a single qubit rotation which can either be directly implemented in non-fault tolerant architectures or performed using a sequence of ${\cal O}(\log(1/\epsilon))$ gates in a fault-tolerant architecture.

\begin{figure}[h]
\[\Qcircuit @R 1em @C 1em {
&\ctrl{1}	&\ctrl{2}	&\qw								&\ctrl{2}			&\ctrl{1}				&\qw\\
&\targ		&\qw 		&\gate{e^{-i\phi_{12}Z}}			&\qw				&\targ					&\qw\\
&\ctrl{1}	&\targ		&\gate{e^{-i\phi_{13}Z}}			&\targ				&\ctrl{1}				&\qw\\
&\targ 		&\qw		&\gate{e^{-i\phi_{34}Z}}			&\qw				&\targ  					&\qw\\
}\]
\caption{\label{fig:diag}Simulation circuit for $e^{-i(Z_1Z_2\phi_{12} +Z_3Z_4\phi_{34} + Z_{1}Z_3 \phi_{13})}$.  This strategy allows $N-1$ such terms to be simulated in parallel using a CNOT chain of depth $\lceil \log(N)\rceil$.}
\end{figure}
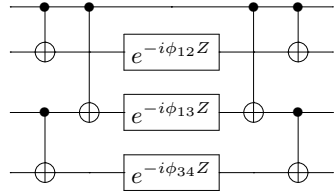

\begin{figure}[h]
\[\Qcircuit @R 1em @C 1em {
&\ctrl{3}	&\qw		&\gate{Z}	&\gate{H}		&\ctrl{3}			&\gate{e^{-i\phi Z}}	&\ctrl{3}&\gate{H}	&\gate{Z}	&\qw	&\ctrl{3}	&\qw\\
&\qw		&\ctrl{1}	&\qw		&\qw			&\qw				&\qw					&\qw&\qw		&\qw		&\ctrl{1}	&\qw		&\qw\\
&\qw		&\targ		&\ctrl{-2}	&\qw			&\qw				&\qw					&\qw&\qw		&\ctrl{-2}	&\targ		&\qw		&\qw\\
&\targ 		&\qw		&\qw		&\qw			&\targ				&\gate{e^{i\phi Z}}		&\targ&\qw		&\qw		&\qw		&\targ		&\qw\\
}\]
\caption{Simulation circuit for $e^{-i2\phi (a_p^\dagger a_q +a_q^\dagger a_p)}$ for use within the Trotter-Suzuki framework illustrated for $q=p+3$.  The analogous networks traditionally used contain $12$ CNOT, $8$ single qubit Clifford operations and $2$ single qubit rotations and the rotations cannot be parallelized.\label{fig:kineticcirc}}
\end{figure}
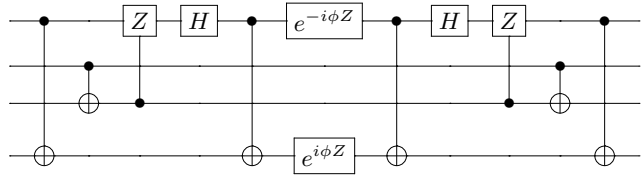

The next simplest such terms are of the form $e^{4i\phi_{pq} n_p n_q}$.  Such terms are slightly more sophisticated and good networks are known for these exponentials as given in~\cite{Whitfield2010}.  While such terms are seldom dominant for second-quantized quantum simulation, for molecules represented in the plane wave dual basis they are among the most numerous terms.  Therefore, it warrants taking some time to devise optimal networks for these circuits.  First, while the approach of \cite{Whitfield2010} groups all three non-identity terms in the expansion of \eq{jw_ops} for $n_p n_q$ into a single circuit, this is not necessarily optimal.  This is because the single qubit terms can be grouped together. Instead, by decomposing the Hamiltonian as per \eq{qubit_ham} directly into Pauli operators we can execute the single qubit terms that come from both the $n_p$ and $n_p n_q$ terms simultaneously.  This allows them to be executed with ${\cal O}(N)$ gates and depth ${\cal O}(1)$.

The $Z_pZ_q$ term is slightly more challenging.  The strategy that we employ, as seen in \fig{diag}, is to break up the sum into sets of $N-1$ terms all of which can be computed by CNOTs acting on disjoint qubits in a logarithmic number of layers.  The simplest such group is $$\{Z_1Z_2,Z_3Z_4,\ldots, Z_{N-1}Z_N,Z_1Z_3,\ldots, Z_{N-2}Z_N,\ldots, Z_1Z_{N-1}\}.$$ There are ${\cal O}(N)$ such sets and so we can perform all $N(N-1)/2$ exponentials using at most $N(N-1)/2$ rotations, ${\cal O}(N)$ of which need to be executed sequentially.  This is a factor of $3$ reduction from the networks of~\cite{Whitfield2010} and in addition this approach requires no ancilla to be parallelized. Next let us focus on the kinetic term.  We employ a new strategy for simulating the kinetic term that is based on ideas from~\cite{Reiher2017}.  The circuit works by diagonalizing the Hamiltonian $a^\dagger_p a_q + a^\dagger_q a_p$ by transforming qubits $p$ and $q$ into the Bell basis.  This is done because $X_pX_q$ and $Y_pY_q$ are simultaneously diagonal in that basis, which can easily be seen from \eq{jw_ops} and the fact that
\begin{equation}
X\otimes X = \begin{pmatrix} 0 & 0&0&1\\0&0&1&0\\0&1&0&0\\1&0&0&0 \end{pmatrix},\qquad Y\otimes Y = \begin{pmatrix} 0 & 0&0&-1\\0&0&1&0\\0&1&0&0\\-1&0&0&0 \end{pmatrix}.
\end{equation}
More specifically, if we let $\ket{B_{ij}}=(CNOT_{1,2})(H\otimes 1)(X^i\otimes X^j)\ket{00}$ then the values of $i$ and $j$ uniquely give the eigenvalues for the $X$ and $Y$ terms in the Hamiltonian.  

The circuit in \fig{kineticcirc} shows such a transformation.  The outer controlled NOT gates in the circuit (as well as a Hadamard that is absorbed into the controlled-$Z$) give the basis change into the Bell basis.  The next controlled NOT computes the Jordan-Wigner string and the controlled $Z$ copies the value onto the qubit that performs the $X$ part of the rotation.  The remaining qubit flips the sign of the $Y$ part of the rotation as needed.  In general, these networks require $2(N+2)$ gates for $N$ spin-orbitals.  Also, when these terms are ordered lexicographically the majority of the Jordan-Wigner strings between adjacent Trotter steps will cancel as discussed in ~\cite{Wecker2015a}.

Note that our work provides a further optimization that was not appreciated in~\cite{Wecker2015a}. The presence of Jordan-Wigner strings requires the introduction of ancillary qubits to parallelize the rotations that appear in the simulation. Similar depth reductions can be achieved by using fermionic swap operations to move each relevant pair of spin orbitals adjacent to each other within the lexicographic ordering implicit in the Jordan-Wigner representation. There are $\mathcal{O}(N^2)$ such terms, however, the proof of~\lem{gates} shows that we can perform a fermionic-swap network in depth $\mathcal{O}(N)$ that will allow us to simulate every hopping term.  Additionally the construction requires no ancillary space, but requires more Clifford gate operations to perform the fermionic swap (which in this case can be performed with 3 CNOT gates and a CZ gate) than would be needed in the nesting approach of~\cite{Wecker2015a}.  Thus, each step in this alternative Trotter-Suzuki approach can also be simulated in linear depth.

The main place where the two approaches differ is in the bounds that fall out of the Trotter-Suzuki decomposition.  Following the same reasoning as was used to find \eq{trotscale}, we obtain that the gate depth needed for simulation is
\begin{equation}
\widetilde{\mathcal{O}}\left(Nr\right) \in {\cal O}\left(\frac{N^{5/2}t^{3/2}}{\Omega\sqrt{\epsilon}}\sqrt{1 +\left(\frac{N^{-1/3}\eta^2}{\Omega^{-1/3}} \right)^2 } \right).
\end{equation}
Note that this bound is likely less tight than the bound that was found for the Fourier-based approach because more terms are present in the Hamiltonian, which necessitates more liberal use of the triangle inequality and also creates more terms that do not commute with each other in the expansion.  For this reason, if we constrain ourselves to simulations with constant electron density then we obtain a worst case scaling of ${\cal O}(N^{9/2})$ scaling.

\section{Alternative Taylor Series Algorithm}
\label{app:onthefly}

In this section, we explain an alternative way to perform the Taylor series algorithm. In particular, we implement the circuit $\textsc{prepare}(W)$ in a different and more complex fashion than in \sec{taylor_alg}. While the asymptotic gate complexity of the two approaches are almost the same (perhaps due to loose bounds), the method described here has significantly lower depth. Whereas \sec{taylor_alg} implemented $\textsc{prepare}(W)$ in a similar fashion to ``database'' algorithm of \cite{BabbushSparse1}, in this section we implement $\textsc{prepare}(W)$ in a similar fashion to the ``on-the-fly'' algorithm of \cite{BabbushSparse1}.

Our approach will be to compute the coefficients of the Hamiltonian ``on-the-fly'' and apply them as phases in order to execute $\textsc{prepare}(W)$ as specified in \eq{prepw}. To accomplish this we will think of each term in the sum over $\nu$ as an individual term in the Hamiltonian and then compress the sum. That is,
\begin{equation}
\label{eq:Wpqbk}
W_{p,q,b} = \sum_{\nu \neq 0} W_{p,q,b,\nu} \quad \quad \quad
W_{p,q,b,\nu} = \begin{cases}
\frac{\pi}{2\, \Omega \, k_\nu^2} - \frac{k_\nu^2}{8 \, N} + \frac{\pi}{\Omega} \sum_{j}\zeta_j \frac{\cos\left[ k_\nu \cdot \left(R_j-r_p\right)\right]}{k_\nu^2}  & p = q \\
\frac{\pi \, \cos \left[k_\nu \cdot \left(r_p - r_q\right)\right]} {4 \, \Omega\,k_\nu^2} & b = 0 \wedge p \neq q \\
\frac{k_\nu^2 \cos \left[k_\nu \cdot \left(r_p - r_q \right) \right]}{4 \, N} & b = 1 \wedge (p + q) \, \!\!\mod\!\! \, 2 = 0\\
\frac{1}{2 \, N} & b = 1 \wedge (p + q) \, \!\!\mod\!\! \, 2 = 1.\\
\end{cases}
\end{equation}
While we can efficiently apply phases to quantum states by controlling on the entire state, one cannot efficiently change the amplitude of a quantum state by controlling on the entire state. Thus, we must take the additional step of further subdividing each term with one more index so that each $W_{p,q,b,\nu}$ is a sum of $\mu$ phases with the same magnitude,
\begin{equation}
W_{p,q,b,\nu} \approx \zeta \sum_{m=0}^{\mu-1} W_{p,q,b,\nu,m} \quad \quad \quad
W_{p,q,b,\nu,m} \in \left\{-1, +1\right\}
\quad \quad \quad
\zeta \in \Theta\left(\frac{\epsilon}{L \, t}\right) \quad \quad \quad
\mu \in \Theta\left(\frac{\textrm{max}_{p,q,b,\nu} \left | W_{p,q,b,\nu} \right |}{\zeta}\right).
\end{equation}
To accomplish this one-the-fly, we perform logic on the output of $\textsc{sample}(W)$ which acts as,
\begin{equation}
\textsc{sample}(W) \ket{p}\ket{q} \ket{b} \ket{\nu} \ket{0}^{\otimes \log \mu}
\mapsto \ket{p}\ket{q}\ket{b} \ket{\nu} \ket{\widetilde{W}_{p,q,b,\nu}}
\end{equation}
where $\widetilde{W}_{p,q,b,\nu}$ is a digital approximation with $\log \mu$ bits to the real-valued $W_{p,q,b,\nu}$. Since the values of $W_{p,q,b,\nu}$ shown in \eq{Wpqbk} are straightforward arithmetic functions of $p$, $q$, $b$ and $\nu$, together with simple logic, we see that $\textsc{sample}(W)$ can be implemented at gate complexity $\widetilde{\cal O}(1)$ with respect to $N$ and $\epsilon$. Note that some of this arithmetic (such as reversible computation of the reciprocal) can be costly to compute to high precision in practice. Furthermore, if we were concerned about scaling with number of nuclear charges, we could also break up the $Z_p$ coefficients in terms of the nuclei $j$ using a number of ancilla scaling logarithmically in the number of nuclei. 

Given the $\textsc{sample}(W)$ circuit, we can construct the $\textsc{prepare}(W)$ circuit by performing logic followed by phase-kickback on the output of the $\textsc{sample}(W)$ register. The values of $W_{p,q,b,\nu,m}$ are always either +1 or -1 but we actually need the square root of these values for the $\textsc{prepare}(W)$ superposition (see \eq{prepw}). Thus, we need the circuit
\begin{equation}
\textsc{kickback}(W) \ket{m} \ket{\widetilde{W}_{p,q,b,\nu}}
\mapsto \begin{cases}
 \ket{m} \ket{\widetilde{W}_{p,q,b,\nu}} &  \widetilde{W}_{p,q,b,\nu} > \left(2m - \mu\right) \zeta \\
i  \ket{m} \ket{\widetilde{W}_{p,q,b,\nu}} &  \widetilde{W}_{p,q,b,\nu}  \leq \left(2m - \mu\right) \zeta
\end{cases}.
\end{equation}
$\textsc{kickback}(W)$ can also be implemented with gate complexity $\widetilde{\cal O}(1)$ with respect to $N$ and $\epsilon$. We put these circuit together with some Hadamard gates to form the complete $\textsc{prepare}(W)$ circuit as shown in \fig{preparew}. We see that
\begin{equation}
H = \zeta \sum_{p,q,b,\nu,m} W_{p,q,b,\nu,m} H_{p,q,b}.
\end{equation}

\begin{figure}[h]
\vspace{-.5cm}
\[\Qcircuit @R 1em @C .75em {
&&\lstick{\ket{0}_p^{\otimes \log N}} & \gate{H^{\otimes \log N}} & \multimeasure{3}{} & \qw & \multimeasure{3}{} & \qw \\
&&\lstick{\ket{0}_q^{\otimes \log N}} & \gate{H^{\otimes \log N}} & \ghost{} & \qw & \ghost{} & \qw \\
&&\lstick{\ket{0}_b} & \gate{H} &\ghost{} & \qw & \ghost{} & \qw \\
&&\lstick{\ket{0}_{\nu}^{\otimes \log N}} & \gate{H^{\otimes \log N}} &\ghost{} & \qw & \ghost{} & \qw \\
&&\lstick{\ket{0}^{\otimes \log\mu}} & \qw & \gate{\textsc{sample}(W)} \qwx & \multigate{1}{\textsc{kickback}(W)} &\gate{\textsc{sample}(W)} \qwx & \qw \\
&&\lstick{\ket{0}_m^{\otimes \log \mu}} & \gate{H^{\otimes \log \mu}} & \qw & \ghost{\textsc{kickback}(W)} & \qw & \qw
}\]
\caption{\label{fig:preparew} The $\textsc{prepare}(W)$ circuit. Note that the $H$ gate is a Hadamard. $\textsc{sample}(W)$ is called twice to uncompute the ancilla register. As there are only ${\cal O}(\log N)$ ancilla and the gate complexities of both $\textsc{sample}(W)$ and $\textsc{kickback}(W)$ are bounded by $\widetilde{\cal O}(1)$, the overall gate complexity required to implement $\textsc{prepare}(W)$ is $\widetilde{\cal O}(1)$ with respect to $\epsilon$ and $N$.}
\end{figure}
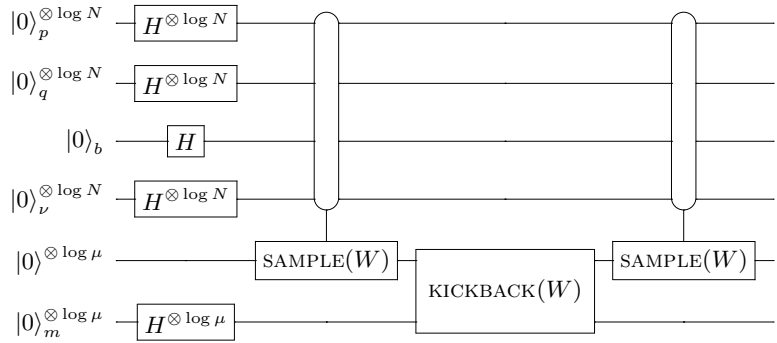

While our implementation of $\textsc{prepare}(W)$ is significantly more efficient than the method outline in \sec{taylor_alg}, by breaking up the Hamiltonian into these different terms the normalization $\Lambda$ becomes,
\begin{align}
\label{eq:fly_lambda}
\Lambda & \in \widetilde{\cal O}\left( \zeta \sum_{p,q,b,\nu,m} \left | W_{p,q,b,\nu,m} \right | \right) =  \widetilde{\cal O}\left(\sum_{p,q,b,\nu} \left | \max_{p,q,b,\nu} \left |W_{p,q,b,\nu} \right | \right | \right)\\
& = \widetilde{\cal O}\left(N^3 \max_{p,q,\nu} \left[ \left |\frac{k_\nu^2 \cos \left[k_\nu \cdot \left(r_p - r_q\right)\right]}{N}  \right |, \left| \frac{\cos \left[k_\nu \cdot \left(r_p - r_q\right)\right]}{\Omega \, k_\nu^2} \right |  \right]\right)\nonumber\\
& = \widetilde{\cal O}\left( N^2 k_{\textrm{max}}^2 + \frac{N^3}{\Omega \, k_{\textrm{min}}^2} \right)
= \widetilde{\cal O}\left( \frac{N^{8/3}}{\Omega^{2 /3}} + \frac{N^3}{ \Omega^{1/3}} \right) \nonumber
\end{align}
which is significantly higher than the value of $\Lambda$ which applies to the method of \sec{taylor_alg}. Since the gate complexity of implementing $\textsc{select}(H)$ is $\widetilde{\cal O}(N)$ and the gate complexity of implementing $\textsc{prepare}(W)$ is $\widetilde{\cal O}(1)$, from \eq{result} we find that the total gate complexity of our Taylor series approach is no more than
\begin{equation}
\widetilde{\cal O}\left( \frac{N^{11/3}}{\Omega^{2 /3}} + \frac{N^4}{\Omega^{1/3}}\right)
\end{equation}
with only polylogarithimic dependence on precision. We see that the gate complexity at fixed density becomes $\widetilde{\cal O}(N^{11/3})$, which is better than the $\widetilde{\cal O}(N^4)$ scaling of the method in \sec{taylor_alg}. Furthermore, the oracle for $\textsc{select}(H)$ can be parallelized to $\widetilde{\cal O}(1)$ depth using arbitrary two-qubits gates. This can be taken advantage of by our ``on-the-fly'' algorithm but not by our database algorithm due to the difference in scaling of $\textsc{prepare}(W)$. 

To see this consider the following.  The $\textsc{select}(H)$ oracle consists of five cases depending on the values of $p$, $q$ and $b$.  These cases can be executed sequentially without sacrificing more than a constant factor in depth. The cases corresponding to the kinetic energy terms are the only ones that require $\widetilde{\mathcal{O}}(N)$ sized circuits.  However, they can be performed in depth $\widetilde{\mathcal{O}}(1)$ using the following protocol.  First, fanout a qubit string that replicates $N$ copies of $p$, $q$ and $b$.  This can be achieved in ${\cal O}(\log N)$ depth.  Next for each qubit compute the value of the control bit that decides whether the conditions for that term to be activated are met.  This requires $\mathcal{O}(\log N)$ operations.  Next compute for qubit $j$ whether $j=q$, $j=p$ or $j\in (p,q)$ and using Toffoli gates conditioned on these qubits as well as the flag that determines whether the term is activated to begin with, apply $X$, $Y$ or $Z$ on the qubit in question as dictated by $\textsc{select}(H)$.  By construction, the depth needed for this process is $\widetilde{\mathcal{O}}(1)$. After this has been performed, uncompute all ancillae, which can be done in $\widetilde{\mathcal{O}}(1)$ depth.  The entire process requires then clearly requires $\widetilde{\mathcal{O}}(1)$ depth. Since $r=\widetilde{\mathcal{O}}(N^{8/3})$ segments are required for the simulation from \eq{fly_lambda} and each segment can be performed in $\widetilde{\mathcal{O}}(1)$ depth, we find that the overall gate depth of our algorithm is $\widetilde{\cal O}(N^{8/3})$. This depth is substantially lower than any previously described algorithm for electronic structure simulation in the literature.

\end{document}